%% file: main_classical_ldpc_arXiv.tex
\documentclass[aps,prx,twocolumn,superscriptaddress,noeprint,longbibliography, nofootinbib]{revtex4-2}

\input{preamble}

\usepackage[framemethod=tikz]{mdframed}
\usepackage{mathrsfs}

\newcommand{\taumix}{\ensuremath{\tau_{\rm mix}}} % mixing time

\newcommand{\codespace}{\ensuremath{\mathcal C}}
\newcommand{\codeword}{\ensuremath{{\vec x}^{(\codespace)}}}
\newcommand{\codestate}{\ensuremath{{\vec \sigma}^{(\codespace)}}}

\newcommand{\tanner}[2]{\ensuremath{\mathcal T\left(#1, #2\right)}}

\newcommand{\Edyn}{\ensuremath{\epsilon_{\rm mem}}}
\newcommand{\Eglass}{\ensuremath{\epsilon_{\rm G}}}

\newcommand{\Tdyn}{\ensuremath{T_{\rm mem}}}
\newcommand{\Tglass}{\ensuremath{T_{\rm G}}}

\newcommand{\pGibbs}{\ensuremath{p_{\rm G}}}
\newcommand{\pGibbsTr}[1]{\ensuremath{p_{\rm G}^{(#1)}}}

\newcommand{\gdeg}{\ensuremath{s}}
\newcommand{\depth}{\ensuremath{R}}
\makeatletter
\def\l@subsubsection#1#2{}
\makeatother

\begin{document}

\title{
Expansion creates spin-glass order in finite-connectivity models: \\a rigorous and intuitive approach from the theory of LDPC codes
}

\author{Benedikt Placke}
\email{benedikt.placke@physics.ox.ac.uk}
\affiliation{Rudolf Peierls Centre for Theoretical Physics, University of Oxford, Oxford OX1 3PU, United Kingdom}
\author{Grace M. Sommers}
\affiliation{Department of Physics, Princeton University, Princeton, NJ 08544, USA}
\author{Nikolas P.\ Breuckmann}
\affiliation{School of Mathematics, University of Bristol, Bristol BS8 1UG, United Kingdom}
\author{Tibor Rakovszky}
\affiliation{Department of Physics, Stanford University, Stanford, California 94305, USA}
\affiliation{Department of Theoretical Physics, Institute of Physics,
Budapest University of Technology and Economics, M\H{u}egyetem rkp. 3., H-1111 Budapest, Hungary}
\affiliation{HUN-REN-BME Quantum Error Correcting Codes and Non-equilibrium Phases Research Group,
Budapest University of Technology and Economics,
M\H{u}egyetem rkp. 3., H-1111 Budapest, Hungary}
\author{Vedika Khemani}
\email{vkhemani@stanford.edu}
\affiliation{Department of Physics, Stanford University, Stanford, California 94305, USA}

\date{\today}

\begin{abstract}
Complex free-energy landscapes with many local minima separated by large barriers are believed to underlie glassy behavior across diverse physical systems. This is the heuristic picture associated with replica symmetry breaking (RSB) in spin glasses, but RSB has only been rigorously verified for certain mean-field models with all-to-all connectivity. 
In this work, we give a rigorous proof of finite temperature spin glass order --- defined in terms of a complex free-energy landscape --- for a family of models with local interactions on finite-connectivity, non-Euclidean expander graphs.
To this end, we bypass the RSB formalism entirely, and instead exploit the mathematical equivalence of such models to certain low-density parity check (LDPC) codes.
We use \emph{code expansion}, a property of LDPC codes which guarantees extensive energy barriers around ground states. Together with mild additional assumptions, this allows us to construct an \emph{explicit} decomposition of the low-temperature Gibbs state into disjoint \emph{components}, each hosting an asymptotically long-lived state associated with a local minimum of the landscape. Each component carries at most an exponentially small fraction of the total weight (``shattering''), and almost all components do not contain ground states --  which we take together to define spin-glass order.
The proof is elementary, and treats various expanding graph topologies on the same footing, 
including those with short loops where existing approaches such as the cavity method fail.
Our results apply rigorously to certain diluted $p$-spin glasses for sufficiently large (but finite) $p$, and, while unproven, we also expect our assumptions to hold in a broader family of codes. 
Motivated by this, we numerically study two simple models, on random regular graphs and a regular tesselation of hyperbolic space, respectively. 
We provide evidence that both models undergo \emph{two} transitions as a function of temperature, corresponding to the onset of weak ergodicity breaking and spin glass order, respectively. 
\end{abstract}

\maketitle

\clearpage
\newpage

\section{Introduction}

Complex free-energy landscapes provide a unifying framework for understanding glassy behavior across diverse physical systems. From spin glasses in disordered magnets \cite{anderson1978concept,binder1986review,bramwell2001review} to structural glasses in supercooled liquids \cite{debenedetti_stillinger_supercooled2001, tarjus2005frustration, kivelson2008search,bertier2011rmp}, the emergence of rugged, hierarchical landscapes populated by many local and global minima offers potential explanations for the dramatically slow dynamics, aging, and memory effects that define glassy phenomenology. Beyond physics, the mathematical machinery developed around landscape-based theories has found applications in computer science and information theory, spanning neural networks, constraint satisfaction, error correction, and machine learning~\cite{hopfield1982network,amit1985sg_nn,stein2013spin,Krzakala2024machine_learning,mezard2009information,monasson1999complexity_transitions,achlioptas2008algorithmic,mezard2002survey_propagation,krzakala2007gibbs,ricci_tersenghi2010xorsat,franz2002dynamic,di2004weight_enumerators}. 

This intuitive picture of a complex landscape can be formalized in terms of the properties of the \emph{Gibbs state}, in cases where the local minima are separated by \emph{macroscopic} free-energy barriers. The existence of such barriers permits a \emph{decomposition} of the Gibbs state into  distinct `extremal components'~\cite{mezard2009information}, each of which defines a long-lived (meta)stable state associated with a local minimum.  The task of \emph{rigorously} defining such a Gibbs state decomposition and characterizing its complexity is highly non-trivial~\cite{friedli2017statistical}, as we elaborate in the next section. Fundamental questions therefore remain:   How can landscape complexity be quantified and proved?  When and why do complex landscapes emerge? 
What are the minimal ingredients required for such complexity?

The most sophisticated approach for addressing these questions is the formalism of replica symmetry breaking (RSB), developed in the context of spin-glasses \cite{mezard1987spin}. RSB quantifies landscape complexity by postulating an abstract hierarchy of `replicas', and an order parameter that takes the form of an $n\times n$ matrix capturing `overlaps' between the replicas, evaluated in a formal limit $n \to 0$. In principle, each replica corresponds to one of the extremal equilibrium Gibbs states discussed above\footnote{These are called `pure states' in the terminology of Ref.~\onlinecite{parisi1983order}.}~\cite{mezard2009information, parisi1983order, franchini2023rsbwr}, but the correspondence is indirect  and subtle~\cite{huse_fisher1987incongruent}, and the underlying physical mechanisms responsible for landscape complexity are obscured by heavy mathematics. 
This machinery yields exact results for certain `mean-field' fully-connected (`all-to-all') models [\autoref{fig:intro_graphs}(a)], most famously the Sherrington–Kirkpatrick model~\cite{sherrington1975solvable, parisi1979solution1,parisi1979solution2}, many aspects of which have since been rigorously verified \cite{guerra2003rsb,Talagrand2011a,Talagrand2011b,panchenko2013sherrington}.  While more direct proofs, not involving replicas, also exist, they also remain limited to similar all-to-all models~\cite{alaoui2023shattering,gamarnik2023shattering,arous2024shattering}. 

Extending these rigorous results to more realistic `sparse' models with finite local connectivity has proven significantly more challenging \cite{dembo_montanari2010review, Panchenko2014structure1,Panchenko2016structure2}. This is true not only of models in finite Euclidean dimensions [(\autoref{fig:intro_graphs}(b)] (where the applicability of the RSB landscape picture remains hotly debated~\cite{mcmillan1984scaling,fisher_huse1986, huse_fisher1987droplet,moore2021droplet, newman_stein1992,newman_stein1997,baity_jesi2021temperature_chaos, read2014short_range,read2022complexity}), but also in the case of various non-Euclidean graphs that sit between the fully-connected and finite-dimensional Euclidean extremes [(\autoref{fig:intro_graphs}(c)]. As such, these latter models are as close to the mean-field limit as possible, while retaining a finite local connectivity, and indeed, various methods suggest that they can exhibit RSB and spin glass order~\cite{montanari2001glassy,franz2001ferromagnet,franz2002dynamic,krzakala2007gibbs}. 

\begin{figure}
    \centering
    \includegraphics{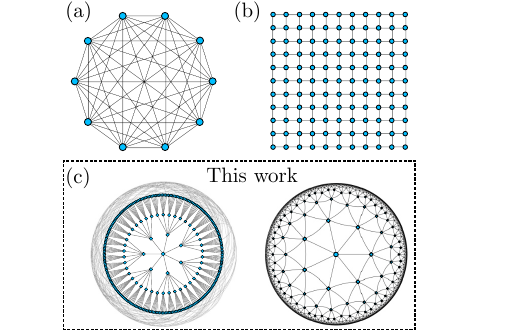}
    \caption{Different interaction graphs: (a) Fully connected, `all-to-all', such as in the Sherrington-Kirpatrick model;  (b,c) Sparse and local with finite connectivity, in (b) low-dimensional Euclidean space, such as in the Edwards Anderson model, and on (c) non-Euclidean expander graphs. Expander graphs can be locally tree-like (left) or have small loops (right), but both cases are closed at the boundaries so every vertex is (statistically) identical to every other.}
    \label{fig:intro_graphs}
\end{figure}

\textbf{Our work presents a transparent---and fully rigorous---proof of spin-glass order, and a complex free energy landscape, in families of sparse, finite-connectivity systems.}
We focus on models defined on non-Euclidean \emph{expander graphs}. [\autoref{fig:intro_graphs}]. These have the defining property that the bulk-to-boundary ratio is finite: for \emph{any} subset of vertices smaller than some finite fraction of the graph, the number of neighbours outside the set grows in direct proportion to the size of the subset itself. These constitute a large family of graphs, including both random constructions---e.g. finite-degree random regular graphs (RRGs)---and deterministic ones---e.g. Cayley graphs of certain groups. While some expander graphs (including RRGs) are \emph{locally tree-like}  \footnote{A  finite tree with an open boundary is not an expander graph because subsets of sites at the boundaries have neighbor-sets whose size grow too slowly. A graph is called \emph{regular} if every vertex has the same degree (number of neighbors). } [\autoref{fig:intro_graphs}(c), left], but closed with no boundaries, others include many short loops, such as regular tesselations of the hyperbolic plane [\autoref{fig:intro_graphs}(c), right]. 
Our proof technique treats all these cases on an equal footing, but we note that only certain random constructions are known to satisfy all the assumptions of our proof.

A canonical example of the kind of sparse model we consider is the \emph{diluted} $p$-spin model  in which interactions involve $p>2$ spins and are strictly local on an RRG. This model has been extensively studied in the spin-glass literature using the `cavity method'~\cite{Mezard1986cavity,mezard1987spin,mezard2015cavity,montanari2001glassy,franz2001ferromagnet,franz2002dynamic,krzakala2007gibbs}, a heuristic, but powerful extension of mean‑field reasoning that is most reliable on loop‑free tree graphs\footnote{In its simplest version, one deletes a spin, treats the neighbors as independent cavities, and enforces self‑consistency, yielding recursive Bethe‑Peierls (belief‑propagation) equations. These are exact on a tree because deleting a spin disconnects the tree into independent sub-trees, which can be solved recursively.} \cite{mezard2001bethe,mezard2006reconstruction}, but becomes increasingly uncontrolled as loops proliferate. 
On closed graphs that are only \emph{locally} tree‑like, it has been rigorously proven that the cavity solution is asymptotically exact as long as correlations decay before long loops close; this defines the so-called \emph{replica symmetric} regime, corresponding to a lack of spin-glass order \cite{DemboMontanari2010_Ising,DemboMontanariSun2013_Factor,DemboMontanariSlySun2014_Potts}.
On lowering the temperature (or increasing frustration), correlations survive around loops, destabilizing the replica symmetric solution.  Heuristic `one-step' RSB cavity calculations then predict that the diluted $p$-spin model undergoes a transition to a spin-glass phase. 
Numerics support this picture \cite{franz2001ferromagnet}, and it has been established that the cavity solution always provides a variational lower bound for the free energy \cite{franz_leone2003replica,panchenko_talagrand2004bounds}.
However, rigorously establishing the correctness of the RSB cavity solution in the low-temperature spin glass phase on any closed sparse graph remains, to the best of our knowledge, an open problem \cite{dembo_montanari2010review, Panchenko2014structure1,Panchenko2016structure2}.

\begin{figure*}
\includegraphics[]{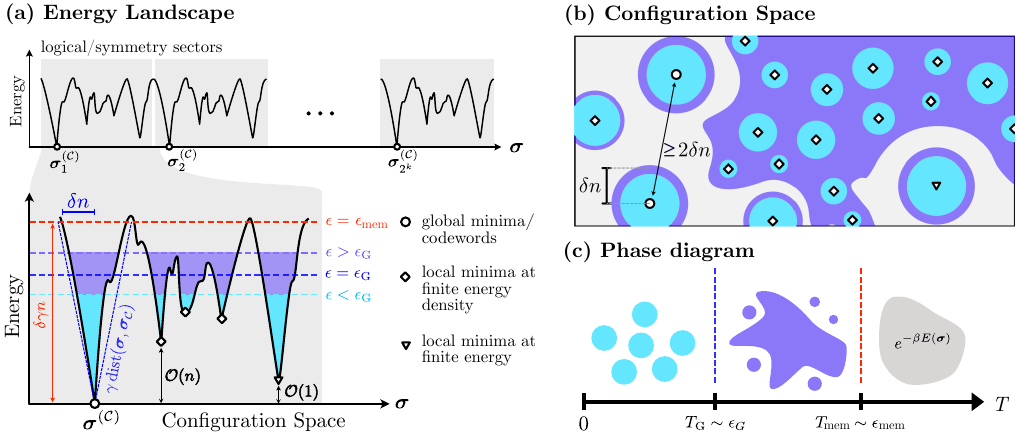}
\caption{
The energy landscape of expander LDPC codes, and resulting phase diagram. (a) Expander codes have extensive energy barriers surrounding ground states (`codewords') $\codestate$, denoted by circles (\autoref{sec:Expanders}). This also implies extensive energy barriers around states at finite energy (triangles) and even finite energy density (squares) (\autoref{sec:Glassiness}). Codewords are related by exact symmetries of the Hamiltonian. 
(b) The energy barriers lead to clustering of states in configuration space: Below a cutoff energy density $\epsilon<\Eglass$, configuration space shatters into disjoint clusters surrounded by extensive barriers (\autoref{sec:Glassiness}), each of which carries only an exponentially small fraction (in the number of variables $n$) of the total weight (light blue-shaded regions).
At higher energies, $\Eglass < \epsilon < \Edyn$, there is a regime of weak clustering: there are still exponentially many clusters but one cluster carries almost all configurations (purple-shaded regions).
(c) The two characteristic energy scales $\Eglass$ and $\Edyn$ lead to two phase transitions in the structure of the Gibbs state $\pG\sim \exp(-\beta E)$ as a function of temperature. The two low-temperature phases correspond to spin glass order, and weak ergodicity breaking, respectively (\autoref{sec:spin_glass_order}).}
\label{fig:landscape_intro}
\end{figure*}

Our work bypasses the cavity method altogether and utilizes an entirely different approach. Diluted $p$-spin-models on RRGs—with \emph{uniform} ferromagnetic couplings—are mathematically equivalent to a family of classical error correcting codes~\cite{mezard2009information}, namely Gallager’s low density parity check (LDPC) codes~\cite{gallager1960thesis,gallager1962low}, and to the linear XORSAT constraint satisfaction problem~\cite{ricci_tersenghi2010xorsat}. 
Cavity calculations have already shown that this ferromagnetic system can exhibit glassy behavior, even in the absence of the disorder and frustration in couplings which is traditionally associated with glassiness~\cite{montanari2001glassy,franz2001ferromagnet,franz2002dynamic,krzakala2007gibbs}. We build on results from coding theory to rigorously establish the existence of spin-glass order for such ferromagnetic diluted $p$-spin-models at sufficiently large (but finite) values of $p$.

From the perspective of error correcting codes, Gallager codes belong to a broader family of \emph{expander LDPC codes} \cite{sipser_spielman1996}, which can be defined on diverse families of expander graphs, including loopy non-random geometries like the hyperbolic tessellations mentioned above. Our framework imports tools from coding theory to analyze these systems, treating this wider class of models on the same footing and opening new territories for rigorous spin-glass analysis in regimes where traditional methods fail.

The main contributions of this work may be summarized as follows: 

\begin{itemize}
    \item  \textbf{We provide a rigorous proof of spin-glass order and a complex energy landscape in certain sparse finite-connectivity models on closed graphs at low temperatures.}
   
    \item  \textbf{Our framework, inspired by coding theory, is orthogonal to existing replica and cavity approaches. In particular, it is also well suited to study loop-rich graphs where cavity methods become uncontrolled}. 

   \item \textbf{Our derivation is transparent, and reveals new insights on physical mechanisms for landscape complexity.} 
    
\end{itemize}

We now provide a more detailed overview of our technical approach and results. 

We consider Hamiltonians of $n$ spins with sparse ferromagnetic interactions, associated with LDPC codes defined on families of expander graphs whose size grows with $n$.
We furnish an \emph{explicit} decomposition of the low-temperature Gibbs state into disjoint components separated by \emph{extensive} free-energy barriers of size $\propto n$. We then define and prove spin-glass order through properties of this decomposition, as quantified by various configurational entropies. We do this by proving a closely related result on the complexity of the \emph{energy} landscape, and use this to derive properties about the \emph{free-energy} landscape. 
 
While our results imply 1-RSB (and do not rule out higher levels of RSB) in the sense of a nontrivial overlap distribution, we bypass the question of exactness of the corresponding cavity solution entirely.
Instead, our approach directly analyzes the structure of the Gibbs state itself. The analysis is transparent, requiring only basic linear algebra with binary variables, and reveals the underlying geometric mechanisms for landscape complexity. 

Our proof uses \emph{code expansion} --- a property defined in coding theory which implies that ground states (`codewords') are surrounded by \emph{extensive} energy barriers: the cost of flipping spins starting from a ground state is lower bounded by a linearly growing function, until an extensive fraction of spins have been flipped\footnote{Graph expansion and code expansion are distinct but closely related concepts. In particular, in order for an LDPC code to be expanding in this sense, the underlying interaction graph must also be an expander.} [\autoref{fig:landscape_intro}(a)]; the rate of this linear growth is given by the so-called \emph{expansion coefficient}. Codes with a finite expansion coefficient are called expander codes, and we explain how such barriers naturally arise via the interplay of local constraints and expanding geometry. 

Next, we show that due to the linearity of the underlying constraint satisfaction problem, expansion also implies that all states below an energy density cutoff are surrounded by extensive energy barriers.  
We then impose two additional assumptions: (i) a strong enough expansion coefficient and (ii) lack of redundancies in the interactions\footnote{A redundancy is a subset of interaction terms that multiply to the identity. For the Ising model with nearest-neighbor $\sigma_i\sigma_j$ interactions, the loops of the graph define the redundancies. However, for more general interactions, the redundancy structure need not be set by loops in the graph.}. Under these assumptions, we prove that the configuration space below a cutoff energy density, $\epsilon < \epsilon_G$, \emph{shatters} into exponentially many (in $n$) disjoint `clusters', each carrying an exponentially small fraction of the total weight. Each cluster is separated from the others by extensive Hamming distance and extensive energy barriers [\autoref{fig:landscape_intro}]. Thus, these clusters are local minima of the low energy landscape. In addition, we prove that the clustering is \emph{incongruent}\footnote{The choice of this term is inspired by, but not identical to, the terminology in \cite{huse_fisher1987incongruent}, as we will explain later}, meaning that most clusters do \emph{not} contain ground states and are \emph{not} related to each other by symmetries of the Hamiltonian. This is crucial because  our models are `good' LDPC codes with \emph{exponentially} many ground states belonging to different symmetry sectors. Thus, good codes may be shattered purely from symmetry considerations. Incongruence ensures shattering occurs even \emph{within} individual symmetry sectors. 

The properties of shattering and incongruence (defined precisely later) together define a complex energy landscape. At low enough temperatures,  extensive energy barriers translate to extensive \emph{free-energy barriers}. This, in turn, enables a decomposition of the Gibbs state into exponentially many components associated with the local energy minima, each hosting a long-lived equilibrium state. We define and prove spin-glass order in terms of the complexity of this decomposition. 
The flow of our argument, with pointers to the corresponding sections of this work, is sketched in Figure~\autoref{fig:proof_flowchart}. 

\begin{figure}
    \centering
    \includegraphics{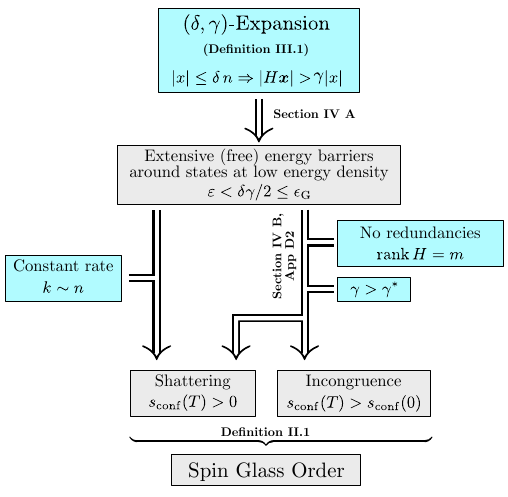}
    \caption{Flowchart sketching the sequence of implications from expansion to spin glass order. Sections discussing the reasoning behind certain implication are indicated at the respective arrow. Taken together, the right half of the flowchart encompasses the proof of \cref{thm:sg_expanders}.}
    \label{fig:proof_flowchart}
\end{figure}
 
Our results apply rigorously to diluted $p$-spin glass on RRGs for sufficiently large (but finite) $p$, which instantiate all our assumptions. However, while unproven, we also expect our assumptions to hold more broadly, including on loopy graphs. The main technical challenge here lies in the fact that the available bounds on the code expansion of these graphs are too weak for our proof to apply\footnote{Understanding of redundancies is also more limited in this case, but we expect this requirement to be less stringent, for reasons we discuss below.} However, we expect (supported by numerical evidence) that the actual value of the expansion coefficient is much larger than what is suggested by these existing bounds. As such, improvements from coding theory would directly lead to an extension of our proof of spin-glass order to these models as well. 

We support our heuristic picture with numerics on two simple models, one on a locally-tree like graph and one on a loopy hyperbolic tesselation [\autoref{fig:intro_graphs}(c)], both with degree $s=7$. The assumptions on strong code expansion that we need to instantiate our theorems cannot be rigorously proved for these models. A range of numerical experiments using Metropolis dynamics simulations on closed, finite graphs nevertheless show close agreement with our analytic pictures. For the locally tree-like model, we also compare the simulation results with those of the 1-RSB cavity method (realized via a recursive population dynamics calculation on trees), finding close agreement. We also show numerically that the models in fact exhibit three distinct thermodynamic regimes as a function of temperature, each corresponding to a qualitatively distinct structure of the Gibbs state.
This is sketched in \autoref{fig:landscape_intro}(c): In between the low temperature spin glass phase and the high-temperature paramagnetic phase lies a regime of \emph{weak ergodicity breaking}: there are (exponentially) many Gibbs state components, but a single component carries almost all the weight.

Tangentially, we note that having new approaches to study spin glasses on loopy graphs is also essential for extending the analysis of spin-glass theory to \emph{quantum} LDPC codes, which are a subject of much recent interest in quantum computation, and where short-loops are necessarily present in the interaction graph. Indeed, we have recently tackled the quantum problem using these new methods, and uncovered a new type of \emph{topological} quantum spin-glass \cite{placke2024tqsg}.

The rest of this paper is organized as follows. We begin, in \autoref{sec:spin_glass_order} by defining Gibbs state decompositions as well as spin-glass order in terms of the properties of these decompositions. Then, in \autoref{sec:Expanders}, we review classical LDPC codes, code expansion, and discuss some examples of their construction. In \autoref{sec:Glassiness}, we give an informal presentation of our main results, and sketch their derivation. Formal definitions and proofs are relegated to \appref{app:complexity}, which is self-contained.
Finally, we present our numerical results in \autoref{sec:numerics}.
We note that \appref{app:TreeIsing} contains a detailed study of the phase diagram of the Ising model on a tree from a memory perspective, which might be of independent interest.

\tableofcontents

\section{Spin Glass Order\label{sec:spin_glass_order}}

In this section, we provide a brief but self-contained review of Spin Glass Order and its distinction from other nontrivial thermodynamic phases, in particular those with (conventional) spontaneous symmetry breaking. 
To this end, we first define what we mean by a ``nontrivial thermodynamic phase'' in terms of a decomposition of the Gibbs state into distinct components, following Ref.~\onlinecite{mezard2009information}. We then discuss how to distinguish different nontrivial phases, including different kinds of spin glass order. We conclude this section by sketching a connection of our definitions to replica symmetry breaking. 

\subsection{Non-uniqueness of Gibbs states\label{sec:gibbs_decomposition}}

We say that system shows non‑trivial order at temperature $T$ if there is more than one Gibbs state at that temperature. The two dimensional Ising model, for example, has \emph{two} Gibbs states 
in the ferromagnetic phase below its critical temperature $T_c$, one with positive and one with negative magnetization; above $T_c$, the Gibbs state is unique. 

There are different ways of formalizing the non-uniqueness of Gibbs states~\cite{georgii2011gibbs,friedli2017statistical}. A standard `Peierls' approach is to consider the distribution of Gibbs states obtained in the bulk of a system with open boundaries, upon taking the  limit of increasing system size with different fixed boundary conditions\footnote{This is closely related to the cavity method, and we will utilize a version of this approach in \autoref{sec:recursive} to compare against our results on locally tree-like graphs.} (for example the `all-up' and `all-down' boundary conditions in the Ising model). 
However, relating models with open and closed boundaries is non-trivial, especially for graphs that go beyond finite-dimensional Euclidean lattices, such as the expander graphs that are at the focus of the present paper. 
We will instead use the approach outlined in Chapter 22 of Ref. \onlinecite{mezard2009information} (see also Ref. \onlinecite{dembo_montanari2010review}) which is applicable to a sequence of finite systems of increasing size without relying on boundary conditions, making it well suited for  the examples we discuss in this paper. 
We now outline how this approach works.

We consider a system of $n$ classical spins,  living on a configuration space $\chi$. The Gibbs distribution assigns probability 
\begin{equation}\label{eq:def:gibbs}
    \pGibbs(\vec\sigma) \propto e^{-\beta \hamil(\vec\sigma)}
\end{equation} to the configuration $\vec\sigma \in \chi$, where $\beta$ is the inverse temperature and $\hamil$ is the system's Hamiltonian .  
In a finite closed system, the Gibbs state is always unique. However, in a non-trivial phase, the distribution decomposes into multiple \emph{components}, corresponding to different subsets of $\chi$, which look macroscopically different and are separated by macroscopic free-energy barriers, in a manner we make precise below.   In the infinite system size limit, each component can then be thought of as a distinct (approximate) Gibbs state in its own right, with diverging lifetime. 
 Again, the two-dimensional Ising model serves to illustrate the idea: below $T_c$, the Gibbs distribution  concentrates around two subsets of configurations with extensively different magnetizations; these 
 contain all the probability weight except for a small correction that vanishes (exponentially) in the limit $n\to\infty$.

\begin{figure}
    \centering
    \includegraphics[width=0.9\columnwidth]{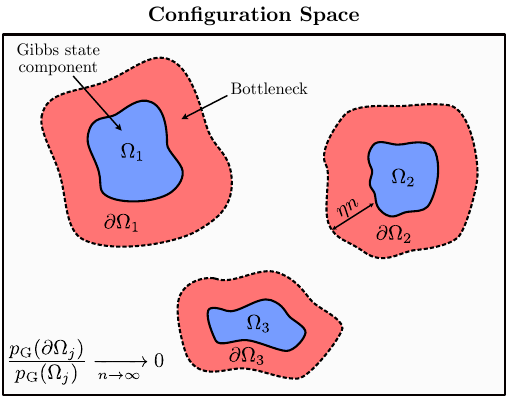}
    \caption{
    Sketch of a system with multiple Gibbs states. Configuration space splits into regions $\Omega_j$ each of which has weight $<1$ in the Gibbs state $\pG\propto e^{-\beta \hamil}$ and each of  which is surrounded by boundaries $\partial\Omega_j$ with negligible relative probability [see \autoref{eq:bottleneck_def}]. These boundaries act as bottlenecks for detailed-balance satisfying dynamics, meaning that each $\Omega_j$ hosts a state whose lifetime diverges wit increasing $n$.
    \label{fig:Gibbs_states}}
\end{figure}

More concretely, consider a subset $\Omega\subset \chi$ of the configuration space. We say that $\Omega$ hosts a Gibbs state component if it is surrounded by a ``bottleneck'' of low (relative) probability configurations. In particular we require that set of states within distance $\eta n$ of $\Omega$ has a probability that is vanishingly small compared to that of $\Omega$ itself: 
\begin{subequations}\label{eq:bottleneck_def}
\begin{align}
\frac{\pGibbs(\partial_\eta \Omega)}{\pGibbs(\Omega)} \leq \Delta(n) \quad \text{ where } \Delta(n) \xrightarrow[n \to \infty]{} 0.
\end{align}
Here $\pGibbs(\Omega) \equiv \sum_{\vec \sigma \in \Omega} \pGibbs(\vec \sigma)$ and 
we defined the boundary of $\Omega$ as\footnote{Here, we here take the width of the bottleneck region ($\eta n$) to be extensive. This is appropriate for the models we consider, but it could be easily relaxed to a weaker dependence on $n$ as long as the width diverges as $n\to\infty$.}
\begin{equation}
  \partial_\eta\Omega \equiv \{\vec\sigma \notin\Omega; 0 < \dist(\vec \sigma,\Omega) \leq \eta n \}.
\end{equation}
\end{subequations}
The distance $\dist(\vec\sigma,\vec\sigma')$ is defined by the Hamming distance, which is the number of spins on which two configurations $\vec\sigma$ and $\vec\sigma'$ differ; likewise, the distance between two subsets of configurations is the smallest distance between their elements, which formalizes the notion of `macroscopically different'. 
Given such a subset $\Omega$ of configurations, we define a Gibbs state component $\pGibbsTr{\Omega}$ as
\begin{equation}
    \pGibbsTr{\Omega}(\vec \sigma) \propto \begin{cases}
         \pGibbs(\vec \sigma) & \text{if} ~ \vec\sigma \in \Omega \\
        0 & \text{otherwise}
    \end{cases},
    \label{eq:truncated_gibbs}
\end{equation}
with a suitable normalization. Further, such an approximate Gibbs state $\pGibbsTr{\Omega}$ is termed \emph{extremal} if it cannot be further decomposed on multiple subregions of $\Omega$ while fulfilling the bottleneck condition \autoref{eq:bottleneck_def}. We consider a system to be nontrivial if it allows for multiple Gibbs state components corresponding to distinct, non-overlapping regions $\Omega$ (\autoref{fig:Gibbs_states}). We review in \appref{app:decomposition_Ising} how precisely to apply this definition to the ferromagnetic Ising model in two dimensions. 

Another way of looking at this definition is that the bottleneck condition \autoref{eq:bottleneck_def} formalizes the notion of a \emph{free-energy barrier} separating $\Omega$ from the rest of the configuration space: the states forming the bottleneck region $\partial_\eta\Omega$ need to have sufficiently high energy compared to the states in $\Omega$, while also need to be sufficiently small in number, such that the bottleneck ratio in~\autoref{eq:bottleneck_def} goes to zero in the limit $n\to\infty$. Formally, we can define the free energy of a subset via $e^{-\beta F(\Omega)} \equiv \sum_{\vec \sigma \in \Omega} e^{-\beta \hamil(\vec \sigma)}$, so that $\Delta(n) = e^{-\beta (F(\partial_\eta \Omega)- F(\Omega))}\rightarrow 0$. Thus,  states initially within $\Omega$ have to overcome this macroscopic free-energy barrier to escape, which implies a diverging lifetime for the Gibbs state component $\pGibbs^{(\Omega)}$.

Indeed, an appealing aspect of this approach is that nontrivial thermodynamics (in the sense of the existence multiple Gibbs state components) immediately implies the existence of nontrivial \emph{dynamical} properties. In particular, the so-called Markov chain bottleneck theorem~\cite{levin2017markov}, implies that the Gibbs state components as defined above are approximate steady states of any sufficiently local Markov process that has the global Gibbs distribution $\pGibbs$~\cite{mezard2009information} as a steady state (e.g. by virtue of obeying detailed balance). In particular, if we initialize the system in the state $\pG^{(\Omega)}$, it will remain close to this initial state up to times that are lower-bounded by $\Delta(n^{-1})$ in~\autoref{eq:bottleneck_def}. 

The existence of multiple Gibbs state components therefore implies that the system serves as a \emph{passive memory} under any local (Monte Carlo) dynamics: different initial states remain distinguishable up to a time scale that diverges as $n\to \infty$. For example, in the $D$-dimensional Ising model, the sign of the global magnetization could serve as a readout for the stored information and is expected to survive for an exponentially long time in (linear) system size (see~\appref{app:decomposition_Ising}). The situation for the ising model on a tree is more complicated, as we discuss in~\appref{app:TreeIsing} but one still ends up with a long-lived memory at sufficiently low temperatures. Flipping this argument on its head, we expect passive memories to generically be associated with nontrivial thermodynamics. Strikingly, it may also be possible to obtain this non-trivial decomposition and passive memory even as the (global) partition function remains analytic at all temperatures \cite{eggarter1974cayley, montanari2006bethe, rakovszky2023physics, hong2024quantum}.

\subsection{Spin Glass Order}

We defined a non-trivial phase by the Gibbs state having multiple, macroscopically distinct components. An obvious next question is in regard to the number and nature of these components. In the two-dimensional Ising model, there are just two components, related to each other by symmetry, exemplifying the idea of spontaneous symmetry breaking\footnote{In $d\geq 3$, and at low enough temperatures, there are in fact other states satisfying the definition in \autoref{eq:bottleneck_def}, e.g. those with a pair of far-separated flat domain walls stretching through the entire system. However, these have large energy and thus give a negligible contribution to the overall Gibbs state; see our discussion below in this section.\label{fn:Dobrushin}} More generally, however, the number of components can grow rapidly (super-polynomially, or even exponentially) with $n$ and they need not be related to each other by any exact symmetries of the classical Hamiltonian $\hamil(\vec\sigma)$. The spin glass models we study exemplify both of these features. 

Going beyond merely counting the number of components, we should also ask how much each of them contributes to the global Gibbs state $\pG$ [\autoref{eq:def:gibbs}]. To define what we mean by this, let us decompose the configuration space into a collection of disjoint subsets $\chi = \biguplus_j \Omega_j \uplus \Lambda$, where the $\{\Omega_j\}$ satisfy \autoref{eq:bottleneck_def} (we might also take them to be extremal, for good measure) while $\Lambda$ is a ``junk" region whose total weight $w_\Lambda = \pGibbs(\Lambda)$ vanishes in the limit $n\to\infty$ (it could contain, for example, contain all the boundaries $\partial_\eta \Omega_j$). We can then write
\begin{subequations}\label{eq:gibbs_decomposition_weights}
\begin{equation}
    \pGibbs(\vec \sigma)
    = \sum_j w_j \,\pGibbsTr{\Omega_j}(\vec \sigma) + w_\Lambda\, \pGibbsTr{\Lambda}(\vec \sigma),
\end{equation}
where 
\begin{equation}
    w_j = \pGibbs(\Omega_j) \equiv \sum_{\vec \sigma \in \Omega_j} \pGibbs(\vec\sigma).
\end{equation}
\end{subequations}
The distribution $\{w_i\}$ is normalized, and gives a more detailed characterization of the phase beyond the number of Gibbs states itself. 

We can conveniently characterize this distribution through its Shannon entropy\footnote{By contrast, the number of Gibbs state components corresponds to the \emph{Hartley entropy}.}, which is called the \emph{configurational entropy} (also referred to as the complexity): 
\begin{equation}\label{eq:sconfig_def}
    S_{\rm conf} \equiv n\,s_{\rm conf} = -\sum_i w_i \log w_i.
\end{equation}
The configurational entropy quantifies the number of distinct components that contribute relevantly to the naive Gibbs state\footnote{For example, in the case of the Ising model in dimensions $d \geq 3$ mentioned in Footnote \ref{fn:Dobrushin}, while the number of Gibbs states increases polynomially with $n$, due to the different possible placement of the flat domain walls, $S_{\rm conf} = \ln{2}$ is independent of $n$, since only the two homogenous uniformly magnetized states give non-vanishing contributions to $\pGibbs$.}.

$S_{\rm conf}$ is a function of both the temperature $T = 1/\beta$ and the system size $n$. If we take the limit $T\to 0$ at a fixed $n$, the configurational entropy simply counts the (logarithm of) the number of exact minimum energy configurations (ground states). In the cases we are interested in, such exact ground state degeneracies stem from symmetries of $\hamil(\vec\sigma)$ and thus $S_{\rm conf}(T=0)$ gives a measure of the degeneracy arising from symmetry-breaking\footnote{More generally, one might have to be more careful about taking limits. For example, in a quantum version of these models, the ground state degeneracy would cease to be exact; in this case, rather than simply taking $T=0$, one would probably want to consider a small, $n$-dependent temperature that correctly captures the size of the approximately degenerate ground state manifold.}. The difference $S_{\rm conf}(T) - S_{\rm conf}(T=0)$, on the other hand, measures the fraction of Gibbs state components that do not contain any of the ground states and are not accounted for by symmetry breaking.

Another interesting characterization of the decomposition is the weight of the largest component, which corresponds to the \emph{configurational min-entropy}:
\begin{equation}\label{eq:sconfmin_def}
    \Sconf^{(\rm min)} \equiv  n\,\sconf^{(\rm min)} = -\log {\rm max}(\{w_i\}).
\end{equation}
and which lower bounds the configurational (Shannon) entropy. If the configurational min-entropy density is finite, this means that no single component carries more than an exponentially small fraction in $n$ of the weight.

Given this, we can now give a definition of spin glass order fit to our purposes:

\begin{definition}[informal]\label{def:SG_def}

     A classical spin Hamiltonian is said to have \emph{spin glass order} at temperature $T$ if it exhibits
     \begin{itemize}
         \item \textbf{Shattering}: no single Gibbs state component contributes more than an exponentially small fraction  to $\pG$:
         \begin{equation}\label{eq:SG_def1}
	         \sconf^{(\rm min)}(T) > 0.
         \end{equation}
         
         \item \textbf{Incongruence}: the number of Gibbs state components with relevant contributions is exponentially larger than the number of ground states
        \begin{equation}
            s_{\rm conf}(T) > s_{\rm conf}(T=0) \geq 0,
            \label{eq:SG_def2}
        \end{equation}
     \end{itemize}
\end{definition} 

We note that our choice of the word ``incongruence" is inspired by (but not identical to) the usage in ~\cite{huse_fisher1987incongruent}, where an incongruent decomposition is one in which different Gibbs states are not even \emph{locally} related by the symmetries of the Hamiltonian.  

This definition corresponds to a particularly strong form of spin glass order. While, as we will see, this is appropriate to the models that we discuss below, it could be broadened significantly to include a much larger class of models that may reasonably be called spin glasses. For example, we could require  that $\Sconf(T)$ and $\Sconf^{\rm (min)}(T)$ only grow as a weaker power $n^\alpha$ with some power $1 > \alpha > 0$. This would include, for example, the Sherrington-Kirkpatrick model which does not satisfy the stricter definition in \autoref{eq:SG_def1}. Similarly, \autoref{eq:SG_def2} could be weakened to the requirement that $S_{\rm conf}(T) - S_{\rm conf}(T=0)$ grows with a weaker power $n^\alpha$, which is sufficient to guarantee that most Gibbs states with relevant contributions at $T$ do not look like any of the ground states. 
However, the stronger condition we use has a simple interpretation. As mentioned, in the models studied below, the extensive ground state degeneracy ($s_{\rm conf}(T=0) > 0$) originates from a large symmetry group (discussed in \autoref{sec:Expanders}). One could then divide the configuration space into ``sectors'' related to each other by symmetry. The condition \autoref{eq:SG_def2}  means that there are exponentially many relevant Gibbs states (and hence, shattering) even within a single such sector. We also note that this definition of spin-glass order may not be relevant to finite-dimensional Euclidean spin-glasses, where the applicability of the landscape picture remains unclear. For example, the competing droplet theory of Huse and Fisher~\cite{huse_fisher1987droplet} predicts only two symmetry-related Gibbs state components whose structure changes chaotically with temperature \cite{bray_moore1987temperature_chaos,fisher_huse1988temperature_chaos}.

\subsection{Weak vs. strong ergodicity breaking\label{sec:weak_vs_strong}}

The discussion of the previous section also leads to another distinction. It might be the case that even when the Gibbs state is not unique so that multiple components fulfilling \autoref{eq:bottleneck_def} exist, only one of them contributes significantly to $\pG$. In this case, $S_{\rm conf} \to 0$ in the limit $n \to \infty$,  so that the naive Gibbs state $\pG$ itself becomes extremal in the thermodynamic limit, within the convex set of all Gibbs states.

Dynamically, in this situation, typical initial conditions (sampled from $\pG$ itself) relax back to the naive Gibbs state. Instead, it is only certain fine-tuned initial states that get stuck in one of the other Gibbs state components. For this reason, we will also refer to this scenario as \emph{weak ergodicity breaking}. This is in contrast to the strong ergodicity breaking that occurs when $S_{\rm conf} > 0$ and multiple components contribute to $\pG$. In the strong ergodicity breaking case, typical initial conditions will relax to one of these components, rather than $\pG$ itself. 

In the systems we study, we will show that there is strong ergodicity breaking at the lowest temperatures, $T < \Tglass$, satisfying \cref{def:SG_def}. However, we find that this low-temperature phase is separated from the trivial high-temperature phase by an intermediate phase with weak ergodicity breaking; see \autoref{fig:landscape_intro}. We denote $\Tdyn$ as the transition temperature between the intermediate and trivial phases, because the nontrivial  Gibbs state components, which first appear at $\Tdyn$, can serve as passive memories. 

A similar phenomenology is known to occur for the Ising model on the Bethe lattice \cite{eggarter1974cayley,muller1974new,matsuda1974infinite,chayes1986mean,bleher1995purity,ioffe1996extremality,mezard2001bethe,mezard2006reconstruction,magan2013memory}, reviewed in ~\appref{app:TreeIsing}. In that case, the (naive) partition function is analytic and trivially factorizes at all temperatures~\cite{eggarter1974cayley} (a feature shared with the expander codes studied in this work), so that equilibrium two-point correlation functions decay exponentially at all temperatures. Nevertheless,  the model exhibits three distinct regimes, separated by two phase transitions at temperatures $\Tglass < \Tdyn$. These are characterized by changes in the decomposition of the Gibbs states, and can be diagnosed by applying different boundary conditions on finite trees.  

At high temperatures, $T > \Tdyn$, the Gibbs state in the bulk of the tree is unique and independent of boundary conditions. Below \Tdyn, this uniqueness is lost and ergodicity is weakly broken: only certain fine-tuned (polarized) BC  lead to a change in the bulk state.\footnote{
That is, below \Tdyn, the Bethe-Peierls equations admit nontrivial solutions (accessible via fine-tuned BCs), hence the designation of \Tdyn~as the ``Bethe-Peierls temperature'' in some works~\cite{eggarter1974cayley,muller1974new}.} 
Finally, below \Tglass, we enter a regime of strong ergodicity breaking where even a typical boundary condition---e.g., one obtained by sampling from the naive Gibbs measure at that temperature, then freezing the boundary---leads to a bulk Gibbs state that is different from the naive (free BC) one. These transitions also have interesting dynamical consequences. 
On a finite tree, it is known that the Ising model is rapidly mixing (i.e.  the mixing time scales similarly to the paramagnetic phase) at all temperatures above \Tglass \cite{martinelli2003ising,Martinelli2004}, seemingly insensitive to the weak ergodocity breaking transition at $\Tdyn$.. Nevertheless,  we show in \appref{app:TreeIsing} that the change in Gibbs state structure at $\Tdyn$ does reflect in the relaxation (memory) time of the bulk magnetization.

\subsection{Connections to replica symmetry breaking}

The definition of spin glass order given above may, on first sight, look quite distinct from that in terms of replica symmetry breaking as first proposed by Parisi as the (full) solution to the Sherrington-Kirkpatrick model \cite{sherrington1975solvable,parisi1979solution1,parisi1979solution2}. 
A connection between these approaches has been made in the literature, as we discuss in the following for completeness (see e.g. Ref. \onlinecite{mezard2009information} for a more comprehensive review).  We note, however, that it has also been pointed out that the translation from the overlap distributon of replicas to the Gibbs state decomposition is not always reliable~\cite{huse_fisher1987incongruent}. It is the latter that is of physical interest and,  indeed,  one of the strengths of our approach is that we work directly with the Gibbs state decomposition. 

The central object to characterize replica symmetry breaking is the overlap function,
\begin{equation}\label{eq:overlap}
    Q_{\alpha\gamma} = \frac{1}{n} \sum_{j = 1}^n \expval{\sigma_j}_{\alpha} \expval{\sigma_j}_{\gamma}
\end{equation}
where the subscripts $\alpha$ and $\gamma$ in the expectation values denote that they are taken with respect to different ``replicas''.
In the original approach \cite{sherrington1975solvable}, the replicas are introduced as literal copies of the system, as a computational trick to evaluate the quenched free energy 
\begin{equation}\label{eq:replica_trick}
    -T F = \overline{\log Z} = \lim_{M\to0} \frac{1}{M}(\overline{Z^M} - 1)
\end{equation}
where $\overline{\bullet}$ denotes the disorder average. For the SK model, one can calculate the disorder-averaged partition function of $M$ coupled copies of the system ($\overline{Z^M}$) and only in the final result one lets $M\to 0$ to obtain the result for $\overline{\log Z}$. 
While one naively expects the $M$ copies to behave identically, i.e. $Q_{\alpha\gamma} = Q$, Parisi showed, however, that for some range of parameters, a physical solution is attained by assuming a more complex structure for $Q_{\alpha\gamma}$ (in the regime where the replica symmetric solution has negative entropy and is hence unphysical \cite{sherrington1975solvable}); for the SK model in particular, the overlap takes a \emph{continuous range} of values \cite{parisi1979solution2}.

In the language of Gibbs states, different replicas are (heuristically) associated with the extremal Gibbs state components introduced above~\cite{parisi2002physical, franchini2023rsbwr}.
The subscript of the expectation value in \autoref{eq:overlap} then denotes the probability distribution with respect to which the expectation is evaluated. Recalling how the Gibbs state decomposition is defined, \autoref{eq:gibbs_decomposition_weights}, each replica hence corresponds to a distinct ergodic component of the system, confined to a distinct region of configuration space $\Omega_\alpha$ [\autoref{eq:bottleneck_def}]. 
The overlap distribution could then in principle be sampled from by the following dynamical experiment. Initial configurations $\vec\sigma^{(\alpha)} \in\Omega_{\alpha}$, $\vec\sigma^{(\gamma)} \in\Omega_{\gamma}$ are drawn according to the global Gibbs state $\pG$, and then evolved in time due to some detailed-balance obeying stochastic dynamics to obtain a time trace $\{\vec\sigma^{(\alpha)}(t)\}_{t=0}^{\mathcal T}$. The Ising variables are averaged over time on each site individually, such that overlap in \autoref{eq:overlap} is that of the resulting `fingerprints'
\begin{equation}
   \expval{\sigma_j}_{\alpha}^{\mathcal T} \equiv \sum_{t=0}^{\mathcal T} \sigma^{(\alpha)}_j(t) 
\end{equation}
For finite systems, the bottleneck condition [\autoref{eq:bottleneck_def}] implies that the dynamics is confined to the region $\Omega_{\alpha}$ during the whole measurement if $\mathcal T \ll \Delta(n)^{-1}$. Hence, assuming ergodicity within the set $\Omega_\alpha$, we can choose the measurement time such that $\mathcal T\to\infty$ as $n\to\infty$ and $\expval{\sigma_j}_{\alpha}^{\mathcal T} \to \expval{\sigma_j}_{\alpha}$.

Beyond the distinction of strong and weak ergodicity breaking mentioned before, the exact shape of the overlap distribution $P(Q)$ distinguishes a full wealth of different spin glass orders. The arguably simplest case of RSB, 1-step RSB, is the case where the overlap takes exactly two values, depending on whether the replica indices are identical or not
\begin{equation}\label{eq:onne_rsb}
    Q_{\alpha\gamma} = \delta_{\alpha\gamma} (Q - Q') + Q'
\end{equation}
which has been argued by the cavity method to be the situation realized in models very similar to those we consider below \cite{franz2001ferromagnet,franz2002dynamic,dembo_montanari2010review}.
The case where the overlap takes up to $k$ different values and is referred to as $k$-step RSB, with the case of the continuous distribution as realized in the SK model called \emph{full} replica symmetry breaking. In this latter case, one can further show that any nontrivial distribution must have a very particular structure called \emph{ultrametricity} \cite{parisi2000ultrametricity}.

We reiterate, however, that in this work, we will not rely on the `replica trick' or overlap distributions [\autoref{eq:replica_trick}] to infer conclusions about the structure of distinct ergodic components; instead we \emph{directly} show the existence of a decomposition of the Gibbs state with nontrivial structure.

\section{Classical Expander Codes}\label{sec:Expanders}

In this section, we provide a general introduction to the models which are the subject of this paper, which are based on a class of error correcting codes called \emph{expander codes} \cite{sipser_spielman1996, richardson2008modern,guruswami2019essential}. We describe these codes, and their associated Hamiltonians. 

\subsection{Low Density Parity Check Codes}\label{sec:ldpc}

The goal of error correction is to store classical information, concretely $k$ \emph{logical bits}, in a way that protects it against noise. 
This can be done by means of redundancy, that is we embed the logical bits into a larger number of $n$ \emph{physical bits} by identifying $2^k$ \emph{codewords} (the states of the logical bits) among the $2^n$ possible bit strings of this larger configuration space. For the information to be stored robustly, these codewords should not be `close', that is they must differ in many different positions. The smallest number of bits that have to be flipped between any two codewords is called the \emph{distance} of the code, denoted by $d$.
These three main characteristics of a code are often denoted as a triplet, and we say that a code with $n$ bits, $k$ logical bits and distance $d$ is a $[n, k, d]$ code.
The perhaps simplest example is the \emph{repetition code}, which has two codewords ($k=1$) defined as the all-zeros and all-ones bitstrings, which have distance $d = n$. 
The repetition code hence is a $[n, 1, n]$ code.

In general\footnote{We here consider only so-called linear codes, see below for a more technical definition.}, the codewords are selected by defining $m$ so-called  \emph{parity checks} $C_i$. Each parity check specifies a subset of bits, which imposes a constraint: codewords are defined as those bitstrings in which the sum of variables in all these subsets is even\footnote{Technically, the codewords are hence the solutions of a specific instance of the XORSAT problem.}.
We are interested here in so called \emph{low-density parity check} (LDPC) codes, in which each bit interacts only with finitely many others: that is, every check $C_i$ only contains a bounded number of bits, and each bit is contained only in a bounded number of checks.

 To make the connection to statistical physics, it is natural to identify the $n$ bits with Ising variables $\sigma_j =\pm1$, $j=1\dots n$.  The even-parity constraints can be enforced energetically, by defining a Hamiltonian of the form
\begin{equation}\label{eq:Ising}
    \hamil = -\frac{1}{2} \sum_i^m \prod_{j \in C_i} \sigma_j + \frac{m}{2}.
\end{equation}
 The LDPC condition ensures that each interaction acts on finitely many spins, and that the number of interaction terms is (at most) extensive in $n$. The Hamiltonian is therefore sparse. Equivalently, there exists a finite-connectivity (possibly non-Euclidean) graph on which the Hamiltonian is local. 
 The $2^k$ codewords are the ground states of this Hamiltonian. As discussed below, this ground state degeneracy reflects spontaneous symmetry breaking of the symmetry group, $\mathbb{Z}_2^k$.

We note that the set of checks is not unique for a given set of codewords:
the repetition code, for example, is realized by the (ferromagnetic, nearest-neighbor) Ising model on an arbitrary graph. The two ground states (codewords) are the symmetry broken `all-up' and `all-down' states, and the Hamiltonian has a $\mathbb{Z}_2$ spin flip symmetry.

\newcommand{\errorvec}{\ensuremath{\vec \varepsilon}}

To make connections  to the computer science literature, we now define the above in more technical terms.  A classical code is a set of bitstrings of length $n$. In this paper, we consider \emph{linear codes} (or parity check codes), where the subset is defined as a \emph{subspace} $\codespace \subset \mathbb{F}_2^n$. Here and below, $\mathbb{F}_2 = \{0,1\}$ is the field with two elements, and arithmetic in $\mathbb{F}_2$ is defined mod 2.
The code subspace is the kernel of the \emph{parity check matrix} $H\in \mathbb{F}_2^{m\times n}$, i.e. $\Hcheck\codeword  = 0$ for all codewords $\codeword\in \codespace$ and $d = \min_{\vec z \in \mathcal C}\abs{\vec z}$, where $\abs{\vec x} \equiv \sum_j x_j$ denotes the Hamming weight. 
The matrix $H$ is related to the checks in \autoref{eq:Ising} by setting $\Hcheck_{ij} = 1$ if $j\in C_i$ and $\Hcheck_{ij} = 0$ otherwise: each row of $H$ corresponds to exactly one parity check and encodes its support. For bitstrings that are not codewords $\vec x \notin\mathcal C$, $ H\vec x = \vec s$ and the vector $\vec s \in \mathbb{F}_2^m$ is called the \emph{syndrome} of $\vec x$. 
If we identify a general bitstring $\vec x$ with a configuration of Ising spins $\vec\sigma$ by setting $\sigma_j = 1 - 2x_j\,\forall j$, then the nonzero entries in the syndrome of $\vec x$ correspond exactly to violated terms in \autoref{eq:Ising} in the state $\vec \sigma$. This means in particular that for the choice of normalization above
\begin{equation}\label{eq:Ising_energy_syndrome}
    \hamil(\vec\sigma) \equiv \abs{H \vec x}.
\end{equation}
The number of logical bits of the code is given by the dimension of the kernel, $k=\dim\mathcal C = \ker H$.
We are ultimately interested in \emph{families} of codes, $\Hcheck_n$, with increasing length $n$. In fact, many properties of ``codes'' should really be thought of as properties of a family. For example, we say that a family of codes satisfies the low-density parity check (LDPC) constraint if there exists two constants, $\wcheck$, $\wbit$, independent of $n$, such that $\sum_i(\Hcheck_n)_{ij} < \wbit$ and $\sum_j(\Hcheck_n)_{ij} < \wcheck$. In other words: the parity check matrix is sparse.

By design, the model defined in \autoref{eq:Ising} has $2^k$ ground states. Because of the underlying linear structure (on $\mathbb{F}_2^n$) of the problem, all these ground states can be written as the sum (mod 2) of $k$ binary vectors, which are the basis vectors of $\ker \Hcheck$.  In the Hamiltonian language, these ``logical bit flips'' ($\codeword_{\ell=1\dots k}$ such that $\linspan(\codeword_1, \dots, \codeword_\ell) = \mathcal C$), correspond to exact symmetries of the Hamiltonian, which thus has a symmetry group $\mathbb{Z}_2^k$ that is spontaneously broken in the ground state. 
We will denote spin configurations corresponding to codewords, that is the ground states of $\hamil$, as $\codestate$ (see e.g. \autoref{fig:landscape_intro}).

Another property that we will rely on later in the paper (in \autoref{sec:Glassiness}) is the absence of \emph{redundancies}. A redundancy is subset of checks such that their product is trivial, that is $\prod_{i\in R}\prod_{j\in C_i}\sigma_j = 1$. In the absence of such redundancies, under an appropriate, non-local change of variables $\sigma_j \to \tau_j$ the Hamiltonian of the system maps to that of a trivial paramagnet $\hamil = \sum_j \tau_j$, which implies that its partition function $\mathcal{Z}(T) \equiv \sum_{\vec{\sigma}} e^{- E(\vec{\sigma})/T}$ is analytic at all nonzero temperatures~\cite{eggarter1974cayley, montanari2006bethe, weinstein2019universality, rakovszky2023physics, hong2024quantum, yoshida2011feasibility}.
We will show, in \autoref{sec:Glassiness}, that nevertheless, $\hamil$ at low temperatures can realize a nontrivial phase in the sense of the Gibbs state decomposition described in \autoref{sec:gibbs_decomposition}, similar to what is observed in the Ising model on the tree (see \appref{app:TreeIsing}).
In the parity check matrix $H$, a redundancy corresponds to a linear dependency of rows, and the number of independent redundancies is given by its rank deficiency $r \equiv m - \rank H = m - (n-k)$.
The absence of redundancies also directly implies that model has \emph{point like excitations}, that is each term in $\hamil$ can be violated independently. This is because if $H$ has full rank, then the linear equation $H\vec x = \vec s$ has $2^k$ solutions $\vec x$ for \emph{any} $\vec s$. However, these excitations may not be mobile, in the sense that it may be necessary to flip a large number of spins (in general a finite fraction) to 'move' the excitation even to neighboring checks\footnote{Two checks are neighbors if they share a bit.}. This may happen due to the structure of energy barriers in the problem, which we now discuss. 

\subsection{Expander Codes}

\begin{figure*}
    \centering
    \includegraphics{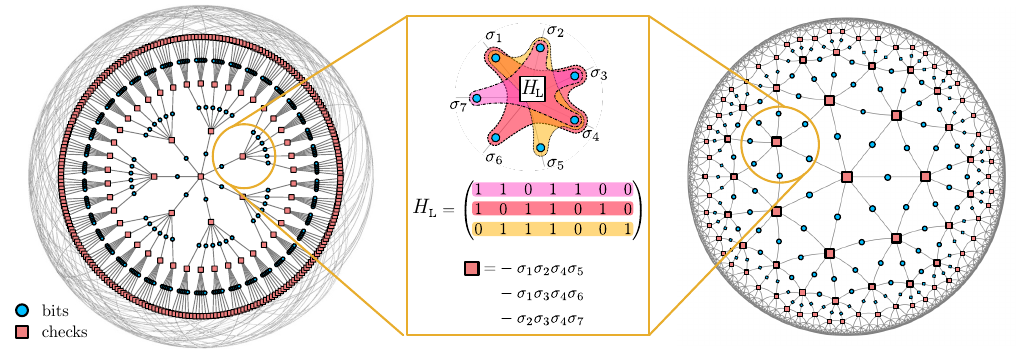}
    \caption{Tanner-Ising models. The models considered in this work are defined on expander graphs, e.g. locally tree-like expanders (left) or regular tessellations of the hyperbolic plane (right). Ising variables (blue circles) are placed on the edges of the graph and interactions (red squares) are defined on its vertices. The interaction is ferromagnetic and given in terms of a local code $\Hcheck_{\rm L}$ (depicted is the Hamming $[7, 4, 3]$ code) such that energy on each vertex is minimized when the Ising variables on incident edges form a codeword of $\Hcheck_{\rm L}$. 
    }
    \label{fig:tanner_codes}
\end{figure*}

The LDPC condition allows for the definition of a generalized notion of locality with respect to a (finite degree) graph, which need not be geometrically local in any finite-dimensional euclidean space. In fact, the models studied in this paper are all defined on \emph{expander graphs} (see \appref{app:constructions} for formal definitions and examples). 
Models defined on such graphs may have properties impossible to realize in any finite euclidean dimensions: for example they can have a number of ground states that scales exponentially with $n$ (i.e. $2^k$ with $k\sim n$), with all these ground states separated by extensive distance ($d\sim n$).  The codes corresponding to such models, with $[n, k, d] = [n, \Theta(n), \Theta(n)]$, are called \emph{good} codes.. While expansion of the underlying graph is a necessary condition to realize a good code, it by no means sufficient (for example, the Ising model on \emph{any} graph, has $k=1$). 

We will in the following study codes with an even stronger property, that is \emph{code expansion}. Informally, this demands that upon flipping bits starting from a ground state, the energy of the configuration grows \emph{linearly} with the number of flips, until a finite fraction of all bits have been flipped (\autoref{fig:lower_bound_sketch}). This ensures extensive energy barriers ($\propto n$) around all $2^k$ ground states. 

Code expansion is a very striking feature: in particular, it is impossible to realize in local models in any Euclidean geometry, where the vanishing bulk to boundary ratio means that some excitations will always have a sub-linear growth of energy with their size. For example, in an Ising model in $D$ dimensions, the energy cost of flipping a domain of $O(n)$ spins, \emph{i.e.} the size of the boudary domain wall, can scale as $O(n^{(D-1)/D})$.   Expansion of the underlying graph is thus necessary (but not sufficient) to ensure code expansion. Even in non-redundant cases where excitations are point-like, expansion guarantees that individual excitations cannot be locally moved  without proliferating more and more excitations.  
We present an intuitive picture in \autoref{sec:expander_examples} below for how the combination of an expanding graph and local constraints can produce code expansion. 

Formally we define code expansion in terms of the parity check matrix:
\begin{definition}
     A linear code defined by a parity check matrix $H\in \mathbb{F}_2^{m\times n}$ is called $(\delta, \gamma)$-expanding if
\begin{equation}
    \abs{\vec x} < \delta(n) ~\Rightarrow~ \abs{\Hcheck \vec x} > \gamma \abs{\vec x}.
\label{eq:expansion}
\end{equation}
for $\delta(n)=\delta\cdot n$ and some $\delta, \gamma > 0$.
\end{definition}
This directly corresponds to the above mentioned statement about energy barriers surrounding code words by \autoref{eq:Ising_energy_syndrome} and linearity\footnote{Use that $H(\vec{y} \oplus \vec{z}) = \Hcheck\vec{y} \oplus \Hcheck\vec{z}$, along with the fact that codewords are in the kernel, $\Hcheck\vec{z} = \vec{0}$.}.
Clearly, expansion implies a lower bound on the code distance, $d \geq \delta n$. We will refer to codes that are $(\delta,\gamma)$-expanding with $\gamma,\delta > 0$ as expander codes. 
Expansion is a special case of a more general property called linear confinement, or robustness\footnote{For experts, we note that code expansion as defined here does not imply local testability with constant \emph{soundness}, since we allow for large errors to have small syndromes.}, where the function $\delta(n)$ in \autoref{eq:expansion} is allowed to scale sublinearly with $n$.

Beyond providing a good code distance, expansion also has important implications for \emph{decoding}. To perform error correction on a linear code in practice, one needs an efficient way to decode, that is map a corrupted codeword $\codeword \oplus \errorvec$ back to the original codeword $\codeword$, which is possible if $\abs{\errorvec} < d/2$ (we use $\oplus$ to explicitly denote addition modulo two, and $\errorvec$ denotes an error corresponding to some spin flips of the codeword). As Sipser and Spielman \cite{sipser_spielman1996} showed, expander codes admit a very simple decoding algorithm which performs this task in linear time (in $n$) for errors up to a finite fraction of the distance. Their decoder, called \texttt{flip}, takes the form of a local greedy algorithm that resembles a zero-temperature Metropolis update\footnote{Note that when making this connection, ``linear time'' means a constant number of sweeps or constant ``parallel time''.}. Bits are flipped one at a time, in a way that ensures that the number of violated checks decreases with each flip. $(\gamma,\delta)$-expansion ensures that this algorithm corrects all errors with $\abs{\errorvec} \leq \delta n /2$ in $\mathcal{O}(n)$ time. From a physical perspective, the linear and extensive energy barriers allow efficient ``cooling" of a corrupted state back to the ground state.

\subsection{Examples of Expander Codes\label{sec:expander_examples}}

We now briefly discuss two broad classes of expander codes.
Deferring the discussion of rigorous guarantees on the parameters $\delta$ and $\gamma$, to \appref{app:constructions}, and the introduction of specific examples to \autoref{sec:numerics}, we focus here on the general construction. The two classes are \emph{Gallager codes} \cite{gallager1960thesis,gallager1962low}, where the parity check matrix $\Hcheck$ is chosen at random with fixed row and column weight, and Sipser-Spielman codes \cite{sipser_spielman1996} where bits are placed on the edges on an expander graph and then subjected to local constraints on each vertex. 
Note that Gallager codes correspond to what are called diluted ferromagnetic $p$-spin models in the spin glass literature \cite{franz2001ferromagnet, mezard2009information}. Since spin glass models based on Sipser-Spielman codes have, to the best of our knowledge, not been discussed in the literature before, we here discuss the general construction and an intuitive picture for their code expansion in some detail. 

The construction is sketched in \autoref{fig:tanner_codes}. Sipser-Spielman codes  are themselves part of a larger family of code constructions called \emph{Tanner codes} \cite{tanner1981recursive} which Sipser and Spielmann placed on expander graphs to construct large families of non-random good expander LDPC codes. We call the models derived from these codes \emph{Tanner-Ising models}.

Given an $s$-regular graph $G$ (i.e. a graph in whcih every vertex has degree $s$) one bit is placed on each edge, and each vertex is identified with a \emph{local code} $\Hcheck_{\rm L}$ defined on $n_{\rm L} = s$ bits. The parity checks of all local codes together define the  \emph{global code}. The codewords of the global code are those configurations of bits such that, for each vertex, the bit configuration on the incident edges is a codeword of the local code. To define a family of Tanner codes, we usually consider a family of $s$-regular graphs with increasing size, leaving the local code fixed. Note that such a family is naturally LDPC since interactions are strictly local on the graph.
One can show that, given certain conditions on the local code, the global code is expanding in the sense of \autoref{eq:expansion} if $G$ is an expander graph (see \cref{thm:expansion_sipser_spielman} for the rigorous guarantee).

\begin{figure}
    \centering
    \includegraphics{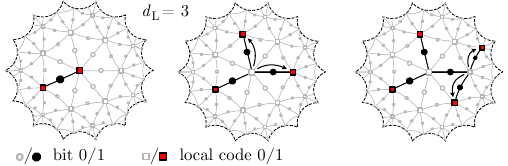}
    \caption{Immobile excitations in Tanner-Ising models. Flipping a single bit starting from the all-zero state (or any ground state) violates a pair of vertex constraints. If the local code has distance $d=3$, satisfying either one requires flipping at least two additional bits around that vertex. This forces excitations to proliferate as bits are flipped.}
    \label{fig:sscodes_excitation_split}
\end{figure}

Tanner-Ising models allow a very physical picture for the origin of expansion, which is illustrated in \autoref{fig:sscodes_excitation_split}. As shown on the left of the figure, flipping a single bit starting from the ground state (here chosen to be the all-zeros state for simplicity) violates a pair of vertex constraints. Since the vertex constraints take the form of a local code, if the local code distance $d_{\rm L}$ is larger than two ($d_{\rm L} = 3$ in the picture), then de-exciting either of the excited codes requires us to flip at least $d_{\rm L}-1 > 1$ additional bits around it. These however will trigger even more checks around them, due to the expanding geometry, and hence the number of excitations proliferates under local moves.
Remarkably, Sipser and Spielman showed that if the underlying geometry is an expander graph that fulfills certain conditions (see \appref{app:constructions} for details) one can guarantee that it is only after an extensive number of spins have flipped that excitations can start recombining enough to lower the energy.
In this case, the combination of a simple local constraint and the expanding geometry guarantees extensive barriers. 
Physically, expansion also produces an extreme immobility of excitations, because local excitations cannot be individually moved without proliferating more excitations, for the same reason as discussed above. This may be interpreted as an extreme  form of (classical) fractonic behavior, with much stronger barriers than has been observed in other constrained classical spin models in Euclidean geometries that have been studied in the context of fractons and spin-glasses \cite{newman_moore1999,garrahan2011kinetically,gromov2024fractons}.

For our rigorous results to apply, we will need that the code expansion parameter $\gamma$ in \cref{eq:expansion} is sufficiently large. The necessary values are known to be realized for Gallager codes with sufficiently large (but finite) bit-degree (see \appref{app:constructions}). However, for Sipser-Spielman codes, despite the  intuitive picture above, it is technically challenging to \emph{rigorously} establish large enough lower bounds on $\gamma$~\cite{breuckmann2021balanced}. However, we expect that in practice the expansion parameter is sufficiently strong for our results to apply even for these models. We verify this expectation numerically in \autoref{sec:numerics}.

\section{Spin Glass Order from Expansion}\label{sec:Glassiness}

The properties of good expander codes discussed in the previous section, namely an exponential number of ground states ($k\sim n$), large rearrangements of the spin configuration needed to move between these ground states ($d\sim n$), and immobile excitations, already hint at a complex energy landscape and nontrivial behavior at low temperature. Indeed, extensive energy barriers around ground states also implies extensive free-energy barriers at low enough temperatures. Thus, the 
the system breaks ergodicity at low temperatures, so that local dynamics initialized in a ground state component stays close to it \cite{montanari2006, hong2024quantum}. This already implies that the system is thermodynamically non-trivial in the sense of supporting multiple Gibbs state components as in  \autoref{sec:gibbs_decomposition}.

Our definition of spin glass order, however, goes beyond just stability of the ground states. As a reminder, for this we demand that the Gibbs state (i) shatters, i.e. no single component carries more than an exponential fraction of the weight, and (ii) it is incongruent, i.e. the number of relevant components grows with temperature, so that there is shattering \emph{within} each symmetry sector. 
While the former already follows from the exponential number of ground states and their stability, the latter property is less obvious to show.

We establish shattering and incongruence by first characterizing the energy landscape. We will then use this to derive properties about the free-energy landscape i.e. the Gibbs state decomposition. A strength of our approach is that we will \emph{directly} and \emph{explictly} make statements about the energy landscape and Gibbs state decomposition, providing a transparent derivation of how the structure of energy barriers in expander codes produces complex landscapes. 

We begin with informal statements of our main results, and then sketch the key steps in the corresponding proofs in the subsequent subsections. We refer the reader to \appref{app:complexity} for precise definitions and rigorous derivations.

We consider the structure of the set of states below an energy-density cutoff $\epsilon$
\begin{equation}
    \Omega(\epsilon) = \{ \vec x ~\text{such that}~ E(\vec x) < \epsilon n\}.
\end{equation}
We will show that this set, for certain models at sufficiently low $\epsilon$, can be decomposed into exponentially many disjoint and far-separated ``clusters" (\autoref{fig:landscape_intro}). Two states $\vec x$ and $\vec y$ are in the same cluster if and only if they are connected by a path in configuration space that flips (at most) a small fraction of bits at any step, and no intermediate state along the path has energy larger than $\epsilon n$. Each cluster is therefore associated with a local minimum in the energy landscape: it is impossible to traverse from one cluster to another via a path of small moves without passing through an extensive energy barrier.

Then, we can write 
\begin{equation}
    \Omega(\epsilon) = \biguplus_j C_j,
\end{equation}
with distinct clusters separated by extensive distance (${\rm dist}(C_i, C_j) > \xi n$ for some $\xi > 0$). 
Remarkably, one can show using only expansion and linearity of the energy functional, $E(\vec x) \equiv \abs{H\vec x}$, that under two assumptions: (i) sufficiently strong expansion $\gamma > \gamma^*$ and (ii) $H$ of full rank, no single cluster carries more than an exponentially small fraction of the weight. Here $H$ is the parity check matrix (\autoref{sec:ldpc}). 
We can then define a configurational entropy of the decomposition denoted by $\mathfrak s_{\rm conf}(\epsilon)$ and $\mathfrak s_{\rm conf}^{(\rm min)}(\epsilon)$ analogously to \autoref{eq:sconfig_def} and \autoref{eq:sconfmin_def}, respectively, but with the weights $w_j$ replaced by $\mathfrak w_i = \abs{C_i} / \abs{\Omega(\epsilon)}$. Stated in terms of this, our first main result is a lower bound on these entropies:

\begin{theorem}[Complexity of the Energy Landscape, Informal]\label{thm:landscape_expanders}
Consider a family of LDPC codes $H_n\in\mathbb{F}_2^{m\times n}$ with rate $r = k / n$, and that
\begin{enumerate}
    \item is $(\delta, \gamma)$ expanding with $\gamma > \gamma^*$
    \item has no redundancies, i.e. $\rank H = m$
\end{enumerate}
where $\gamma^*$ is a constant that depends only on $r$.
Then for sufficiently small $\epsilon$, the decomposition of $\Omega(\epsilon)$ shows \textbf{shattering} 
\begin{equation}
    \mathfrak s_{\rm conf}^{(\rm min)}(\epsilon) > 0
\end{equation}
and \textbf{incongruence}.
\begin{equation}
    \mathfrak s_{\rm conf}(\epsilon) > \mathfrak s_{\rm conf}(0) = r \geq 0.
\end{equation}
\end{theorem}

The above is a direct characterization of the energy landscape. Since the clusters are local energy minima, Theorem \autoref{thm:landscape_expanders} means that, at sufficiently low energy density, the landscape is comprised of exponentially many distinct local energy minima, and most of these minima do not contain ground states. No single minimum contains more than an exponentially small fraction of the low-energy density configurations. 

While the landscape is interesting in its own regard, in statistical mechanics we are ultimately interested in characterizing the structure of the Gibbs state. Here, following the discussion in \autoref{sec:gibbs_decomposition}, we will be concerned with the structure of the Gibbs state decomposition at low temperatures.
To this end, we leverage the decomposition of the landscape at low cutoffs described above, together with the fact that the Gibbs state is supported almost entirely on a microcanonical shell around the average. We can then show that at low temperatures, no Gibbs state component carries more than an exponentially small fraction of the total Gibbs weight.
Similarly to the above, this implies that for a particular Gibbs decomposition, the configurational entropy density [\autoref{eq:sconfig_def}] is finite and an increasing function of the temperature.
s
\begin{theorem}[Spin Glass Order, Informal]\label{thm:sg_expanders} 
Consider a family of codes $\Hcheck_n\in\mathbb{F}_2^{m\times n}$ with rate $r=k/n$, and that 
\begin{enumerate}
\item is ($\delta$, $\gamma$) expanding with $\gamma > \gamma^*$,
\item has no redundancies, i.e. $\rank \Hcheck_n = m$,
\end{enumerate}
where $\gamma^*$ is a constant that depends only on $r$.

Then the classical spin model based on this family [\autoref{eq:Ising}] at sufficiently low temperature realizes spin glass order in the sense of \cref{def:SG_def}. 
\end{theorem}

In the remainder of this section, we sketch the proof of the above theorem. We focus here on proof of spin glass order, and defer readers to \appref{app:complexity} for a proof of the shattering and incongruence of the energy landscape.
The flow of implications in the argument is sketched in \autoref{fig:proof_flowchart}.

\subsection{Expansion implies that all low-energy-density states are surrounded by extensive free-energy barriers\label{sec:bottlenecks_proof_main}}

\begin{figure}
    \centering
    \includegraphics{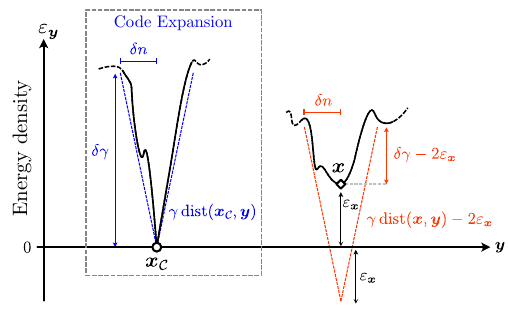}
    \caption{
    Sketch of the lower bound, \autoref{eq:energy_landscape_barriers}, on the energy barriers surrounding low-energy-density states of a code with the $(\delta, \gamma)$-expansion property.
    For states at zero energy, i.e. codewords $\codeword$ (left, blue), the bound just reduces to the expansion property \autoref{eq:expansion}. For states at finite energy density $\varepsilon$ (right, red), if $\varepsilon < \delta\gamma/2$ the bound implies that they are surrounded by extensive energy barriers. 
    }
    \label{fig:lower_bound_sketch}
\end{figure}

A crucial consequence of the form of the energy functional in \autoref{eq:Ising_energy_syndrome} is that expansion does not only implies that codewords/ground states are surrounded by extensive energy barriers but, in fact, that any state at sufficiently low energy density is surrounded by an extensive energy barrier and thus, in turn, an extensive free-energy barrier. 

To make this more precise, let $\vec{x}$ be a bit string corresponding to a configuration with energy $E(\vec x) := \varepsilon_{\vec x} n := \abs{\Hcheck \vec x}$. 
Let $\vec y$ be another string that is close in Hamming distance, i.e. $|\vec x \oplus \vec y| \leq \delta n$. By the triangle inequality, we have $E(\vec x \oplus \vec y) = \abs{H\vec x + H\vec y} \leq \abs{H\vec x} + \abs{H\vec y} = E(\vec x) + E(\vec y)$. Expansion of the code, on the other hand, implies $E(\vec x \oplus \vec y) \geq \gamma |\vec x \oplus \vec y|$. 
Combining these two inequalities, we get
\begin{align}
    E(\vec y) - E(\vec x) 
    	\geq  \gamma |\vec x \oplus \vec y| - 2E(\vec x)
    \label{eq:energy_landscape_barriers}    
\end{align}
which is sketched in the right half of \autoref{fig:lower_bound_sketch}.
We can then take $\vec y$ such that the right hand side is maximal, which gives $E(\vec y) - E(\vec x)  \geq (\gamma \delta - 2 \varepsilon_{\vec x}) n$. Note that for $\varepsilon_{\vec x} = 0$, \autoref{eq:energy_landscape_barriers} reduces to the definition of code expansion in \autoref{eq:expansion}, which is sketched on the left side of \autoref{fig:lower_bound_sketch}.

What does the above imply for the energy landscape of the code? Consider two states $\vec x$ and $\vec x'$, which differ on more than $\delta n$ sites, i.e. $\abs{\vec x \oplus \vec x'} > \delta n$, and that are both at the same energy density $\varepsilon_{\vec x} < \delta\gamma /2$.
The bound in \autoref{eq:energy_landscape_barriers} then implies that to reach $\vec x$ from $\vec x'$ by a sequence of local spin flips, we must have crossed a state $\vec y$ at a higher energy density $\varepsilon_{\vec y} > \varepsilon_{\vec x}$, that is we have crossed an \emph{extensive} energy barrier. 
In other words, expansion implies that all states $\vec x$ with $\varepsilon_{\vec x} < \delta\gamma /2$ are surrounded by extensive energy barriers. This is also sometimes called \emph{clustering} of states \cite{anshu2022cnlts,anshu2022nlts}, see also \cref{lem:clustering}.

As sketched in \autoref{fig:lower_bound_sketch}, this implies in particular that any state at low energy density is either close to some codeword, or separated from all codewords by an extensive barrier, meaning it is part of an energy well surrounding a local minimum.

The discussion so far has focused on the \emph{energy} barriers, while the bottleneck condition used to define nontrivial thermodynamic phases in \autoref{sec:gibbs_decomposition} is in terms of \emph{free-energy} barriers, that is the relative weight of a region $\Omega$ compared to that of its boundary $\Omega_{\eta}$. However, note that the energy barriers separating states are extensive, and entropy can at most be extensive (that is, the number of states in the boundary can at most be exponentially large its size). In other words: extensive energy barriers immediately imply extensive free-energy barriers at sufficiently low temperature.

Formally, around any low-energy state $\vec x$ with $\varepsilon_{\vec x} < \delta\gamma /2$, we can define a set $\Omega(\vec x)$ that is surrounded by a bottleneck in the sense of \autoref{eq:bottleneck_def}. In particular, define $\Omega_{\vec x} := \{\vec x + \vec b; \abs{b} < 2n\epsilon / \gamma\}$ for some $\varepsilon_{\vec x} < \epsilon < \gamma\delta/2$, as the Hamming ball of radius $2n\epsilon/\gamma$ around $\vec x$. Then, for a boundary $\partial_{\eta}\Omega_{\vec x}$  of width $\eta < \delta - 2\epsilon/\gamma$ [cf. \autoref{eq:bottleneck_def}], it is easy to see that
\begin{subequations}\label{eq:bottleneck_low_energy}
\begin{align}
    \frac{\pG(\partial_{\eta}\Omega_{\vec x})}{\pG(\Omega_{\vec x})} 
    &\leq \frac{\sum_{\vec u \in \partial_{\eta}\Omega_{\vec x}} e^{-\beta E(\vec u)}}{\sum_{\vec w\in\Omega_{\vec x}} e^{-\beta E(\vec w)}} \\
    &\leq \sum_{\vec u \in \partial_{\eta}\Omega_{\vec x}} e^{-\beta ( E(\vec u)- E(\vec x) )} \\
    &\leq \abs{\partial_{\eta} \Omega_{\vec x}}\,e^{-\beta n (\epsilon - \varepsilon_{\vec x})} \\
    &\leq e^{n [ \ln(2) -\beta(\epsilon - \varepsilon_{\vec x})]} \xrightarrow[n\to\infty]{\beta > \beta^*}0.
\end{align}
\end{subequations}
In the second line, we have lower bounded the weight of $\Omega_{\vec x}$ by that of $\vec x$ alone. In the third line, we used the definition of the boundary together with \autoref{eq:energy_landscape_barriers}. In the last line, we upper bounded the number of states in the boundary by $2^n$. Since by definition, $\epsilon > \varepsilon_{\vec x}$, the right hand side is vanishes exponentially in $n$ for sufficiently large $\beta>\beta^*$ (low temperature). Naturally, the resulting bound on $\beta^*$ is very loose, but here we are only interested in showing the fact that it is finite.

\subsection{Strong expansion and no redundancies imply shattering and incongruence}\label{sec:gibbs_shattering}

\begin{figure}
    \centering
    \includegraphics{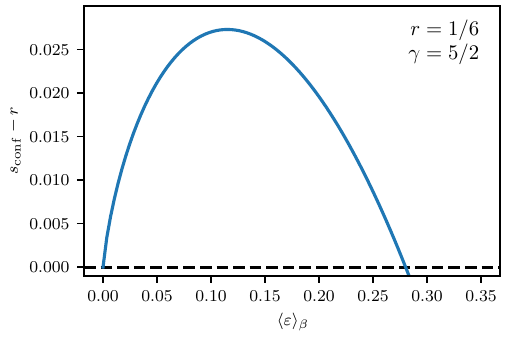}
    \caption{Lower bound for the configurational entropy $\sconf$ in \autoref{eq:pmax_temperature}, as a function of (the expectation value of) the energy $\expval{\varepsilon}_{\beta}$, shown for code rate $r$ and expansion parameter $\gamma$ as realized by the (5, 6)-LDPC ensemble.
    The bound is valid only at low temperatures, in particular, $\expval{\varepsilon}_{\beta} < \delta\gamma/2$ and $\beta > \beta^* > 0$ (see main text for details). 
    Due to the positive slope at $T=0$, and the fact that the bound is tight at $T=0$, an increasing configurational entropy is guaranteed for some finite range of temperature.}
    \label{fig:complexity_energy}
\end{figure}

We have established that low-energy-density states $\vec x$ are surrounded by free-energy barriers. In particular, \autoref{eq:bottleneck_low_energy} implies that we can define a Gibbs state component (in the sense of \autoref{sec:gibbs_decomposition}) around any such state, which is contained within a Hamming ball of radius $2n\epsilon/\gamma$. 
Next, to establish spin glass order in these system, we must characterise the number of distinct components that contribute to the Gibbs state or, more precisely, the configurational entropy defined in \autoref{eq:sconfig_def}.

The configurational entropy can be lower bounded as a function of temperature by a simple counting argument, if we assume sufficiently strong expansion ($\gamma > \gamma^*$) as well as the absence of redundancies ($\rank H = m$). Note that the arguments in the last subsection used neither assumption. 

Recall that the configurational Shannon entropy is lower bounded by the configurational min-entropy, and because of that to establish spin glass order in the sense of \cref{def:SG_def}, it is sufficient to show that
\begin{equation}
    \sconf \geq \sconf^{(\rm min)} = -\log w_{\rm max} > r.
\end{equation}
where $w_{\rm max} \equiv \max{\{w_i\}}$ and we define $\Omega_{\rm max}$ to be the component with Gibbs weight $w_{\rm max}$. 

The weight of any component can now be upper bounded as follows.
At sufficiently low temperature $\beta$ the weight of the Gibbs state is almost entirely concentrated around the average energy density $\expval{\varepsilon}_\beta$.
Because of this, it suffices to upper bound the weight of $\Omega_{\rm max}$ within this microcanonical window:
\begin{align}
	w_{\rm max} 
	&\approx Z_{\beta}^{-1} \, \sum_{\substack{\vec x \in \Omega_{\rm max},\\E(\vec x) \sim n\expval{\varepsilon}_\beta}} e^{-\beta E(\vec x)}\\
	&\approx Z_{\beta}^{-1} \abs{\Omega_{\rm max} \cap \Xi_{\beta}}~e^{-\beta n \expval{\varepsilon}_{\beta}}
\end{align}
where in the second line, $\Xi_{\beta}$ is the micro-canonical shell around energy density $\expval{\varepsilon}_{\beta}$.
Now, as mentioned  above, we know that if $\expval{\varepsilon}_{\beta} < \delta\gamma/2$ and $\beta > \beta^*$, then any Gibbs state component is contained within a Hamming ball of radius $<2n\epsilon / \gamma$, which in turn can be used to upper bound $\abs{\Omega_{\rm max}\cap \Xi_{\beta}}$. This can also be derived simply by inspecting \autoref{fig:lower_bound_sketch}: when the slope $\gamma$ of the energy is increased, any Gibbs state component confined to a given local minima and a given energy density is ``squeezed''.
Finally, in the absence of redundancies, we can derive a closed-form explicit expression for the partition function $Z_{\beta}$ as a function of inverse temperature. 
Putting these steps together (see \appref{app:complexity} for a detailed and rigorous version of this derivation), we arrive at a lower bound for the configurational entropy at low temperature that takes the form
\begin{align}\label{eq:pmax_temperature}
 s_{\rm conf}(T) \geq s_{\rm conf}^{(\rm min)}(T) 
 \geq&~ 
 r + s(T)
\end{align}
where $r = k / n$ denotes the rate of the code. 

The function $s(T)$ goes to zero at $T \to 0$, and further for a given rate $r$, there exist $\gamma^*(r)$ such that for all $\gamma > \gamma^*$, the function is positive and an increasing function of temperature at low temperature. This then concludes the proof of \cref{thm:sg_expanders}.

We show the bound explicitly in \autoref{fig:complexity_energy} as a function of expectation value of the energy, $\expval{\varepsilon}_\beta$, for parameters realized by the $(5,6)$-LDPC ensemble (see \appref{app:constructions} for details on constructing of this ensemble).

\subsection{Comments and possible generalizations\label{sec:shattering_general}}

Before turning to the detailed numerical study of some explicit examples of expander LDPC codes in the next section, we briefly comment on implications of the result and possible generalizations.

\subsubsection{Physical Intuition\label{sec:proof_intutition}}

To us, one of the most attractive features of the proof sketched above is that it gives a physically intuitive picture of the emergence of spin glass order in expander LDPC codes: expansion, together with linearity in the energy functional places tight constrains on the (free) energy landscape, which in turn implies shattering of the Gibbs state. In the last section, we also provided an intuitive picture for how the combination of expanding geometry and local constraints can produce code expansion quite generally.

Beyond the rigorous application to low temperatures, the bounds in \autoref{eq:energy_landscape_barriers}/\autoref{fig:lower_bound_sketch} also suggests some intuition for the behaviour of the system as a function of temperature more generally, which we sketch in the following.

Note that the lower bound on the \emph{height} of the barrier surrounding any state at low energy density $\varepsilon_{\vec x}$ decreases continuously with energy density and, in particular, becomes trivial strictly below the energy density corresponding to the height of  the energy barriers around codewords.
While the codewords remain separated from almost all other states up to an energy of at least $\sim n\gamma\delta$, generic states are only guaranteed to be separated above an energy $\sim \tfrac{n}{2}\gamma\delta$.
This suggest in particular that local minima at finite energy density may become unstable at temperatures strictly below the temperature where global minima become unstable.

This is of course exactly the intuition behind the phase diagram as a function of temperature drawn in \autoref{fig:landscape_intro}. As the temperature is lowered, we expect ergodicity to be broken \emph{weakly} first, at a temperature $\Tdyn$, in the sense that there are exponentially many distinct Gibbs state components, but one component carries almost all the weight. 
In contrast, the Gibbs states shatters only at a lower temperature $\Tglass$. 
Assuming the bound in \autoref{eq:energy_landscape_barriers} to be tight across the full landscape would then suggest $\Tdyn\sim \Edyn\approx\gamma\delta$ and $\Tglass\sim\Eglass\approx\tfrac{1}{2}\gamma\delta$.

\subsubsection{Instantiations}

In \appref{app:complexity} we explicitly show that broad families of Gallager codes provably statisfy all conditions of \cref{thm:sg_expanders} and hence realize spin glass order at low temperatures. These codes correspond  to so-called ``diluted ferromagnetic $p$-spin models'' considered in the spin glass literature, and hence this result is consistent with those based on replica- and cavity-method approaches which also find a low temperature spin-glass phase \cite{franz2001ferromagnet,franz2002dynamic}. Interestingly however, our conclusion is based on quite different arguments.
In particular, while our results use the (sufficiently strong) expansion in these codes, it does not make explicit reference to their underlying (locally) tree-like geometry. This makes our arguments applicable also to models where other approaches may not be useful, such as Sipser-Spielman codes. Importantly, they can also be generalised to quantum LDPC codes \cite{placke2024tqsg}.
Nevertheless, we mention that, for example, the cavity method allows explicit calculation of many quantities of interest (energy, critical temperatures, configurational entropy) if the Tanner/factor graph is locally tree like. In contrast, our results only yield lower bounds on these quantities in general.
We therefore view our method as complementary to existing approaches. 

The other class of expander LDPC codes considered in this work are Sipser-Spielman codes.
Unfortunately, the known lower bounds on the expansion coefficient $\gamma$, as well as available upper bounds on the number of redundancies are much weaker in this case for technical reasons.
However, based on the  picture described above (cf. \autoref{fig:sscodes_excitation_split}) we expect that these bounds are very weak and the expansion parameter $\gamma$ should be sufficiently large in practice. Furthermore, many families of these codes are found to have no or few redundancies (see \appref{app:redundancies} for a numerical study of this).
To verify this expectation we numerically study, in the next section,  two explicit, simple models based on the Sipser-Spielman construction. While for these models, the best available rigorous lower bound on the expansion parameter is in fact \emph{trivial}, we still find strong evidence that these models have, as expected, two critical temperatures corresponding to $\Tdyn$ and $\Tglass$ as explained above.

\subsubsection{Generalizations of the proof}

How much can the assumptions in \cref{thm:sg_expanders} be weakened without changing the conclusion? There are two obvious routes for technical improvement, which we mention for completeness.

First, the assumption of expansion can be weakened to only require linear confinement, that is requiring $\abs{H \vec x} > \gamma \abs{\vec x}$ for all $\abs{\vec x} < \delta(n)$ for some function $\delta(n)$ that scales sub-linearly with $n$. In particular, using arguments from percolation theory as well as graph-locality, one can show that it suffices to assume that $\delta(n)$ grows at least logarithmically with system size ($\log n/ \delta(n) < C~\forall n$ for some constant $C$ that is independent of $n$). The argument is analogous to that presented in Ref. \onlinecite{placke2024tqsg} for quantum codes and we defer a detailed discussion to future work.
We note that even for expander LDPC codes, requiring only a sub-linear scaling of $\delta(n)$ may increase the constant $\gamma$, which is required to be sufficiently large in \cref{thm:sg_expanders}. For example, in Ref. \onlinecite{hsieh2023explicit} (see Lemma 7.4), the authors show strong linear confinement (with $\gamma$ approaching the bit-degree of the code) can be guaranteed for $\delta(n)$ scaling exponentially with the \emph{girth} of the Tanner graph. A generalisation of this result to Tanner codes on high-girth graphs (e.g. with girth $\sim \log n$) would hence provide rigorous justification to apply \cref{thm:sg_expanders} to these codes as well.

Second, while \cref{thm:sg_expanders} requires the code to have \emph{no} redundancies, we expect the result to generalise to the case of families with a subextensive number of redundancies [that is $(m - \rank H) / n \to 0$ for $n\to\infty$] as well.
This is because for the argument presented above, while we used a closed-form expression for the partition function, it is sufficient to \emph{lower bound} it.
By the Kramers-Wannier duality (see e.g. Theorem A.9 in \cite{hong2024quantum} and \cite{rakovszky2023physics2}), a subextensive number of redundancies modifies the partition function only via a sub-extensive contribution to the free energy, so we expect a bound very similar to \autoref{eq:pmax_temperature} to hold in this case.

\section{Results on simple Tanner-Ising models on expander graphs\label{sec:numerics}}

We have shown in~\autoref{sec:Glassiness}, that certain spin models based on expander codes realize a complex energy landscape which leads to spin glass order at sufficiently low temperatures.
While the application of our arguments requires a strong guarantee on expansion that is not always available, we have also argued that we expect the same physics to be realize quite generally in expander codes.
Further, as explained in \autoref{sec:proof_intutition} and sketched in \autoref{fig:landscape_intro}(c), we expect in total \emph{three} different regimes as the temperature is lowered. A high-temperature paramagnetic phase where the Gibbs state is unique, and then two distinct ergodicity breaking phases. 
First, below a temperature $\Tdyn$ we expect weak ergodicity breaking that is there are many distinct Gibbs state components but a single component carries almost all of the total weight. Then below a different, lower temperature $\Tglass < \Tdyn$ we expect the onset spin glass order with shattering and incongruence. 

In this section, we provide evidence, via numerics and semi-analytical calculations, that  two ``minimal'' models inspired by Sipser-Spielman codes realize  this general scenario\footnote{The models are minimal in the sense that they are defined on small-degree graphs ($s=7$), which is the smallest that could potentially lead to good codes. Rigorous bounds establishing large enough $\gamma > \gamma^*$ are not available for these small-degree models.}.

We begin the remainder of this section by introducing the two concrete models that we study in detail. We then describe the numerical experiments and signatures that we used to identify the three regimes introduced above and determine the temperatures $\Tdyn$ and $\Tglass$. 

Finally, we semi-analytically investigate a Tanner-Ising model defined on the Cayley tree using a version of the population dynamics cavity method described in Ref.~\cite{mezard2006reconstruction}. 
We show that one can reproduce many results of the numerical experiments for one of our models, which is defined on a graph that is \emph{locally} tree like (see left side of \autoref{fig:tanner_codes}), by examining the influence of different boundary conditions on observables deep in the bulk. We connect the different choices of boundary conditions to the point-to-set construction on the closed graph. 

\subsection{Tanner-Ising Models}

We consider two specific examples of Tanner-Ising models, both sketched in \autoref{fig:tanner_codes} (see \autoref{sec:expander_examples} for a discussion of the general construction). The two families of models share the same local code but are defined on two different families of graphs. The first is based on a family of random regular graphs chosen in such a way that the graph is guaranteed to have no loops of size less than $\order{\log N_{\rm v}}$ \cite{linial2019rhgrg}, where $N_{\rm v}$ is the number of vertices. We refer to these as ``high-girth random regular graphs'' (HGRRGs). The high girth ensures that locally (below the size of the shortest loop), these graphs are isomorphic to a tree and thus their local structure is fully characterized by their degree $s$ ($s=7$ in our case). 
The other family of graphs we consider are regular tessellations of closed (two-dimensional) hyperbolic manifolds.
These are characterized by the so-called Schl\"afli symbol $\{r, s\}$, which denotes that the tiling uses regular polygons with $r$ edges and $s$ such polygons meet at every vertex of the tessellation. We consider $\{3, 7\}$-tessellations of closed hyperbolic manifolds. See~\appref{app:constructions} for details on the construction of both of these families of graphs.

To turn the graphs into Tanner codes, we need to specify a local code. In all cases, we take the local code to be the $[7, 4, 3]$ Hamming code which has parity check matrix
\begin{equation}
    \Hcheck_{\rm L} = \begin{pmatrix}
    1 & 1 & 0 & 1 & 1 & 0 & 0 \\ 
    1 & 0 & 1 & 1 & 0 & 1 & 0 \\ 
    0 & 1 & 1 & 1 & 0 & 0 & 1 
    \end{pmatrix}
    \label{eq:hamming_code}
\end{equation}
Note that not all bits are equivalent in the above parity check matrix; however, the Hamming code has a \emph{cyclic} property, by which its codespace is symmetric under cyclic permutations of the bits. 
To remove the ambiguity of labeling the edges around each vertex locally, we therefore \emph{symmetrize} the model, that is we add to the global code not only the checks in the parity check matrix defined above, but also all possible linear combinations of the three checks (modulo 2), for a total of seven checks.  
This is equivalent, up to a factor, to energetically distinguishing only between codewords and non-codewords of the local code. Note that this procedure introduces redundant checks, but only locally on each vertex of the graph. Hence, the system still decouples into $N_{\rm v}$ independent few-body systems which can be solved exactly. In particular,  we can efficiently sample states from the global Gibbs distribution $\pGibbs\propto e^{-\beta E}$ (see \appref{app:gibbs_samples} for details).

On hyperbolic manifolds, which are orientable, there is a natural ordering of the edges incident at each vertex, and hence a natural labeling with respect to the local code up to cyclic shifts. In this case, the symmetrization removes all ambiguity in the definition of the model up to a (global) choice of orientation of the manifold. Choosing a specific labeling of the edges with respect to the local code is less obvious for the HGRRGs and we chose an arbitrary labeling\footnote{We did not observe any dependence of the scaling of code parameters, redundancies, or critical behavior on the choice of labeling.}.

An important property of both families of graphs that we use to define our model is that even though we consider finite instances, they have diverging \emph{injectivity radius} $R_{\rm inj}\sim \log N_{\rm v}$. This means that the graphs look locally like their infinite counterpart (that is the Bethe lattice and regular tessellations of the infinite hyperbolic plane) up to distance $R_{\rm inj}$ around \emph{any} vertex of the graph. Note that this would not be the case when considering models defined on a sequence of Cayley trees, or hyperbolic lattices with open boundary conditions where a finite fraction of all spins lie within $\order{1}$ distance of the boundary (although these may still have non-trivial thermodynamics, see \appref{app:TreeIsing}). In the following, we choose system sizes with increasing injectivity radius and hence geometrically spaced number of bits $n\sim e^{R_{\rm inj}}$.

\begin{figure}
    \centering
    \includegraphics{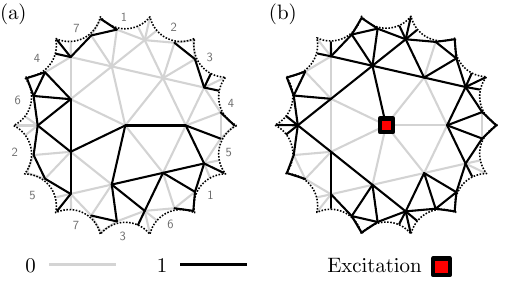}
    \caption{Illustration of a codeword (a), and a state with a single excitation (b), in the Klein-Hamming code. The Klein-Hamming code is constructed as a Tanner code on the Klein-quartic surface with the $[7, 4, 3]$ Hamming code as the local code.
    Hence, in a code word of the Klein-Hamming code each vertex has either zero or at least three incident edges labelled~`1'.
    The 14 sides are identified according to hyperbolic geodesics as indicated in the left panel.
    }
    \label{fig:klein-hamming}
\end{figure}

It is illustrative to consider the smallest member of the abovementioned families, the Klein-Hamming code, 
which is the Tanner code obtained from placing a local Hamming code [\autoref{eq:hamming_code}] on the Klein quartic (which is the smallest compactification of the $\{3, 7\}$-tessellation of the hyperbolic plane). The resulting global code has parameters $[n, k, d] = [84, 12, 26]$ and in \autoref{fig:klein-hamming}, we show both a codeword, in panel (a), as well as a state with a single excitation, in panel (b). Both states differ from the all-zero codewords on many sites.

\subsection{Numerical experiments\label{sec:hgrrg-numerics}}

\newcommand{\tauav}{\ensuremath{\tau_{\rm av}}}
\newcommand{\taumax}{\ensuremath{\tau_{\rm max}}}

\subsubsection{Heating and Annealing Dynamics}

\begin{figure}
\centering{}
\includegraphics{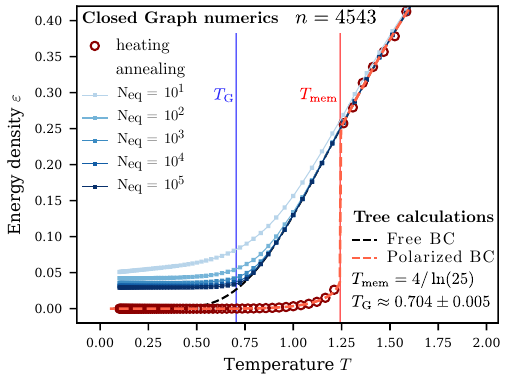}
\caption{
Results of Metropolis dynamics simulations on high-girth random regular graphs. We perform both heating an annealing simulations. For annealing dynamics, we initialize the system in a random state and slowly lower the temperature, with $N_{\rm eq}$ sweeps between each temperature change. For heating dynamics, we initialize the system in a ground state, quench the system to a higher temperature $T$, equilibrate the system at this higher temperature for $10^5$ sweeps, and then measure the energy for $10^3$ sweeps.
We compare the numerical results on a closed graph with calculations performed on the same model, but defined on a tree with boundary conditions (dashed lines), see \autoref{sec:recursive} for details. There is excellent agreement between the numerical results on the closed graph, and the calculation on the tree.
}
\label{fig:rrg_heating_cooling}
\end{figure}

Since our models are derived from codes, and we have established above that memory properties are useful in characterizing the thermodynamic properties we are interested in, we start our numerical study by conducting what could be considered a memory experiment. We initialize the system in the all-zero state (which is a ground state of the model with $E=0$) and evolve it under single spin flip metropolis dynamics at temperature $T$ for $10^5$ sweeps. The late-time value of the energy density reached in the simulation is shown in \autoref{fig:rrg_heating_cooling} (labeled ``heating'') for the model defined on a HGRRG with $n = 4543$ Ising variables. 
We observe a sharp first order like transition where the energy jumps from a value close to zero to the equilibrium value, that is the expectation value with respect to the \emph{global} Gibbs state $\pGibbs\propto e^{-\beta E}$.
We call this temperature $\Tdyn$. Below $\Tdyn$, the system does not reach the equilibrium value and ergodicity is broken. In particular, the state of the system remains close to its initial state (for times that are analytically predicted to be exponentially long in $n$ by the considerations of the bottleneck theorem and section \autoref{sec:bottlenecks_proof_main}). Hence, below $\Tdyn$ the system acts as a passive memory for its ground states. 
While the number of sweeps used to equilbrate the system at temperature $T$ is arbitrary, we have checked that the position of the jump in energy does not depend on the number of sweeps used and is stable on increasing system size. The fact that $\Tdyn$ is independent of such details is further evidenced by comparing the numerical heating results on the closed graph to a calculation on the same model defined on an infinite tree (see below for details) with \emph{polarized} boundary conditions, that is we force the system to be in a ground state at the boundary. The average energy density obtained from this calculation is shown in \autoref{fig:rrg_heating_cooling} as a red dashed line and agrees perfectly with the numerical heating experiment on the closed graph. The tree calculation also yields an analytical result for the transition, $\Tdyn = 4/\ln(25)$. 

\begin{figure*}
    \centering
    \includegraphics{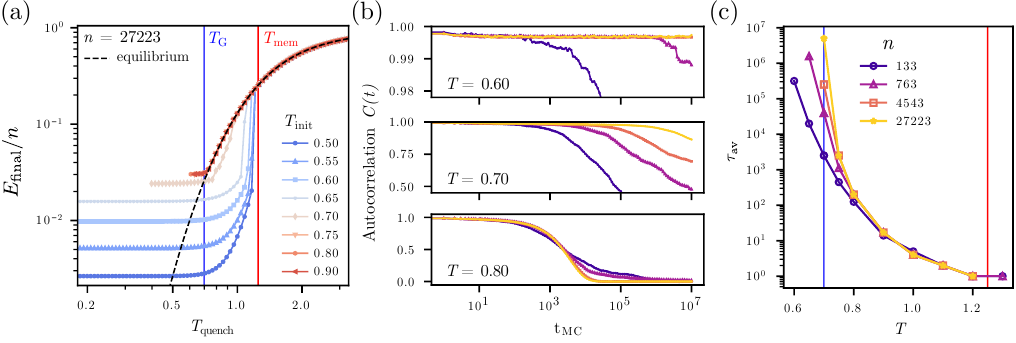}
    \caption{
    Numerical evidence for the glass transition on high-girth random regular graphs.
    (a) Late-time average energy $E_{\rm final}$ after quenches from $T_{\rm init}$ to $T_{\rm quench}$. The energy after quenches to below $\Tglass$ (solid blue vertical line) do never track the equilibrium curve (black dashed line), and for $T_{\rm init} < \Tglass$ the system falls out of equilibrium even when quenched to some range of higher temperatures $T_{\rm init} < T_{\rm quench}$.
    (b) Average autocorrelation function as a function of Monte-Carlo time, for a temperature below, approximately at, and above $\Tglass=0.70\pm0.05$.
    (c) Average relaxation time $\tauav$ [\autoref{eq:tauav}] as a function of temperature.
    }
    \label{fig:glass}
\end{figure*}

Our numerics show that at low temperature and under local Metropolis dynamics, the system does not explore all of configuration state when initialized in a ground state. A converse question would be whether it is possible to reach such a ground state when cooling the system slowly from infinite temperature. 
To address this numerically, we simulated annealing experiment: we initialize the system in a random state ($T=\infty$) and then change the temperature according to a schedule with $101$ geometrically spaced temperatures between $T=10$ and $T=0.1$, and a varying number of $N_{\rm eq}$ metropolis sweeps per temperature. The results of these experiments for the same model as before is also shown in \autoref{fig:rrg_heating_cooling} (marked ``annealing'').
The energy of the simulated annealing tracks the equilibrium curve at high temperature, while at very low temperature, the energy of simulated annealing is almost completely flat, indicating that the system is ``stuck'' in a local minimum.
As  $N_{\rm eq}$ is increased (cooling is done more slowly), the crossover from tracking the equilibrium curve to being flat becomes more narrow, occurring over a small window around $\Tglass\approx 0.7$. 
The existence of many local minimum at finite energy, where local dynamics gets ``stuck'' and hence (numerical) cooling experiments fall out of equilibrium, is a hallmark of (spin) glasses and hence suggest that indeed our Tanner-Ising models realize spin glass order below $\Tglass\approx 0.7$. 
Again, we can compare this  to an estimate from a calculation done on a tree (see \autoref{sec:recursive} below) which yields $\Tglass\approx 0.704\pm0.005$, again in excellent agreement with the closed-graph numerics.

Note that the glass transition temperature is significantly below the memory threshold temperature $\Tdyn\approx 1.25$.
This serves as direct evidence for the existence of two distinct critical temperatures $\Tdyn$ and $\Tglass$ (see sketch in \autoref{fig:landscape_intro}). 

What is the intuition behind these rather simple numeric picking up on exactly the two transitions? While the system has multiple ergodic components below $\Tdyn$, a single large component dominates above $\Tglass$. Simulated annealing below $\Tdyn$ then explores this largest component, which reproduces the expectation values of observables in the global Gibbs state. Only below $\Tglass$ does the system get stuck in a local minimum containing a vanishing fraction of all configurations and thus falls out of equilibrium.

\subsubsection{Quench Dynamics in the Spin Glass Phase}

We now further characterize the nature of the low-temperature phase. According to the picture of spin glass order  in \autoref{sec:spin_glass_order}, below $\Tglass$ the Gibbs state shatters into many distinct components which are dynamically stable, and most of which do not contain ground states. 
This has direct consequences for Metropolis dynamics at  $T < \Tglass$, when initialized in a configuration spin randomly sampled from the Gibbs distribution at the same temperature $T$. Such dynamics should only explore a vanishing fraction of configuration space, and we expect the system to act as passive memory for such initial states. We test these two expectations in the following.

We study \emph{quench dynamics}: we initialize the simulation in a state drawn randomly  from the global Gibbs distribution at a temperature $T_{\rm init}$ and then evolve using local metropolis dynamics at a different temperature $T_{\rm quench}$ for $10^5$ sweeps before measuring the energy for $10^3$ sweeps.
The resulting steady-state energy density $E_{\rm final}/ n$ as a function of $T_{\rm quench}$, for a range of $T_{\rm init}$ is shown in \autoref{fig:glass} (a) as solid curves with markers in a range of colors ranging from blue for low $T_{\rm init}$ to red for large $T_{\rm init}$.

We observe qualitatively different behavior depending on whether $T_{\rm init} < \Tglass$ or $T_{\rm init} \geq \Tglass$. 
In particular, for  $T_{\rm init} < \Tglass$, the energy $E_{\rm final}$ is roughly independent of $T_{\rm quench}$ for $T_{\rm quench} \leq T_{\rm init}$. In contrast, for $T_{\rm quench} > T_{\rm init}$ the energy increases with $T_{\rm quench}$ and there is a sharp jump of $E_{\rm final}$ as a function of $T_{\rm quench}$, similar to the one observed in \autoref{fig:rrg_heating_cooling}. 
The temperature at which the jump occurs decreases as $T_{\rm init}$ increases. 

The first observation, that $E_{\rm final}$ is roughly independent of $T_{\rm quench}$ for $T_{\rm quench} \leq T_{\rm init}$ means that the system is unable to equilibrate to a lower energy than that of the initial state, compatible with the system being stuck in a local minimum of the energy. 
The second observation, a sharp jump in $E_{\rm final}$ as a function of $T_{\rm quench}$, is evidence for the fact that the system acts as passive memory for typical states sampled from the Gibbs distribution at $T < \Tglass$. 
The fact that the position of the jump decreases as $T_{\rm init}$ increases is evidence that, as discussed in \autoref{sec:proof_intutition},  local minima of the energy landscape are less stable as their energy density approaches $\Eglass$.
Finally, for $T_{\rm init} \geq \Tglass$, the quench dynamics basically reproduces the annealing dynamics of \autoref{fig:rrg_heating_cooling}. The value of $E_{\rm final}$ tracks the equilibrium curve for all $T_{\rm quench} > \Tglass$, meaning that even if the system is quenched to a lower temperature, the local dynamics successfully equilibrates to the equilibrium energy. For $T_{\rm quench} < \Tglass$, the energy gets stuck roughly at the equilibrium energy at $\Tglass$, consistent with the dynamics getting stuck around that energy. 

\begin{figure*}
    \centering
    \includegraphics{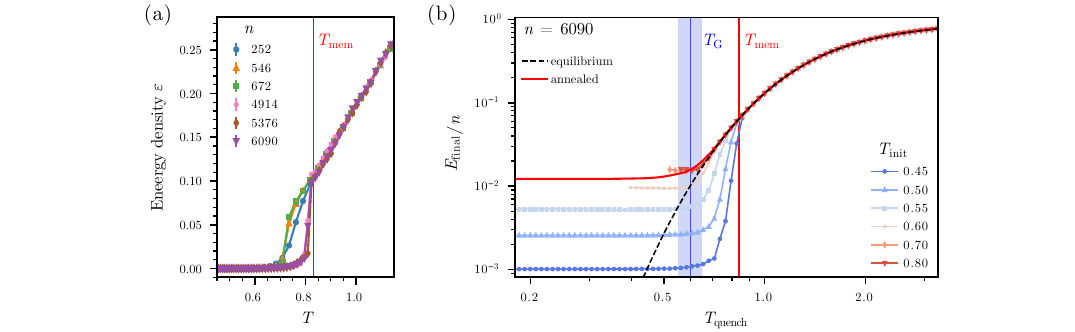}
    \caption{Results for metropolis dynamics simulations on $\{3, 7\}$ hyperbolic lattices. 
    The data and interpretation in panel (a) and (b) are analogous to those presented for the same model on HGRRGs in \autoref{fig:rrg_heating_cooling} and \autoref{fig:glass} (a), respectively. As such, the data is evidence that also Tanner Ising models on hyperbolic lattices have two transitions at distinct temperatures $\Tglass$ and $\Tdyn$, respectively. Since the hyperbolic lattices are not locally tree like, no approximate recursive solution of these models is available. Because of this, $\Tdyn$ and $\Tglass$ are estimated from the numerical data as the temperature where the energy under heating dynamics jumps to the equilibrium value, and the temperature where the energy under annealing dynamics starts digressing from the equilibrium curve, respectively. 
    }
    \label{fig:glass_hyperbolic}
\end{figure*}

Finally, in \autoref{fig:glass_hyperbolic}, we show data equilvalent to that in \autoref{fig:rrg_heating_cooling} and \autoref{fig:glass} (a) but for the Tanner-Ising model on the $\{3, 7\}$ tiling of the hyperbolic plane. Since the hyperbolic lattices are not locally tree like, no approximate recursive solution of these models is available. However, from the numerical data we can conclude $\Tglass = 0.60 \pm 0.05$ and $\Tdyn = 0.84 \pm 0.02$ for the Tanner-Ising model on the $\{3, 7\}$ tiling of the hyperbolic plane. As before, these two transitions are well separated.

\subsubsection{Autocorrelation time}

Another signature of ergodicity breaking is a large autocorrelation time. The autocorrelation time, in particular, quantifies on which time scale on which the system explores regions of configuration space that are `far' from the initial state. Because of this, we expect the autocorrelation time to diverge strongly with system size for $T < \Tglass$.

There are multiple ways to define the autocorrelation time. For a given initial state $\vec\sigma_{\rm init}$,  we define its autocorrelation time to be the first time the expectation value of the autocorrelation function falls below a threshold value
\begin{align}
    \tau(\vec\sigma_{\rm init}) &= \min_t \left\{\frac{1}{n}\sum_j \expval{\sigma_{{\rm init}, j}\,\sigma_j(t)} < 0.9\right\},
\end{align}
where we have taken the threshold value to (arbitrarily) be 0.9. We denote by $\vec\sigma(t)$ the state time-evolved under stochastic single-spin flip metropolis dynamics at a given temperature $T$, and the expectation value $\expval{\bullet}$ is taken over different trajectories of such dynamics (starting from the same initial state).

Naturally, the value of $\tau(\vec\sigma_{\rm init})$ depends sensitively on the initial state. We therefore distinguish the autocorrelation time of ``worst case'' and average initial states by defining
\begin{align}
    \taumax &= \max_{\vec\sigma}  \tau(\vec\sigma) \label{eq:taumin}\\
    \tauav &= \sum_{\vec\sigma} \pGibbs(\vec \sigma) \tau(\vec \sigma),
    \label{eq:tauav}
\end{align}
where the max and sum over $\vec \sigma$ implies taking the maximum and sum over all of configuration space, respectively.
Note that in $\tauav$, we take the Gibbs distribution $\pGibbs$ to be at the same temperature as the one used for the stochastic dynamics. 
As we will see, the quantities \autoref{eq:taumin} and \autoref{eq:tauav} are defined in such a way that they diverge strongly as a function of system size for $T$ below $\Tdyn$ and $\Tglass$, respectively, and show only very weak\footnote{Even in a paramagnetic phase, one expects the autocorrelation time to diverge logarithmically with $n$.} system size dependence otherwise.

Note that the quench dynamics studied above already provides a lower bound on $\taumax$, since the system falling out of equilibrium implies that $\taumax > 10^5$, which is the number of sweeps used in the quenched simulations. 
If we assume that the codewords are the most stable initial states at all temperatures, i.e. $\taumax = \tau(\codestate)$ for an arbitrary code word $\codestate$, then, the results of the previous section also already imply that $\taumax \ll 10^5$ for $T > \Tdyn$ and $\taumax \gg 10^5$ for $T < \Tdyn$ as expected. 

Similarly, the fact that the system does equilibrate after a quench from high temperature as long as $T_{\rm quench} > \Tglass$ suggest that the average autocorrelation time diverges only below $\Tglass$. However, this provides no strict lower or upper bound since the average in the definition $\tauav$ is defined with respect to the Gibbs state at the same temperature. 
To study $\tauav$ more directly as a function of temperature, we initialize the system, as before, in a state drawn from the equilibrium distribution at temperature $T$ and evolve it under local metropolis dynamics until the autocorrelation function, averaged over 500 trajectories, decays to zero, or for a maximum of $10^7$ sweeps. Assuming that the average is not dominated by rare events, this allows us to directly estimate $\tauav$ as long as $\tauav < 10^7$. The result is shown in \autoref{fig:glass} (b) and (c). In panel (b), we show the autocorrelation function $C(t) = \tfrac{1}{n}\sum_j \expval{\sigma_j \sigma_j(t)}$ as a function of Monte-Carlo time $t_{\rm MC}$, while in panel (c) we show the average autocorrelation time $\tau_{\rm av}$ as a function of temperature for different system sizes. As expected, we observe a stark change in the way the autocorrelation function depends on system size between $T > \Tglass$ and $T < \Tglass$. This is particularly apparent in panel (c), where we observe that the average autocorrelation time does not scale with system size for temperature $T > \Tglass$ but diverges as a function of system size for $T < \Tglass$.

\subsection{Recursive solution on trees}\label{sec:recursive}

In the previous subsection, we probed the multiplicity of pure Gibbs states numerically through their \textit{dynamical} signatures, identifying \Tdyn~and \Tglass~with transitions in heating and cooling, respectively. However, as we have defined them, weak and strong ergodicity breaking are not merely dynamical features, but structural ones, which therefore can be probed via certain \textit{equilibrium} correlation functions as well as the configurational entropy. To evaluate these numerically and analytically, we rely on the alternative formulation of pure Gibbs states as the limits of different boundary conditions. This perspective, which differs from that presented up to this point (cf.~\autoref{sec:gibbs_decomposition}), is especially useful on locally tree-like graphs, where extremal Gibbs states correspond to fixed points of belief propagation~\cite{mezard2009information}. At a technical level, we analytically determine \Tdyn~by solving for ``ferromagnetic'' fixed points, while \Tglass~is framed as a ``reconstruction threshold'' assessed in the spirit of the cavity method~\cite{mezard2001bethe,mezard2006reconstruction}.

In general, weak and strong ergodicity breaking manifest not in the standard correlation function between distant spins, but rather, in the correlation between the \textit{boundary} of a graph and a spin deep in the bulk. If this correlation is nonzero, then we have a finite probability to succeed in reconstructing the bulk spin given knowledge of only the boundary configuration~\cite{mezard2009information}. \Tdyn~and~\Tglass~are associated with two different \textit{boundary reconstruction} tasks. Above \Tdyn, the Gibbs state is \textit{unique}, so that the bulk states induced by any pair of boundary conditions are indistinguishable in the bulk~\cite{krzakala2007gibbs}. At and below~\Tdyn, reconstruction is possible for certain fine-tuned BCs, such as fully aligned all up or all down, which induce a nonzero bulk magnetization. The fact that a bulk spin ``remembers'' the fully polarized BCs corresponds to the memory of the ground state in the heating experiment. However, in the intermediate regime $\Tglass<T<\Tdyn$, boundary configurations that are \textit{typical}, in a sense more carefully defined below, still fail to influence the bulk. Not until $T$ falls below $\Tglass$ does the length scale of the correlation between boundary and bulk diverge. This diverging length scale is accompanied by a diverging relaxation time (cf.~\autoref{fig:landscape_intro}c)~\cite{Martinelli2004,Berger2005,montanari2006}, manifesting in the Monte Carlo dynamics as ergodicity breaking in the annealing experiment.

The boundary reconstruction problem is general, but can only be solved controllably on locally tree-like graphs, as the effect of different boundary conditions on the leaves of the tree can be iterated recursively into the bulk. Since there are no loops, this iterative process ``flows'' in one direction, with the initial condition set by the particular boundary condition. Said differently, belief propagation converges to an exact fixed point when the factor graph is a tree. We note that since the boundary of a Cayley tree is an extensive fraction of the bulk, imposing different BCs is a standard way to probe the thermodynamically non-trivial nature of models on the Bethe lattice~\cite{peierls1936ising,friedli2017statistical,chayes1986mean}, despite the triviality of the partition function.
We review how this program plays out for the ferromagnetic Ising model on the Bethe lattice in~\appref{app:TreeIsing}.  

The fixed points of the flow depend on both the temperature and the \textit{ensemble} of boundary conditions used. For $T>\Tdyn$, the uniqueness of the Gibbs state means that \text{any} choice of boundary conditions flows to a unique, trivial fixed point, the paramagnetic global Gibbs state. At $T=\Tdyn$, the recursion relation associated with ``codeword-polarized'' boundary conditions--- those which favor the ground states---develops nontrivial ``ferromagnetic'' fixed points.  The recursion relation associated with ``typical'' boundary conditions, however, does not acquire a nontrivial fixed point until $T=\Tglass$. 

\begin{figure}[t]
\includegraphics[width=\linewidth]{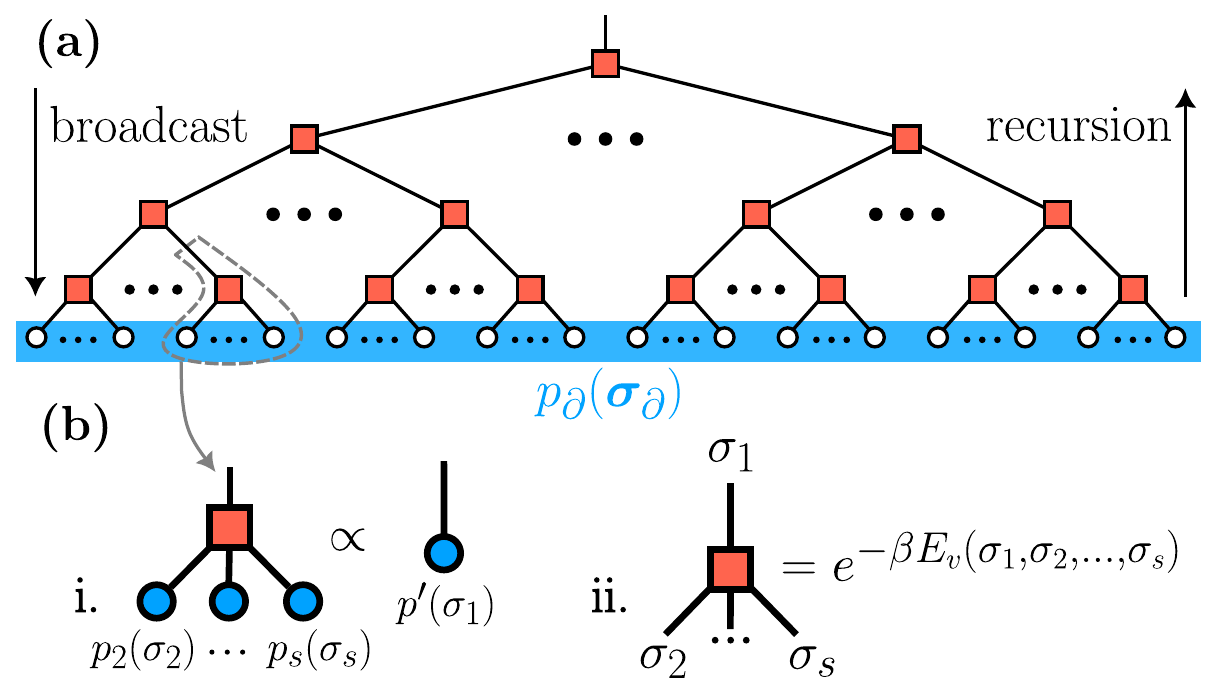}
\caption{\label{fig:tree-sketch} (a) Depth $\depth=4$ rooted tree with a (possibly correlated) probability distribution $p_\partial(\vec{\sigma}_\partial)$ on its boundary (highlighted in blue). Dashed gray curve indicates the neighborhood of one vertex $v$. (b) Recursion step (i) mediated by the interaction $E_v$ around a single vertex (ii) [\autoref{eq:pp-recursion}]. This interaction also defines one step of the noisy broadcasting process (running from root to leaves) discussed in~\autoref{sec:sg-tree}.}
\end{figure}
\subsubsection{Setup}
To be concrete, consider a Tanner code on a \textit{rooted} tree, where each vertex except the root has degree \gdeg, and the total number of layers is \depth. 
A depth $\depth=4$ tree is sketched in~\autoref{fig:tree-sketch}a. The last layer of the tree has a set of edges sticking out (the ``leaves'', shown as white circles), which constitute its boundary. Using the Ising variables $\sigma = 1-2x$, we denote the corresponding set of boundary edges by~$\vec{\sigma}_\partial$ while we reserve~$\vec{\sigma}$ for the configuration of the remaining spins in the bulk of the tree.
The energy is a function of both sets of variables, $E=E(\vec{\sigma},\vec{\sigma}_\partial)$. Given a particular probability distribution for the boundary spins, $p_\partial(\vec{\sigma}_\partial)$, the probability of a global configuration $(\vec{\sigma}, \vec{\sigma}_\partial)$ is
\begin{equation}
p(\vec{\sigma},\vec{\sigma}_\partial) = p_\partial(\vec{\sigma}_\partial) p(\vec{\sigma}|\vec{\sigma}_\partial).
\end{equation}
where
\begin{subequations}
\begin{align}
p(\vec{\sigma}|\vec{\sigma}_\partial) &= e^{-\beta E(\vec{\sigma},\vec{\sigma}_\partial)} / \mathcal{Z}_\partial(\vec{\sigma}_\partial), \\\mathcal{Z}_\partial(\vec{\sigma}_\partial) &\equiv \sum_{\vec{\sigma}} e^{-\beta E(\vec{\sigma},\vec{\sigma}_\partial)}
\end{align}
\end{subequations}

Suppose we sample the boundary conditions with probability $p_\partial(\vec{\sigma}_\partial) = \mathcal{Z}_\partial(\vec{\sigma}_\partial)/\mathcal{Z}$. These are ``free boundary conditions'' in the sense that $p(\vec{\sigma},\vec{\sigma}_\partial) \propto e^{-\beta E(\vec{\sigma},\vec{\sigma}_\partial)}$ is just the Gibbs distribution of the model on the full tree. Alternatively, we can tune the temperature of the boundary conditions relative to the interior by taking 
\begin{equation}\label{eq:alpha}
p_\partial(\vec{\sigma}_\partial) \propto \mathcal{Z}_\partial(\vec{\sigma}_\partial)^{\alpha} = e^{-\alpha \beta F_\partial}.
\end{equation}
This choice of BC interpolates between a uniformly random (``quenched'') BC at $\alpha=0$ and a codeword-polarized (``ferromagnetic'') BC at $\alpha=\infty$, which selects only the ground states.\footnote{Those familiar with one-step replica symmetry breaking may recognize $\alpha$ as the \textit{1RSB Parisi parameter}~\cite{mezard2006reconstruction,krzakala2007gibbs}.}

To relate rooted trees with certain boundary conditions back to closed graphs, such as those on which the models of the previous subsection were defined, consider the following process: (a) Sample from the equilibrium Gibbs state on a closed graph $G$, (b) freeze the configuration outside a ball $B_\depth(i)$ of radius \depth, centered at vertex $i$, and (c) sample inside the ball (the induced subgraph $G(\depth)$) according to the measure conditional on the frozen spins. This procedure gives a rigorous way to impose ``typical'' boundary conditions on $G(\depth)$. The \textit{point-to-set correlation} which probes the spin glass transition is then defined as the correlation between the ``point'' $i$ and the ``set'' $\overline{B}_\depth(i)$, quantified, for example, by the mutual information between $i$ and $\overline{B}_\depth(i)$~\cite{mezard2009information}.

If $G$ is a locally treelike graph, then $G(\depth)$ can be taken to be a Cayley tree. While the rigorous point-to-set construction involves sampling boundary configurations on $G(\depth)$ according to the partition function on the full graph $G$, taking $\alpha = 1$ in~\autoref{eq:alpha} is a minimal way to mimic these correlated BCs, consistent with the formulation of one-step RSB within the cavity method~\cite{mezard2006reconstruction}.

\subsubsection{Memory transition}
We first consider the case $\alpha=\infty$, which probes \Tdyn. While $\alpha=\infty$ selects \textit{all} ground states (by sampling the boundary at zero temperature), we can equivalently favor a \textit{single} codeword, taken without loss of generality to be the all-up codeword.

This fully polarized BC is special on two counts.
First, the boundary distribution \textit{factorizes}, $p_\partial(\vec{\sigma}_\partial) = \prod_{e \in \partial} p_e(\sigma_e)$, allowing the effect of the boundary to be straightforwardly iterated into the bulk. That is, letting $i=2,\ldots,\gdeg$ label the set of edge spins that are connected to the same vertex $v$, summing over $\sigma_e$ gives rise to a new effective distribution on the spin (let us denote it by $\sigma_1$) which connects $v$ to the rest of the tree (\autoref{fig:tree-sketch}b):
\begin{equation}\label{eq:pp-recursion}
p'(\sigma_1) \propto \sum_{\sigma_2,...,\sigma_{\gdeg}} e^{-\beta E_v(\sigma_1,\ldots,\sigma_{\gdeg})} \prod_{i=2}^{\gdeg} p_i(\sigma_i),
\end{equation}
where $E_v$ denotes the terms in the Hamiltonian associated with the checks on vertex $v$. In the language of belief propagation,~\autoref{eq:pp-recursion} defines one step of ``message passing'' from the leaves towards the root.

Defining the conditional magnetization $m = p(+1) - p(-1)$,
\begin{align}\label{eq:m-recursion}
m' &= \frac{\sum_{\sigma_1,\dots,\sigma_\gdeg} \sigma_1 e^{-\beta E_v(\sigma_1,\dots,\sigma_\gdeg)} \prod_{i=2}^{\gdeg} p_i(\sigma_i)}{z(m_2,...,m_\gdeg)} \notag \\
&\equiv F(m_2,...,m_\gdeg)
\end{align}
where
\begin{equation}\label{eq:z}
z(m_2,...,m_\gdeg) = {\sum_{ \sigma_1,...,\sigma_\gdeg}e^{-\beta E_v(\sigma_1,...,\sigma_\gdeg)} \prod_{i=2}^{\gdeg} p_i(\sigma_i)}.
\end{equation}
Second, the polarized boundary distribution is \textit{homogeneous}: $m_e \equiv m_0$ for each boundary spin $e$. Thus, if vertex $v$ is at depth $r$ from the leaves,~\autoref{eq:m-recursion} becomes
\begin{equation}\label{eq:m-recursion-ferro}
m_{r+1} = F(m_r,...,m_r) \equiv f(m_r).
\end{equation}

The fixed points which characterize the bulk are solutions to the equation $f(m_*) = m_*$, and their stability is determined from the derivative $\frac{\text{d}f}{\text{d}m}{\big |}_{m_*}$, since a perturbation away from $m_*$ scales as $\Delta m \sim \lambda^r$ where $\lambda = \frac{\text{d}f}{\text{d}m}{\big |}_{m_*}$. The paramagnetic fixed point, $m=0$, always solves this equation. At high temperatures, this is the unique solution towards which any initial boundary condition flows. However, at sufficiently low temperatures, other, non-trivial solutions can appear. \Tdyn~is the highest temperature which admits nontrivial fixed points: since the polarized boundary condition sets the initial condition $m_0 = \pm 1$, the system flows to a nontrivial bulk magnetization as soon as one of these solutions appear.

Concentrating on the case where the local code is the $[7,4,3]$ Hamming code, we derive in~\appref{app:TensorNetwork} that 
\begin{equation}\label{eq:tree-mem}
\Tdyn = 4 / \ln(25) \approx 1.24.
\end{equation}
As shown in~\autoref{fig:ferro}, the pair of fixed points appearing at $T=\Tdyn$ are at a finite conditional magnetization, $m = \pm \sqrt{2/3}$, indicating a first-order transition.

\begin{figure}[t]
\includegraphics[width=\linewidth]{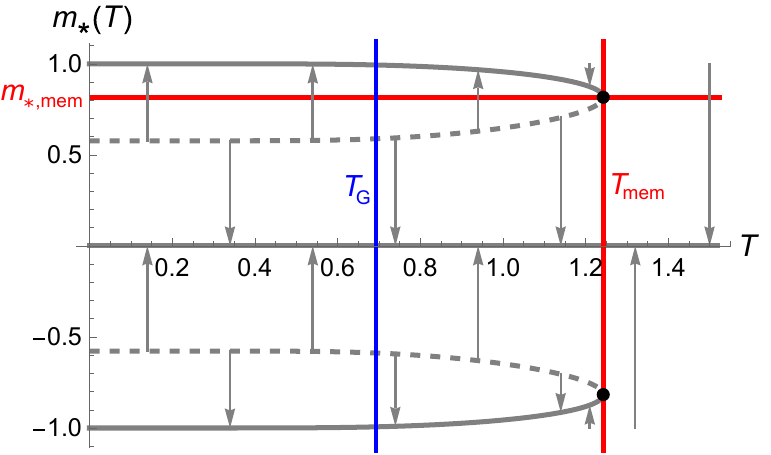}
\caption{Fixed points of the recursion [\autoref{eq:m-recursion-ferro}] under polarized BCs, for a Tanner-Ising model on a tree graph with the $[7,4,3]$ Hamming code as the local code. Solid (dashed) curves indicate stable (unstable) fixed points. Arrows indicate the flow under recursion. Red and blue lines mark \Tdyn~[\autoref{eq:tree-mem}] and \Tglass [\autoref{eq:Tglass}], respectively. Black points mark the nontrivial fixed point magnetization at \Tdyn, $\pm m_{*,\mathrm{mem}} = \sqrt{2/3}$.\label{fig:ferro}}
\end{figure}
In fact, down to $T=0$, the trivial fixed point remains stable ($\frac{\text{d}f}{\text{d}m}{\big |}_{0} = 0 < 1$) and is separated from the stable large-$|m|$ fixed points by a symmetric pair of unstable fixed points. 
Thus, one needs to add a finite boundary field to induce magnetization in the bulk, in contrast with the Ising model on a tree (\appref{app:TreeIsing}). We find similar results for the $[8,4,4]$ local code. 

The analytic calculation of $\Tdyn = 4/\ln(25)$ on a tree is in excellent agreement with the numerical estimate for $\Tdyn$ on the HGRRG with the same local code, discussed in the previous section. Another point of agreement is in the bulk energy density per spin $ \frac{2}{7}\langle E_v\rangle$, where $\langle E_v\rangle$ is the local energy per vertex.\footnote{The prefactor $2/7$ arises because a degree $\gdeg$ regular graph with $N_v$ vertices has $\gdeg N_v/2$ edges.} The latter can be evaluated from the fixed-point magnetization $m_*$ of the root by taking 7 independent rooted trees, each with $p_*(\sigma)=(1+m_*\sigma)/2$ at the root, and joining them at a common vertex $v$:
\begin{equation}
    \langle E_v \rangle = \frac{\sum_{\sigma_1,\dots,\sigma_\gdeg}e^{-\beta E_v(\sigma_1,\dots,\sigma_\gdeg)}E_v(\sigma_1,\dots,\sigma_\gdeg) \prod_{i=1}^\gdeg p_*(\sigma_i)}{\sum_{\sigma_1,\dots,\sigma_\gdeg}e^{-\beta E_v(\sigma_1,\dots,\sigma_\gdeg)} \prod_{i=1}^\gdeg p_*(\sigma_i)}.
\end{equation}

The bulk energy densities for the Tanner-Ising model with the $[7,4,3]$ code at each vertex, in the naive Gibbs state (open BC) and with fully polarized BC are shown by the dashed lines in~\autoref{fig:rrg_heating_cooling}. While the former is a smooth function of the temperature, the latter has a discontinuity at $T=\Tdyn$, indicating the first-order transition, and has excellent quantitative agreement with the energy density obtained from Monte Carlo numerics on the locally tree-like HGRRG initialized in a fully polarized initial state. We thus find that the tree graph model correctly captures the memory transition of the closed graph. 

\subsubsection{Spin glass transition}\label{sec:sg-tree}

The memory transition occurs when the valleys with the largest free-energy barriers---those surrounding codewords--- become stable. Operationally, this means that, given a boundary configuration which is consistent with a global codeword, we can confidently deduce the state of a spin in the bulk. On the other hand, at finite temperature, typical states belong to valleys not surrounding a codeword. The spin glass transition is where those typical valleys, induced by $\alpha=1$ conditions, become stable. We can frame the reconstruction of typical boundaries as a classical decoding problem in the following sense~\cite{evans2000broadcasting,mezard2006reconstruction}.

To generate a typical boundary condition, consider sampling a configuration at inverse temperature $\beta$ in an iterative manner, a process that runs in the opposite direction of the recursive flow (see~\autoref{fig:tree-sketch}a). After fixing the root to $\sigma_0=\pm 1$ at random, we noisily \textit{broadcast} this spin to the leaves, one layer at a time: in each layer, if the input branch to vertex $v$ carries spin $\sigma_1$, then the spin configuration $\sigma_2,...,\sigma_\gdeg$  on the outgoing branches is sampled with probability $\propto \exp[-\beta E_v(\sigma_1,\dots,\sigma_\gdeg)]$ (cf.~\autoref{fig:tree-sketch}b.ii). Running this process up to depth $r$, we can then freeze the boundary configuration $\vec{\sigma}_\partial$ and, from it, try to \textit{reconstruct} the root spin.\footnote{The ensemble of $\vec{\sigma}_\partial$ of course depends on depth $r$, but we omit the superscript to avoid burdening the notation.} The optimal decoding strategy computes a conditional root magnetization, $m_0(\vec{\sigma}_\partial)$,  and guesses that $\sigma_0=\mathrm{sgn}[m_0(\vec{\sigma}_{\partial})]$, succeeding with probability $(1 - |m_0|)/2$. This conditional magnetization belongs to a probability distribution, $Q^{(r)}(m)$\
\begin{equation}\label{eq:Q-r}
    Q^{(r)}(m) = \mathrm{prob}[m_0(\vec{\sigma}_\partial) = m] = \sum_{\vec{\sigma}_\partial} p(\vec{\sigma}_\partial) \delta(m - m_0(\vec{\sigma}_\partial)).
\end{equation}

In the previous subsection, we used a particular (polarized) boundary condition to induce a nonzero magnetization. The broadcasting process, in contrast, generates a $\mathbb{Z}_2$-symmetric ensemble of boundary conditions, such that the average magnetization must vanish. Thus, the onset of spin glass order must instead be probed by a higher moment of the magnetization distribution. The variance of a bulk spin, akin to the Edwards-Anderson order parameter, is discussed in~\appref{app:TensorNetwork}; here, applying the reconstruction perspective, we consider the mutual information between the root and the boundary~\cite{evans2000broadcasting,mezard2006reconstruction}\footnote{The integrand is simply the Kullback-Leibler divergence, for a given boundary condition, between the conditional root distribution $p(\sigma_0)=(1 \pm \sigma_0 m)/2$, and the unconditional distribution $p(\sigma_0)=1/2$~\cite{mezard2006reconstruction}.}:
\begin{align}\label{eq:I-r}
I(r) &=\frac{1}{2}\int \sum_{\sigma_0=\pm 1} (1 + \sigma  m) \log_2(1 + \sigma_0 m) \mathrm{d}Q^{(r)}(m).
\end{align}
At depth $r=0$, reconstruction is perfect: the root and boundary are identical, so $Q^{(0)}(m) = [\delta(m-1) + \delta(m+1)]/2$, and $I(0)=1$. The mutual information decreases monotonically with depth, and either converges towards zero (signaling that the bulk and boundary have decoupled, $T>\Tglass$) or plateaus at a finite value ($T<\Tglass$). 

To determine $I(r)$, we simulate a distributional recursion relation for $Q^{(r)}$, generalizing the recursive update for $m_r$ in the previous subsection\footnote{We can interpret this recursion relation as performing \textit{survey propagation} with surveys weighted by Parisi parameter $\alpha=1$, rather than the standard $\alpha=0$~\cite{krzakala2007gibbs}.}:
\begin{equation}\label{eq:eq18}
Q^{(r+1)}(m) \propto \int \delta(m - F(\bm{m})) z(\bm{m}) \prod_{i=2}^\gdeg dQ^{(r)}(m_i)
\end{equation}
where $\bm{m} = (m_2,...,m_s)$, and $z(\vec{m})$ was defined in~\autoref{eq:z}. We simulate~\autoref{eq:eq18} numerically up to large $r_{max}$ via the population dynamics method of Ref.~\cite{mezard2006reconstruction}. Further details on the method are provided in~\appref{app:TensorNetwork}. 

\begin{figure}[t]
\includegraphics[width=\linewidth]{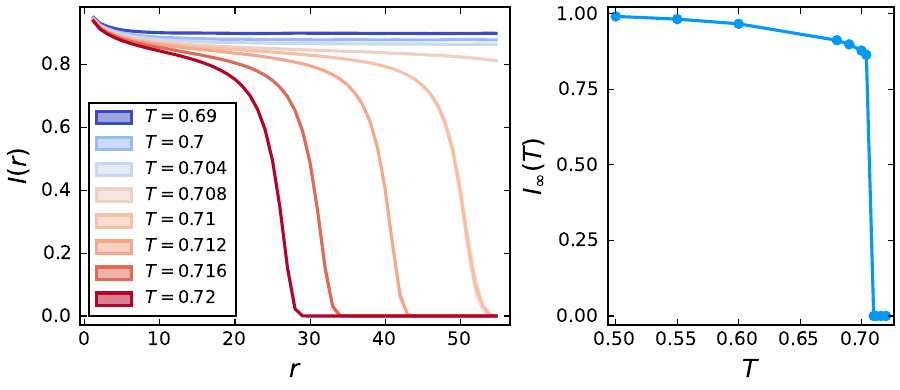}
\caption{Left: Mutual information between root and boundary [\autoref{eq:I-r}] as a function of the tree depth $r$ at varying temperatures $T$. The local code at each node is the Hamming [7,4,3] code. Thin ribbons around each curve indicate the standard error across 10-20 independent simulations. Right: Plateau value of the mutual information, i.e. $\lim_{r\rightarrow\infty} I(r)$, as a function of temperature, determined by averaging over the last 10 iterations of the recursion relation ($r_{max}=80$ for $T\geq 0.704$). Error bars are smaller than the data points. \label{fig:Ir-plateau}}
\end{figure}

The data in~\autoref{fig:Ir-plateau} indicate a discontinuous transition at
\begin{equation}\label{eq:Tglass}
    \Tglass =0.704 \pm 0.005
\end{equation}
is in good agreement with that estimated from HGRRGs in the previous section. The first-order nature of the transition is evidenced by the discontinuity in $I_\infty(T) = \lim_{r\rightarrow\infty}(I(r)).$

Since the recursion relation~\autoref{eq:eq18} does not admit an analytical solution, in \appref{app:TensorNetwork} we modify the problem in a way that does admit an exact analytical treatment. This is achieved by taking $\alpha=2$ in~\autoref{eq:alpha}, which leads to a recursion relation on two copies of the original problem and can be formulated elegantly in the language of tensor networks. From this, one can calculate a critical temperature $T^{(2)}$. Since this modified boundary condition is effectively sampling the boundary conditions from a lower temperature compared to the bulk, we expect it to provide an upper bound on the true spin glass temperature. Indeed, we find $\Tglass < T^{(2)} \approx 1.009 < \Tdyn$.

\subsubsection{Configurational entropy}
In~\autoref{sec:gibbs_shattering}, we obtained a lower bound,~\autoref{eq:pmax_temperature}, on the configurational entropy $s_{\mathrm{conf}}$. This lower bound requires a sufficiently strong expansion parameter $\gamma$, and becomes trivial for the Tanner-Ising models considered in this section, for which it has not even been proven that $\gamma > 0$. In the absence of rigorous bounds, we can nevertheless demonstrate, by numerically estimating the configurational entropy within the population dynamics method, that $s_{\mathrm{conf}}$ is nonzero and increases with temperature at sufficiently low temperatures for a locally tree-like Tanner-Ising model whose local code is the symmetrized Hamming [7,4,3] code.

Recall that for these models, the global Gibbs state factorizes across vertices, allowing us to exactly solve for the entropy density $s_\beta \equiv \frac{1}{n}\frac{\partial(T \log Z_\beta)}{\partial T}$. While $s_\beta$ itself is insensitive to the spin glass transition, below \Tglass, it splits into two contributions: the average entropy density \textit{within} a typical component, and the Shannon entropy density of the  Gibbs decomposition, which is none other than the configurational entropy density.  Thus, we can recover \sconf~by subtracting from $s_\beta$ the intracomponent contribution. This contribution is isolated by imposing $\alpha=1$ BCs: below \Tglass, a typical boundary condition drawn from this ensemble drives the system into a single extremal component at that temperature or energy. Thus, 
\begin{equation}
\sconf(T) = s_{\beta}(T) - s_{\alpha=1}(T).
\end{equation}
At $T=0$ (or equivalently, at energy $\expval{\varepsilon}_\beta = 0$), $s_{\mathrm{conf}} = s_{\beta} = r$, where $r=1/7$ is the global code rate. \autoref{fig:complexity-tree} shows the surplus configurational entropy beyond this bare code rate, $\sconf - r$, as a function of $\expval{\varepsilon}_\beta$ up to $T=0.68<\Tglass$. $s_{\alpha=1}$ is evaluated using population dynamics, which at sufficiently low temperatures, is in excellent agreement with a low-temperature expansion to leading order in $y= \exp(-4\beta)$ (\appref{sec:complexity-tree}):
\begin{equation}\label{eq:s-conf-series}
    s_{\mathrm{conf}}(y) = \frac{1}{7} + 2y(1 - \ln(2y))/\ln(2).
\end{equation}

\begin{figure}[t]
\centering
\includegraphics[width=0.8\linewidth]{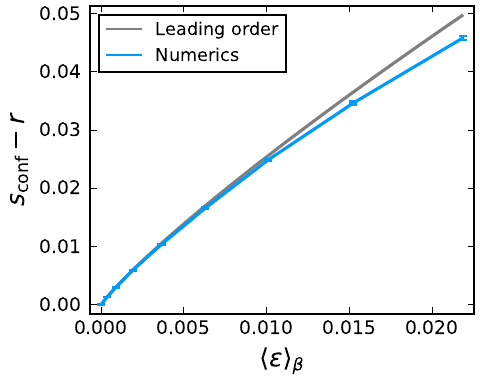}
\caption{Configurational entropy density of the Tanner-Ising model on a locally tree-like graph, minus the global code rate $r=1/7$, as a function of the expectation value of the energy $\expval{\varepsilon}_{\beta}$. The blue curve is the result obtained from population dynamics, averaged over $\geq 50$ independent simulations. The gray curve is a low-temperature expansion to leading order in $\exp(-4\beta)$ [\autoref{eq:s-conf-series}]\label{fig:complexity-tree}.
}
\end{figure}

\section{Summary and Outlook}

In this paper, we have established a direct connection between \emph{code expansion}, a property of large class of classical error correcting codes defined on sparse non-Euclidean graphs, and \emph{spin glass order}, defined in terms of complex energy and free energy landscapes, which dictate the properties of the associated classical spin Hamiltonian at low  temperatures.
We establish this connection rigorously, under the condition that the codes have sufficiently strong expansion and no redundancies. These conditions are fulfilled, for example, for models based on families of Gallager codes, which take the form of (a particular variant of) diluted $p$-spin glasses \cite{franz2001ferromagnet,franz2002dynamic}. 

Our rigorous results can be separated into two main contributions: a characterization of the structure of the energy landscape, and a characterization of the structure of the Gibbs state. For the energy landscape, the set of low energy density configurations decomposes into disjoint \emph{clusters} -- corresponding to local minima of the landscape --- and under the above-mentioned conditions and at low cutoff energy we show (i) that the clusters are separated by an extensive distance in configuration space, (ii) \emph{shattering}: no single cluster contains more than an exponentially small fraction of all configurations below the cutoff, and (iii) \emph{incongruence}: almost all clusters do not contain ground states.

The properties of the energy landscape lead to a decomposition of the Gibbs state into disjoint \emph{components}. These components are restrictions of the (global) Gibbs state to regions of configuration space which are separated from each other by extensive free energy barriers. This leads to a bottleneck under any local, detailed balance obeying dynamics and thereby endows the components with dynamical stability. 
We explicitly construct such a decomposition of the Gibbs state, such that (i) all components are surrounded by bottlenecks (ii) \emph{shattering}: no single component carries more than an exponentially small fraction of the weight, and (iii) \emph{incongruence}: almost all components have no support on the ground states.

The proof of these results is simple and also physically intuitive. First, sufficiently strong code expansion guarantees steep energy barriers around ground states, and the linearity of the energy functional implies that these barriers carry over to states at low but finite energy densities.  
This already implies a decomposition of the energy landscape into multiple components. A simple counting argument, relying on the fact that there are no redundancies of the checks and hence number of states with a given energy can be easily bounded, then allows for a lower bound on the relative weight of these clusters. 
For the Gibbs state results, we then additionally use the fact most of the Gibbs weight is concentrated in an infinitesimal window around the average energy density.

We argued that the barrier structure necessary for the above argument arises naturally in certain spin models which we called \emph{Tanner-Ising models}, defined on expander graphs and with spins subjected to a local constraint.
Motivated by this, we numerically studied two such models of codes that are not rigorously known to realize the conditions of the proof (that is, sufficiently strong expansion and the lack of redundancies) but for which the intuitive picture applies. We called these models . 
The two models use the same local constraints but are defined on locally tree-like random regular graphs (RRGs), and hyperbolic lattices, respectively. 

The family of locally tree-like RRGs allows comparison between extensive dynamical Monte-Carlo simulations on closed finite graphs, and semi-analytic recursive techniques (i.e., the cavity method) applied to the same model on a tree. We find excellent agreement between these different approaches, providing strong evidence that this model undergoes \emph{two} transitions in the structure of its Gibbs state as temperature is lowered, at two distinct temperatures $\Tdyn$ and $\Tglass$, respectively. The family of hyperbolic lattices has small loops and hence does not allow a solution by recursive methods. Still, in Monte-Carlo simulations on closed graphs show the same hallmark features as for the family of RRGs, providing strong evidence that the same two transitions are realized in this setting. 

At high temperatures $T> \Tdyn$, the Gibbs state is unique, with a single component carrying all weight. Below $\Tdyn$, ergodicity is broken \emph{weakly}: there are exponentially many Gibbs state components, but a single one still carries almost all weight of the global state.
Dynamically, this corresponds to a situation where the system serves as a passive memory when initialized in special initial states, while typical initial conditions equilibrate to the global Gibbs state. 
Finally, at $T < \Tglass < \Tdyn$ the Gibbs state shatters and develops spin glass order: no single component carries more than an exponentially small fraction of the weight. Dynamically, this is the temperature below which the system fails to reach the global equilibrium energy under local annealing, and the system serves as a passive memory for typical initial states.

In summary, our results provide a new perspective on sparse, non-euclidean models of spin glasses with linear constraints. While these have been studied extensively in the literature by both rigorous \cite{franz_leone2003replica,panchenko_talagrand2004bounds, DemboMontanari2010_Ising,DemboMontanariSun2013_Factor,DemboMontanariSlySun2014_Potts} and non-rigorous \cite{franz2001ferromagnet, franz2002dynamic,krzakala2007gibbs, ricci_tersenghi2010xorsat} methods, our results are complementary to these existing approaches and, to the best of our knowledge, constitute the first rigorous proof of finite-temperature spin-glass order in models with closed finite-degree interaction graphs. 

Our methods have the advantage of not relying on the locally tree-like nature of the underlying interaction graph, making them applicable, at least in principle, to a larger setting such as models defined on hyperbolic lattices (and also allows generalization to quantum models, explored in Ref. \cite{placke2024tqsg}). Further, in contrast to existing approaches, we are able to construct an explicit sequence of decompositions of the Gibbs state at finite size, and explicitly characterize properties of this decomposition, providing a direct proof of spin glass order. While our approach, unlike the cavity method, does not yield an explicit calculations of quantities of interest such as the (free) energy, or the configurational entropy \cite{franz2002dynamic}, it nevertheless gives a lower bound on the configuration entropy for a specific decomposition of the Gibbs state without having to make additional uncontrolled assumptions (such as the validity of the replica symmetri breaking solution in the cavity method).

Our work suggests multiple directions for future work. While our proof provides a bound on the number and size of Gibbs state components, a more detailed understanding of their distribution would give further insight into the low-temperature physics of these models. This could include a proof of the intermediate, weak ergodicity breaking phase observed in our numerics. At the same time, while we present a basic numerical study of a diluted spin glass model on a graph with small loops (the hyperbolic Tanner-Ising model in \autoref{sec:numerics}), it would be interesting to see whether the presence of small loops \emph{qualitatively} changes any of its thermodynamic properties, e.g. correlations or scaling of the barriers. 

It would also be desirable to instantiate the conditions required by our main rigorous results for a broader class of models, such as the Tanner-Ising models that we studied numerically. The main challenges here lies in establishing sufficiently strong bounds on the expansion parameter $\gamma$, as well as the absence of redundancies. To this end, we note that, in principle, for many of our results to apply it is only necessary to bound the parameter $\gamma$ for a subextensive range of spin flips (i.e. for $\delta(n)$ sublinear in \autoref{eq:expansion}, see also the results of Ref. \onlinecite{placke2024tqsg}) and stronger lower bounds may be possible in this setting. Furthermore, rigorously proving the absence of redundancies (which we numerically confirm for a range of Tanner codes in \appref{app:redundancies}) has not been of central interest in coding theory (since redundancies are generally beneficial) so it is not unlikely that improvements over known bounds are possible. Furthermore, we expect that it should be possible to relax the condition of no redundancies, e.g. by requiring merely that their number is sub-extensive, while retaining our main results.

\section{Acknowledgments}

We thank Daniel Fisher, Silvio Franz, Luis Golowich, David Huse, Steven Kivelson, Ethan Lake, Roderich Moessner, Akshat Pandey, Shivaji Sondhi, Gilles Tarjus and Ruben Verresen, for helpful discussions, and Andrea Montanari for pointing out useful references.

B.P. acknowledges funding through a Leverhulme-Peierls Fellowship at the University of Oxford and the Alexander von Humboldt foundation through a Feodor-Lynen fellowship.
T.R. was supported in part by Stanford Q-FARM Bloch Postdoctoral Fellowship in Quantum Science and Engineering, by the HUN-REN Welcome Home and Foreign Researcher Recruitment Programme 2023 and by the Supported Research Groups Programme, HUN-REN-BME-BCE Quantum Technology Research Group (TKCS-2024/34). Numerical work by G.M.S. was completed using computational resources managed and supported by Princeton Research Computing, a consortium of groups including the Princeton Institute for Computational Science and Engineering (PICSciE) and the Office of Information Technology's High Performance Computing Center and Visualization Laboratory at Princeton University.
This work was done in part while N.P.B. was visiting the Simons Institute for the Theory of Computing, supported by DOE QSA grant \#FP00010905.
V.K. acknowledges support from the Packard Foundation through a Packard Fellowship in Science and Engineering and the Office of Naval Research Young Investigator Program (ONR YIP) under Award Number N00014-24-1-2098. 

\bibliography{references}

\appendix

\section{Gibbs State decomposition for the two-dimensional Ising model\label{app:decomposition_Ising}}

In this appendix, we explain the pure state decomposition of the Euclidean Ising model at sufficiently low temperatures. Recall from the main text that this means we want to find a decomposition of the configuration space into non-overlapping subsets $\Omega_{\pm}$ such that
\begin{subequations}
\begin{align}
\frac{\pGibbs(\partial_\epsilon \Omega_{\pm})}{\pGibbs(\Omega_{\pm})} \xrightarrow[n \to \infty]{} 0
\end{align}
where
\begin{equation}
  \partial_\epsilon\Omega \equiv \{\vec\sigma \notin\Omega; \dist(\vec \sigma,\Omega) \leq \epsilon n \}.
\end{equation}
\end{subequations}
For the Ising model below the critical temperature the two sets can be choosen as $\Omega_{+} = \{\vec \sigma; \sum_j\sigma_j > 2\epsilon n\}$ and
$\Omega_{-} = \{\vec \sigma; \sum_j\sigma_j < -2\epsilon n\}$.

For free boundary conditions, since $p_{\rm G}(\Omega_+) = \pGibbs(\Omega_-)$ the bottleneck condition is equivalent to $\pGibbs(-2\epsilon n < M < 2\epsilon n ) \xrightarrow[n \to \infty]{} 0$ where $M = \sum_j \sigma_j$ is the total magnetization.

This in turn follows from classic results in large deviation theory of the Ising model (see e.g. Theorem 4 in Ref. \onlinecite{schonmann1987second}) which states that if $T < T_c$ then for any $-M_{\rm eq}(T) < a < b < M_{\rm eq}(T)$ there exist $A_1, A_2, c_1, c_2 > 0$ such that
\begin{equation}
    A_1 e^{-c_1 \abs{\partial \Psi}} \leq \pGibbs(a \leq M(\Psi) \leq b) 
    \leq A_2 e^{-c_2 \abs{\partial \Psi}}
\end{equation}
where $M_{\rm eq}(T)$ is the equilibrium magnetization at temperature $T$, $\Psi$ is any connected set of \emph{sites}, and $M(\Psi) = \sum_{j\in\Psi} \sigma_j$.
To show the bottleneck condition we choose $2\epsilon n < M_{\rm eq}(T)$ and $\Psi$ as the whole system. In this case, $\partial \Psi=\Theta(\sqrt n)$.

\section{The ferromagnetic Ising model on a tree}\label{app:TreeIsing}

A useful reference point, which shares some features of our models while differing from them in other respects, is the ferromagnetic Ising model on a regular tree graph, shown in~\autoref{fig:tree_ising_model}a. This is a problem with a long history~\cite{eggarter1974cayley,muller1974new,matsuda1974infinite,chayes1986mean,bleher1995purity,ioffe1996extremality,mezard2001bethe,mezard2006reconstruction,magan2013memory}, and previous studies have revealed a surprisingly rich structure, which we now summarize. 
Here, following Ref.~\onlinecite{baxter2007exactly}, we use the term \textit{Cayley tree} to refer to the finite tree inclusive of its boundary, which comprises a finite fraction of the total volume, while using the term \textit{Bethe lattice} when referring to properties deep in the bulk (where each vertex sees the same local environment) and sending the boundary to infinity. Transitions occur when even the infinitely distant boundary can influence the bulk, due to sharp changes in the nature of the Gibbs state.

\begin{figure}[t]
    \centering
    \includegraphics{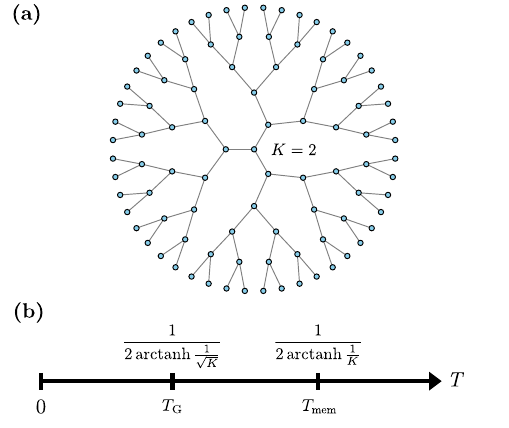}
    \caption{(a) Cayley tree with forward-branching number $K=2$ and $R=5$ generations. Classical Ising variables are placed on each vertex. (b) Phase diagram of the Ising model on the infinite tree as a function of temperature.}
    \label{fig:tree_ising_model}
\end{figure}
\subsection{Phase transitions despite a trivial partition function}

The Ising model on a graph $(V,E)$, with vertices $V$ and edges $E$ is defined by the energy $E_\text{Ising}(\mathbf{\sigma}) = -J\sum_{(i,j) \in E} \sigma_i \sigma_j$, where $\sigma_i = \pm 1$ is an Ising spin on vertex $i \in V$ and we denote an edge $(i,j) \in E$ by its two endpoints. We consider the case of a regular rooted tree graph, with branching ratio $K \geq 2$, such that every vertex has $K+1$ neighbors. In the following, we take $J=1/2$ to match the conventions of~\autoref{eq:Ising}.

As with the redundancy-free LDPC codes considered in this paper, each term in the Ising Hamiltonian can be independently satisfied or violated, since the tree has no closed loops. This means that one can introduce new variables $\tau_{ij} = \sigma_i \sigma_j$, in terms of which the model becomes a paramagnet, with a trivial partition function $\mathcal{Z}(T) = \sum_{\mathbf{\sigma}} e^{-E_\text{Ising}(\mathbf{\sigma}) / T}$ that is a smooth function of the temperature $T$ everywhere~\cite{eggarter1974cayley}. Relatedly, two-point correlation functions decay exponentially at all $T$, as $\langle \sigma_i \sigma_j \rangle = (\tanh{(\beta/2)})^{|i-j|}$ where $\beta = 1/T$ and $|i-j|$ is the graph distance between vertices $i$ and $j$ on the tree (i.e., the length of the unique path connecting them). Despite this apparent triviality on the finite Cayley tree, the model on the infinite Bethe lattice is known to have at least two critical temperatures (\autoref{fig:tree_ising_model}b), whose interpretation is similar to the transitions at \Tglass~and \Tdyn~identified in our models. 

Above \Tdyn, the Gibbs state is \emph{unique}~\cite{bleher1995purity,georgii2011gibbs}: for \textit{any} boundary condition, the bulk behavior is that of a paramagnet. Below \Tdyn, the Gibbs state ceases to be unique, as (uncountably many~\cite{georgii2011gibbs}) different Gibbs states appear. In particular, if we impose ``ferromagnetic'' or ``fully polarized'' BCs, with all spin up (or all down) on the leaves of the tree, the magnetization of a spin deep in the bulk remains finite, signified by stable ferromagnetic fixed points for the recursion from boundary to bulk~\cite{baxter2007exactly}.

However, immediately below \Tdyn, \textit{typical} boundary conditions still flow to the paramagnetic fixed point. In other words, the Gibbs state associated with ``free'' BCs remains pure, down to a second critical temperature \Tglass~\cite{ioffe1996extremality}. Below \Tglass, the free BC Gibbs state ``shatters'' into a mixture of an exponentially large number of different pure states, corresponding to different choices of boundary conditions~\cite{Gandolfo2020}. On a more information-theoretic level, this corresponds to a transition in our ability to reconstruct the value of the central spin based on a knowledge of the boundary spins alone~\cite{evans2000broadcasting}, which has been shown to be related to one-step replica symmetry breaking characteristic of mean-field spin glasses~\cite{mezard2006reconstruction}.

Let us take advantage of the simplicity of the Ising model to cast these two transitions in more familiar thermodynamic terms. First, note that the separation of the two transitions stands in contrast with the thermodynamics of the Ising model on finite-dimensional Euclidean lattices. The latter models possess a single critical temperature below which the Gibbs state is non-unique \textit{and} the free BC is the mixture of the two pure states corresponding to fully polarized boundaries (see \appref{app:decomposition_Ising}). 

Second, an alternative motivation for considering different boundary conditions comes from considering instead a nonzero field.
While the partition function at zero field, $\mathcal{Z}(T)$, is smooth, the partition function on the Cayley tree in the presence of a field, $\mathcal{Z}(T,h) = \sum_{\mathbf{\sigma}} e^{-(E_\text{Ising}(\mathbf{\sigma}) - h \sum_i \sigma_i) / T}$, is only smooth as a function of $h$ at $T > \Tdyn$ and develops a singularity at $h=0$ below this temperature~\cite{eggarter1974cayley,muller1974new}. Taking the thermodynamic limit $n = |V| \to\infty$ first, the limits $h \to \pm 0$ give rise to two different states. They are distinguished from each other by the local magnetization in the bulk of the tree, on vertices close to its center. This can be diagnosed by the one-site susceptibility $\chi_0 \equiv \frac{\partial \langle \sigma_0 \rangle}{\partial h} = \beta \sum_j \langle \sigma_0 \sigma_j \rangle$, where $0$ denotes the central spin. While the individual correlations decay exponentially with distance, there are exponentially many vertices $\propto K^r$ at a distance from the origin. Thus, one has that $\chi_0 \propto (K\tanh({\beta}/2))^n$ which vanishes when $\beta > 1 / \Tdyn = 2 \tanh^{-1}(1/K)$ and diverges when $\beta < 1 / \Tdyn$~\cite{morita1975susceptibility}. In fact, the same is true if we add the magnetic field $h$ \emph{only} on the leaves of the tree. The fixed points of the recursion relation at zero bulk field and polarized boundary conditions are shown in~\autoref{fig:ising}. As in our treatment of \Tdyn~in the Tanner-Ising models in the main text (\autoref{sec:recursive}), the boundary field $h$ sets the initial condition of the recursive flow, $m_0=\tanh(h)$. 

As $T$ is decreased further, a careful study of $\mathcal{Z}(T,h)$ shows that the precise nature the aforementioned non-analyticity at $h=0$ changes at a number of discrete temperatures, but ceases to change below $\Tglass = (2\mathrm{arctanh}(1/\sqrt{K})^{-1}$~\cite{muller1974new}. Below this temperature, in the limit $h \to \pm 0$, the average \textit{magnetization} $\frac{1}{n} \sum_i \langle \sigma_i \rangle$ (which is always dominated by contributions from sites near the boundary) is still zero~\cite{matsuda1974infinite}, but the global \textit{susceptibility} on the Cayley tree, $\chi = \frac{1}{n} \sum_{i,j} \langle \sigma_i \sigma_j \rangle$, diverges as $n\rightarrow\infty$~\cite{matsuda1974infinite,morita1975susceptibility}.

In contrast with the LDPC Tanner codes considered in this work, the Ising model undergoes a continuous transition at both \Tdyn~and \Tglass. Below \Tdyn, the paramagnetic fixed point becomes unstable to infinitesimal ferromagnetic perturbations on the boundaries, i.e., to any boundary distribution with nonzero average magnetization. Similar, below \Tglass, the paramagnetic fixed point becomes unstable to the spin glass fixed point, i.e. to any boundary distribution with nonzero variance~\cite{chayes1986mean}. In contrast, the paramagnetic fixed point for the Tanner-Ising model with local code [7,4,3] code remains stable to both types of perturbations, and ferromagnetic and spin glass order, respectively, are accessed only by sufficiently strong boundary fields.

Another point of contrast is in the relationship between models on the Bethe lattice and locally tree-like graphs containing loops. Whereas the LDPC Tanner codes exhibit spin glass order \textit{without} frustration, the ferromagnetic Ising model on a random regular graph has no glass phase, and in fact has a reconstruction threshold at $T=\Tdyn$, not $T=\Tglass$~\cite{Gerschenfeld2007}. Rather, the tree reconstruction threshold corresponds to the spin glass transition of the $\pm J$ Ising model on a random regular graph, which has both disorder and frustration.
\begin{figure}[t]
\includegraphics[width=\linewidth]{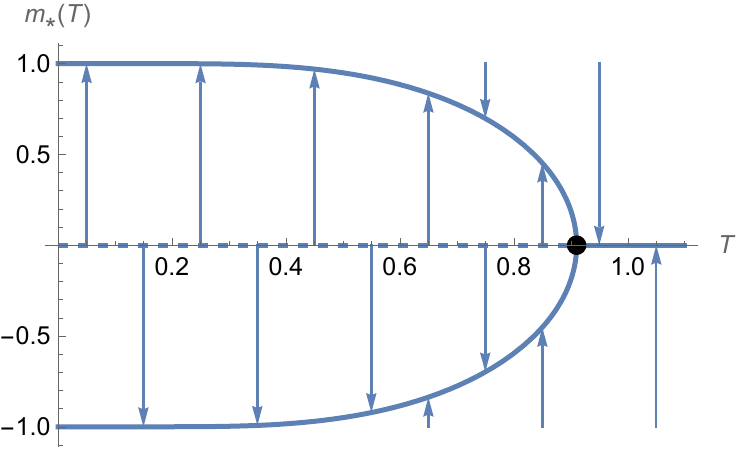}
\caption{Fixed points of the root magnetization under polarized BCs in the Ising model with $K=2$. Arrows indicate the flow under recursion (cf.~\autoref{fig:ferro}).\label{fig:ising} Black point indicates \Tdyn.}
\end{figure}
\subsection{Memory in the Ising model on a tree\label{app:tree_ising_memory}}

It is interesting to consider the relationship of the above transitions to the dynamics of the model. It has been shown~\cite{martinelli2003ising,Martinelli2004} that Glauber dynamics for the Ising model on a tree undergoes a transition at $T = T_\text{SG}$. Above this temperature, Glauber dynamics has a fast mixing time, reaching the global equilibrium state on a time $\mathcal{O}(\log{n})$, while this time-scale becomes polynomial in $n$ below $T_\text{SG}$ (see also Ref. \onlinecite{Berger2005} for a rigorous proof the latter fact for sufficiently large $\beta$).

We present numerical data supporting this picture regarding the dynamics of the bulk magnetization in the Ising model. In particular, we provide evidence for the fact that while the global mixing time changes its dependence on system size only at the spin glass transition \Tglass, this is not true for magnetization of the central site. 
Our physical picture for this is that while the mixing time $\taumix$, which characterizes the relaxation of an arbitrary initial state to equilibrium, scales as $\sim \log n$ at all temperatures above \Tglass, the \emph{mechanism} for this decay is different above and below $\Tdyn$.

Above $\Tdyn$, all sites relax independently with a relaxation time $\tau\in \order{1}$. The global autocorrelation function of any initial state as function of time then decays exponentially as $C(t) \sim e^{-t/\tau}$ and in particular decays to an $\order{1}$ value on a time scale $\taumix\sim\log n$.

For $\Tglass < T < \Tdyn$, the bulk of the tree is ordered under appropriate boundary conditions, but with free boundary conditions the mixing time is still $\taumix\sim\log n$ even when starting from a ferromagnetic initial state. The physical intuition behind this fact however is now less trivial. Since \emph{boundary} spins still relax on an $\order{1}$ time scale and are only at $\log n$ distance from the central spin, domain walls propagating with finite speed inwards from the boundary disorder the whole system on a timescale $\taumix \sim \log n$.

The above picture suggest that while the \emph{mixing time} does not change its system size dependence qualitatively at $\Tdyn$, the memory time of the central spin $\sigma_0$ when initializing the whole system in a uniform state should change from $\tau_0 \sim \order{1}$ above $\Tdyn$ to $\tau_0\sim \order{\log n}$ for $\Tglass < T < \Tdyn$.

\begin{figure}
    \centering
    \includegraphics{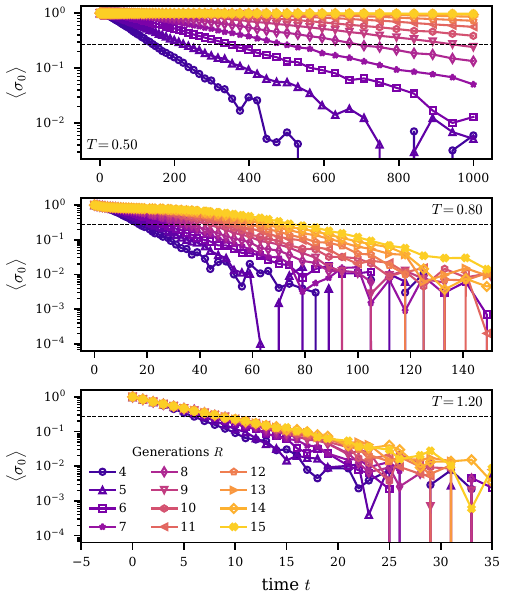}
    \caption{Magnetization of the central site, $\expval{\sigma_0}$, as a function of Monte-Carlo time $t$, when initializing the system in the all-zero state. We show data for three different temperatures, each in one of the three distinct phases (see main text for details). }
    \label{fig:tree_ising_m0t}
\end{figure}

We give numerical evidence for this picture, by computing the expectation value of the central spin, $\expval{\sigma_0}$ as a function of time $t$ when initializing the whole system in the all-zeros state. In \autoref{fig:tree_ising_m0t}, we show this for $K=2$, and a range of system sizes at three different temperatures $T$. The data shown is averaged over $2\times 10^5$ individual traces.
Each temperature is in one of the three phases described above (for $K=2$, we obtain $T_{\rm G} \approx 0.57$ and $\Tdyn \approx 0.91$. At the lowest temperature, $T = 0.50 < T_{\rm G}$, there is clearly a strong system size dependence of the relaxation time, while at the highest temperate, $T = 1.20 > \Tdyn$, there is clearly no such dependence for $R > 7$.
In the intermediate regime, $T=0.80$, there is clearly a weak system size dependence on the time scale that is needed to reach a particular value of the central-spin expectation value. 
However, the slope of the decay of the expectation value at late times is roughly independent of system size. Instead, the expectation value is more long-lived due to an initial plateau, the extend of which increases with system size. This is consistent with the picture draw above, that disordering the bulk happens via domain walls proliferating into the system from the boundaries. 

\begin{figure}
    \centering
    \includegraphics{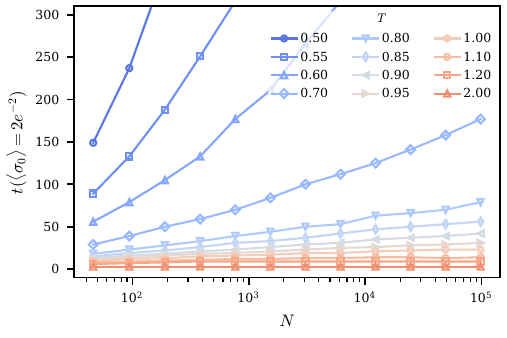}
    \caption{Memory time of the Ising model on the tree as a function of temperature. Shown is the time at which the expectation value of the central spin drops below a threshold (horizontal dashed line in \autoref{fig:tree_ising_m0t}) as a function of system size for different temperatures $T$.}
    \label{fig:tree_ising_tau}
\end{figure}

More quantitatively, we can estimate the memory time of the central spin explicitly, by computing the time at which the average expectation value of the central spin drops below a certain threshold. We here take the threshold to be $2/e^2 \approx 0.27$, which is marked as a vertical dashed line in \autoref{fig:tree_ising_m0t}.
The data is consistent with the fact that below $\Tdyn$, the memory time does scale with the logarithm of the number of spins $N$ (that is it grows linearly in the number of generations $R$). The growth is expected to become super-logarithmically in $N$ (super-linear in $R$) below $T_{\rm G}$.

In summary, our data is consistent with the picture that even in the Ising model on the tree, the three different temperatures have direct consequences for the use of the Ising model as a memory. Special initial conditions (such as all-zeros) have a diverging memory time already below $\Tdyn$. In contrast to the Tanner-Ising models studies in the main text, the Ising model on the tree is not a very \emph{good} memory because of boundary effects.

\section{Construction of Expander Codes\label{app:constructions}}

In this appendix, we explain different constructions of expander codes including, in particular, the models used for the numerical studies in \autoref{sec:numerics}.

\subsection{Gallager codes}

\begin{figure}
    \centering
    \includegraphics{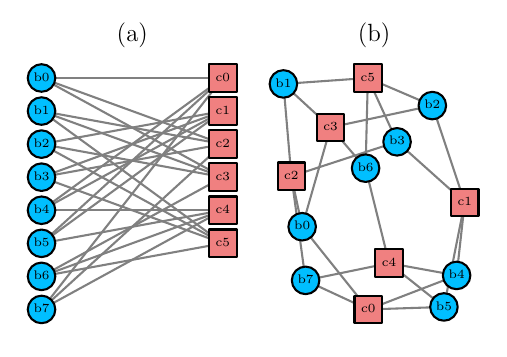}
    \caption{Tanner graph for a Gallager code with $\wbit = 3$, $\wcheck= 4$ and $n=8$. Bits are represented by blue circles and checks by red squares. We draw the same code in (a) a layout separated by bits and checks and (b) in spring layout.}
    \label{fig:gallager_tanner_graph}
\end{figure}

LDPC codes were invented by Robert Gallager in his PhD thesis \cite{gallager1960thesis}. In particular, he introduced the so called $(\wbit, \wcheck)$ random code ensembles, where the degree of every check (the number of bits it is acting on) is $\wcheck$ and the number of checks that involve a given bit (the bit degree) is $\wbit < \wcheck$ for every bit, but otherwise the checks are chosen randomly. 
It is often useful to think of such codes in terms of their \emph{Tanner graph}, which is a bipartite graph where one set of vertices corresponds to bits and another to checks; two vertices are connected only if the check corresponding to one includes the bit corresponding to the other. 
Gallager codes then correspond to $(\wbit,\wcheck)$-regular random bipartite graphs. In \autoref{fig:gallager_tanner_graph}, we show the Tanner graph for a Gallager code with $\wbit = 3$, $\wcheck= 4$ and $n=24$.
One can obtain such graphs, for example, by drawing a random permutation of $n\wbit = m\wcheck$ elements. This defines a bi-partite graph with degree $1$, from which the desired graph can be obtained by combining nodes in groups of $\wbit$ and $\wcheck$ into a single node on the two sides of the graph~\cite{gallager1962low}. 

One can show that the checks generated in this random way are linearly independent with high probability as $n\to\infty$, and hence the code rate approaches $k/n = 1 - \wbit/\wcheck$. In particular
\begin{lemma}[Gallager Codes have no redundancies, Lemma 3.27 in \cite{richardson2008modern}]\label{lem:gallager_no_redundancies}
    Consider the $(\wbit, \wcheck)$ ensemble of classical LDPC codes with $n$ bits, and let
    \begin{equation}
        k_{\rm des} = 
        \begin{cases}
            n\,\left( 1 - \tfrac{\wbit}{\wcheck} \right) - 1 & \text{if}~\wbit~\text{even} \\
            n\,\left( 1 - \tfrac{\wbit}{\wcheck} \right) \phantom{-1} & \text{if}~\wbit~\text{odd} \\
        \end{cases}.
    \end{equation}
    
    Then, for $\wcheck > \wbit \geq 2$, we have for the actual number of logical bits $k := n-\rank H$
    \begin{equation}
        {\rm Prob}(k = k_{\rm des}) \xrightarrow[n\to\infty]{} 1
    \end{equation}
\end{lemma}

Additionally, Gallager codes are with high probability expanding in the sense of \autoref{eq:expansion}~\cite{sipser_spielman1996}. This follows from the fact that random bipartite graphs are unique neighbor expanders

\begin{definition}\label{def:UNexpander}
    A bipartite graph $G=(V=L\uplus R,E)$ with left-degree $\ell$ is a \emph{$(\delta,\gamma)$-unique neighbor expander} if for every $S\subseteq L$ with $|S| \leq \delta |L|$ it holds that $\abs{\Gamma_u(S)} \geq \gamma |S|$, where $\Gamma_u(S) = \lbrace v\in R \mid \exists!\, u\in S : \{ v,u\} \in E \rbrace$ is the set of unique neighbors of $S$.
\end{definition}

We can associate to every biregular graph $G$ with $n = |L| > m=|R|$ a bi-adjacency matrix $\Hcheck_G\in \mathbb{F}_2^{m\times n}$.
It follows that if $G$ is a $(\delta,\gamma)$-unique neighbor expander graph then the LDPC codes defined by~$\Hcheck_G$ is $(\delta,\gamma)$-expanding.
The expansion of Gallager codes then follows from the following theorem.
\begin{theorem}\label{thm:gallagerexp}
    Pick an $(\ell,r)$-biregular graph $G$ uniformly at random and choose $\eta \in \left(0,1-\frac{2}{\ell}\right)$.
    Then it holds with probability at least $1-O\left(n^{-(1-\eta)\ell+1}\right)$ that there exists a $\delta \in (0,\frac{1}{\eta r})$ such that $G$ is a $(\delta,\eta \ell)$-unique neighbor expander.
\end{theorem}
\begin{proof}
    This is a direct consequence of Theorem~8.7 in \cite{richardson2008modern} and
    Lemma~11.3.4 in~\cite{guruswami2019essential}.
\end{proof}

Gallager codes are thus good expander LDPC codes.

\subsection{Sipser-Spielman Codes}

Another construction of families of good LDPC codes which are \emph{Tanner codes} \cite{tanner1981recursive} based on expander graphs, also called Sipser-Spielman Codes. Here, one bit is placed on each edge of a given $s$-regular graph, and each vertex is identified with a ``local code'' defined on $n_{\rm L} = s$ bits, as also illustrated in \autoref{fig:tanner_codes} of the main text. A ``global code'' is then defined by identifying as codewords those configurations of bits such that for each vertex, the bit configuration on the incident edges is a codeword of the local code. To define a family of Tanner codes, we usually consider a family of $s$-regular graphs with increasing size and leave the local code fixed. Note that such a family is naturally LDPC.

The parameters $[n, k, d]$ of the global code are determined by the properties of the graph and the parameters $[n_{\rm L}, k_{\rm L}, d_{\rm L}]$ of the local code. 
For example, the global number of logical bits $k$ can be lower bounded using the number of  checks $m$ of the global code
\begin{equation}
    k \geq n - m = n \left(\frac{2 k_{\rm L}}{n_{\rm L}} - 1 \right).
\end{equation}
A family of global codes thus has finite rate if the local rate  is larger than one half: $r_{\rm L} \equiv \frac{k_{\rm L}}{n_{\rm L}} > \tfrac{1}{2}$.

\subsubsection{Expander Graphs}

Sipser and Spielman \cite{sipser_spielman1996} further showed that the Tanner construction yields \emph{good} LDPC codes ($k \propto n$ and $d \propto n$) if one uses defined the code on an \emph{expander graph}. Informally, a graph is called expanding if it has a nonvanishing surface-to-volume ratio for any subset of vertices. 
There are several ways of formalizing this intuition, for example \emph{edge expansion} which for a given graph $G=(V,E)$ is quantified by the so called Cheeger constant
\begin{equation}\label{eq:cheeger}
    h(G) := \min \left\{\frac{\abs{\partial S}}{\abs{S}} : S \subset V, \abs{S} < |V|/2\right\},
\end{equation}
where $\partial S$ denotes the set of edges connecting a subset of vertices $S$ to its compliment, and $\abs{\bullet}$ denotes the number of elements of a set.
We call a family of graphs with increasing size $|V|$ edge expanding if $h(G)$ remains larger than zero even in the limit $|V| \to \infty$. We discuss several constructions very explicitly below.

While expansion can be hard to prove for a family of graphs, it is easy to check numerically for a set of given examples. This is because the Cheeger constant, remarkably, can be related to the spectral properties of the graph via the so-called Cheeger inequalities
\begin{equation}\label{eq:app:cheeger_bound}
    \tfrac{1}{2} \left( s - \lambda_2 \right) \leq h(G) \leq \sqrt{2s \left(s - \lambda_2 \right)}
\end{equation}
where $\lambda_2$ is the second-largest eigenvalue of the adjacency matrix of $G$ (the largest eigenvalue being $s$). We call a graph with a non-trivial value of $\lambda_2 < s$ a \emph{spectral expander}. The Cheeger inequalities imply that edge expansion implies spectral expansion and vise versa.

\subsubsection{Expander Codes from Spectral Expanders}

One can then prove that the Tanner code is expanding, if the underlying graph is a sufficiently good expander and the small code itself has a sufficiently large distance. In particular, Sipser and Spielman derived a lower bound on the distance of the global code \cite{sipser_spielman1996}
\begin{equation}
    d \geq \frac{(d_{\rm L} - \lambda_2) d_{\rm L}}{(s - \lambda_2)s} n,
\end{equation}
which in particular implies a linear distance if $d_{\rm L} > \lambda_2$\footnote{Note that this already implies that the underlying graph must be expanding, since otherwise $\lambda_2 \to s \geq d_{\rm L}$.}. 
Similarly, one can lower bound the expansion parameters itself
\begin{theorem}[Code Expansion from Spectral Expansion, Theorem 12 in \cite{breuckmann2021balanced}]\label{thm:expansion_sipser_spielman}
    Let $\Hcheck$ be the parity check matrix of a Tanner codes defined on a graph $G$ and with local code $[n_{\rm L}, k_{\rm L}, d_{\rm L}]$, and let $\lambda_2$ be the second-largest eigenvalue of the adjacency matrix of $G$. Then for $\delta > 0$, the global code is $(\delta, \gamma_1\gamma_2)$-expanding with
    \begin{subequations}
    \begin{align}
        \gamma_1 &= \frac{\sqrt{\lambda_2^2 + 4 n_{\rm L}(n_{\rm L} - \lambda_2)\delta}-\lambda_2}{n_{\rm L}(n_{\rm L} - \lambda_2)\delta} \\
        \gamma_2 &= \frac{n_{\rm L}(d_{\rm L} - \lambda_2) - 4\delta(n_L - \lambda_2)}{d_{\rm L}n_{\rm L}}
    \end{align}   
    \end{subequations}
\end{theorem}

Note that the bound on $\gamma$ is nontrivial for some $\delta >0$ exactly if $d_{\rm L} > \lambda_2$.

In summary, the Tanner construction yields a good expander LDPC code if the rate of the local code exceeds one half, and the distance of the local code is larger than the spectral gap of the underlying graph.

\subsubsection{Constructions used for our numerical results}

\begin{table}[]
    \centering
    \begin{ruledtabular}
    \begin{tabular}{ccc|ccc}
     \multicolumn{3}{c}{HGRRG} & \multicolumn{3}{c}{Hyperbolic}\\
n & k & d & n & k & d \\
133 & 19 & 32 & 84 & 12 & 25 \\
763 & 109 & & 252 & 35 & \\
4543 & 649 & & 546 & 78 &  \\
27223 & 3889 & & 672 & 96 & \\
    \end{tabular}
    \end{ruledtabular}
    \caption{Parameters of the codes used in the numerical study of Tanner-Ising models in \autoref{sec:numerics}. Note than in all cases, $k = N_{\rm v} (\tfrac{s}{2} - (n_{\rm L} - k_{\rm L}))$, and hence there are \emph{no} redundancies in the parity check matrix of the global code $\Hcheck_{\rm G}$. 
    Note that only the smallest sizes are amenable to computation of the distance since the computational effort required is exponential in $k$.
    }
    \label{tab:codes}
\end{table}

For the numerical study in \autoref{sec:numerics}, we consider two families of Sipser-Spielman codes, both with the local code chosen to be the $[7, 4, 3]$ Hamming code but defined on two different families of expander graphs. We consider $\{3, 7\}$ tessellations of closed hyperbolic manifolds and degree$-7$ high-girth random-regular graphs with increasing girth $g$ (see below for details on how to construct these graphs). We tabulate the resulting code parameters in \autoref{tab:codes}.

\subsection{Constructing expander graphs\label{app:constructing_expander_graphs}}

\begin{table*}[]
    \centering
    \begin{ruledtabular}
    \begin{tabular}{c c c c | c c c c | c c c c c}
 \multicolumn{4}{c}{HGRRG ($s=7$)}  & \multicolumn{4}{c}{Hyperbolic $\{3, 7\}$} & \multicolumn{5}{c}{LPS}\\
 girth &  $N_v$ &   $N_e$ & $\lambda_2$ 
 & systole & $N_v$ &   $N_e$ & $\lambda_2$ 
 & $p$ & $q$ &   $N_v$ &   $N_e$ & $\lambda_2$ \\
   3 &     38 &    133 & 3.92 &   4 &     24 &     84 & 2.65 &   7 &  11 &   1320 &   5280 & 5.12 \\
   4 &    218 &    763 & 4.67 &   6 &     72 &    252 & 4.76 &   7 &  13 &   2184 &   8736 & 5.00 \\
   5 &   1298 &   4543 & 4.80 &   7 &    156 &    546 & 5.30 &   7 &  17 &   4896 &  19584 & 5.11 \\
   6 &   7778 &  27223 & 4.87 &   7 &    192 &    672 & 5.23 &   7 &  19 &   3420 &  13680 & 5.06 \\
 & & & &   10 &  1404 &   4914 & 6.10 &   7 &  23 &  12144 &  48576 & 5.08 \\
 & & & &   8 &  1536 &   5376 & 6.05 &   7 &  29 &  12180 &  48720 & 5.24 \\
 & & & &   12 &  1740 &   6090 & 6.03 &   7 &  31 &  14880 &  59520 & 5.05 \\
 & & & & & & & &  11 &   7 &    168 &   1008 & 6.00 \\
 & & & & & & & &  11 &  13 &   2184 &  13104 & 6.41 \\
 & & & & & & & &  11 &  17 &   4896 &  29376 & 6.00 \\
 & & & & & & & &  11 &  19 &   3420 &  20520 & 6.27 \\
 & & & & & & & &  11 &  23 &  12144 &  72864 & 6.57 \\
 & & & & & & & &  13 &  11 &   1320 &   9240 & 6.24 \\
 & & & & & & & &  13 &  17 &   2448 &  17136 & 6.11 \\
 & & & & & & & &  13 &  19 &   6840 &  47880 & 6.85 \\
 & & & & & & & &  13 &  23 &   6072 &  42504 & 7.17 \\
 & & & & & & & &  13 &  29 &  12180 &  85260 & 6.76
    \end{tabular}
    \end{ruledtabular}
    \caption{Table of expander graphs referenced in this work. The high-girth random-regular graphs and the regular $\{r, s\}$ tessellations of the hyperbolic plane are used for the numerical simulations in \autoref{sec:numerics}, while the LPS expanders are relevant for a construction in \appref{app:landscape_proof}. For the latter we show all instances with $p \leq 13$ and less than $10^5$ edges. In all cases, we tabulate the second largest eigenvalue of the adjacency matrix, $\lambda_2$, of the instances used in this work.}
    \label{tab:graphs}
\end{table*}

There are many known constructions of expander graphs, both random and explicit. We will here review three of them which are relevant for this work which are (1) a family of random-regular graphs with guaranteed girth (that is the size of the shortest loop) (2) regular tesselations of the hyperbolic plane, and (3) a symmetric construction due to  Lubotzky, Phillips, and Sarnak \cite{lps1988expanders}. We tabulate the instances studied in this work in \autoref{tab:graphs}.

\subsubsection{High-girth random-regular graphs}

Random regular graphs (RRGs) with degree $s\geq 3$ have been shown to be expanders with high probability by Friedman \cite{friedman2003rrg}. We use the construction from Ref. \onlinecite{linial2019rhgrg}, which constructs RRGs with high girth $g$ (the girth is the size of the smallest loop in the graph), modified by some heuristics to improve efficiency.
In particular, given a degree $s \geq 3$, we construct an RRG with girth at least $g$ on 
$N_v = (s-1)^{(g-1)}$ vertices as follows.

Initialize $G = (V, E)$ as the cycle graph with $N_v$ vertices. Then while the minimum degree in $G$, $s_{\rm min}$ is less than $s$ do
\begin{enumerate}
    \item Choose a vertex $v$ randomly among those vertices in $G$ with degree $s_{\rm min}$
    \item Choose a second vertex $u$ among those vertices in $G$ with degree $s_{\rm min}$ and graph distance ${\rm dist}(u, v) > g$
    \item Add the edge $(u, v)$ to $G$
\end{enumerate}

The procedure described above can get stuck in principle, for example if at any point all vertices with minimum degree left in the graph are close.
However, as shown in Ref. \onlinecite{linial2019rhgrg}, the procedure terminates with high probability.

Of course, the family of RRGs constructed here is not exactly the same as the one considered by Friedman in Ref. \onlinecite{friedman2003rrg} and hence the fact that they are expanding is not rigorously established. 
Still, in \autoref{tab:graphs} we show the second eigenvalue of the adjacency matrix, $\lambda_2$, for the examples used in the main text. Although $\lambda_2 > 3 = d_L$ (for the $n=7$ Hamming code), which results in a trivial bound on the expansion parameter $\gamma$, we find that it is bounded away from the degree ($\lambda_2  < s$) and hence the graphs are edge expanders with nonzero Cheeger constant. In fact, for HGRRGs, we even have $\lambda_2 \leq 2\sqrt{s-1}$ meaning that these graphs are \emph{Ramanujan expanders}.

\subsubsection{Regular tessellations of closed hyperbolic manifolds}

It is well known that the perimeter of any connected area scales with the size of the enclosed area in negatively curved space. Regular tessellations of closed hyperbolic manifold are hence expanding, intuitively, due to the negative curvature of the underlying manifold.

Regular tessellations are characterized by their so called Schl\"afli symbol $\{r, s\}$, which denotes that the tiling consists of regular $r$-gons and $s$ such polygons meet at every corner. 
The hyperbolic plane in admits an infinite number of regular tessellations, for any Schl\"afli symbol $\{r, s\}$ as long as $r^{-1} + s^{-1} < \tfrac{1}{2}$.

A detailed introduction to the construction of regular tessellations of closed and orientable hyperbolic manifolds can be found for example in Ref. \onlinecite{breuckmann2016hyperbolic}. In the following, we will simply describe the procedure without motivating its particular details.
We associate with a particular tessellation $\{r, s\}$ of the infinte plane a group of orientation and distance preserving maps
\begin{equation}
    G_{r, s}^+ = \expval{\rho, \sigma \mid p^r=\sigma^s = (\sigma\rho)^2 = 1}
\end{equation}
where $\expval{S\mid R}$ denotes the presentation of a group with generators $S$ and relations $R$. 
Then, to obtain a tessellation of an orientable, closed surface we consider the quotient group $G_{r, s}^+ / H$, where $H$ is a torsion free, normal subgroup of $G_{r, s}^+$.
Note that such subgroups have been tabulated e.g. by Conder \cite{conder_hyperbolic} for the case of triangle groups, from which $G_{r, s}^+$ follows by taking $\rho = xy$ and $\sigma=yz$.
The quotient group $G_{r, s}^+ / H$ can be constructed using standard computer algebra software such as \texttt{GAP}, \texttt{MAGMA}, or \texttt{sage}. To construct a graph, note that vertices, edges, and faces of the tessellation correspond to left cosets of the subgroups $\expval{\sigma}$, $\expval{\rho}$, and $\expval{\rho\sigma}$, respectively. Further, vertices, edges, and faces are incident if and only if their associated cosets share an element.

\subsubsection{Cayley Graphs of Non-Abelian Simple Groups}

Finite simple groups are the fundamental building blocks of all finite groups. Analogous to prime numbers for integers, they cannot be decomposed into smaller normal subgroups. 
They include families like alternating groups, groups of Lie type (such as $\operatorname{PSL}_n$ and classical groups), and sporadic groups.

One can establish a connection between the representation theory of these groups and the expansion of their Cayley graph. This is done via Fourier analysis of the adjacency matrix:
for a Cayley graph $\operatorname{Cay}(G,S)$ with symmetric generating set $S$, the eigenvalues of the adjacency matrix are precisely $\lambda_\chi = \sum_{s\in S} \chi(s)$, where $\chi$ ranges over all irreducible characters of $G$. 
The trivial representation yields the largest eigenvalue $|S|$, while expansion is controlled by the second-largest eigenvalue in absolute value.

The intuition behind many results on Cayley expanders for finite simple groups is that when $S$ is chosen appropriately, the character sums for non-trivial representations remain bounded away from the trivial case. 
The representation-theoretic perspective thus transforms the combinatorial problem of understanding graph structure into the algebraic problem of controlling character sums.
This enables a characterization through tools from the harmonic analysis on finite groups.

In \cite{FSG_expanders} it was shown that there exist constants $k \in \mathbb{N}$ and $\epsilon > 0$ such that every non-abelian finite simple group~$G$ (except possibly Suzuki groups) possesses a generating set $|S|$ of size $k$ making $\operatorname{Cay}(G,S)$ an edge expander with $h(\operatorname{Cay}(G,S))>\epsilon$. 

The case for Suzuki groups was later proved in \cite{suzuki_expanders}.

In fact, expander graphs with \emph{optimal} spectral expansion can be constructed from the finite simple groups $\operatorname{PSL}_2(\mathbb{F}_q)$ as shown in \cite{margulis1988expanders,lps1988expanders}.

\section{Complex Energy Landscapes and Spin Glass Order from Expansion\label{app:complexity}}

In this appendix, we formulate the rigorous versions of, and prove, \cref{thm:landscape_expanders} and \cref{thm:sg_expanders} of the main text. The appendix is designed to be self-contained, at the cost of re-introducing some nomenclature and definitions already presented in the main text.

\subsection{Clustering and Configurational Entropy of the Energy Landscape}

We start by characterizing the complexity of the low \emph{energy landscape}. We do this by analyzing the structure of the configuration space below a given cutoff energy density $\epsilon$:
\begin{equation}\label{eq:omega_eps}
    \Omega(\epsilon) = \{\vec x; \abs{H \vec x} < \epsilon n\}.
\end{equation}
We will show that this set, for certain models at sufficiently low $\epsilon$, can be decomposed into exponentially many disjoint and far-separated ``clusters". Each cluster is associated with a local minimum in the energy landscape: it is impossible to traverse from one cluster to another via a path that flips only (at most) a small fraction of spins at a time without passing through an extensive energy barrier.  
Further, most clusters do not contain ground states, and 
no cluster contains more than an exponentially small fraction of the total number of states below the energy cutoff. Taken together, these features define a complex energy landscape, which will be quantified by the configurational entropy of the cluster decomposition.

To proceed in making these statements precise and rigorous, let us review some definitions, in particular that of expansion, clustering and their relation (from Refs. \onlinecite{anshu2022cnlts, anshu2022nlts}).

\begin{definition}\label{def:expansion}
    Let $\delta,\gamma > 0$.
    We call a binary matrix $\Hcheck\in\mathbb{F}_2^{m\times n}$ \emph{$(\delta,\gamma)$-expanding} if for all $\vec x\in \mathbb{F}_2^{n}$ with $\abs{\vec x} \leq \delta n$ it holds that $\abs{\Hcheck \vec x}\geq \gamma \abs{\vec x}$.

    We call a sequence of binary matrices $\{\Hcheck_i\}_{i\in \mathbb{N}}$ \emph{$(\delta,\gamma)$-expanding} if $\Hcheck_i$ is \emph{$(\delta,\gamma)$-expanding} $\forall i$.
\end{definition}

Note that from \Cref{def:expansion} it immediately follows that the distance of a code defined by a $(\delta,\gamma)$-expanding parity check matrix $\Hcheck$ is larger than $\delta n$.

\begin{definition}\label{def:clustering}
    Let $\epsilon > 0$ and $0<\mu<\nu$.
    We call a binary matrix $\Hcheck\in\mathbb{F}_2^{m\times n}$ \emph{$(\epsilon,\mu,\nu)$-clustering} if for all $x\in \mathbb{F}_2^{n}$ with $\abs{\Hcheck \vec x} \leq \epsilon n$ it holds that either $\abs{\vec x}\leq \mu n$ or $\abs{\vec x}\geq \nu n$.
\end{definition}

\begin{lemma}[Lemma 8 in \cite{anshu2022nlts}, see also Theorem 9 in~\cite{anshu2022cnlts}]\label{lem:clustering}
    If $\Hcheck\in\mathbb{F}_2^{m\times n}$ is $(\delta,\gamma)$-expanding then it is $(\epsilon,\frac{\epsilon}{\gamma},\delta)$-clustering for any $0 < \epsilon < \delta \gamma$.
\end{lemma}
\begin{proof}
    Assume $H$ $(\delta, \gamma)$-expanding and $\abs{H \vec x} < \epsilon n$. If $\vec x < \epsilon/\gamma n$ then we are done, otherwise $\abs{H \vec x} \leq \epsilon n \leq \gamma \abs{\vec x}$. Since $H$ is $(\delta, \gamma)$ expanding this can only be true if $\abs{\vec x} \geq \delta n$.
\end{proof}

One can use the definition above to define an equivalence relation on states such that each state of energy $\leq \epsilon n$ belongs to exactly one equivalence class (the ``clusters'').
In Ref.\ \onlinecite{anshu2022nlts}, the authors show that for states in $\Omega(\epsilon)$ with $\epsilon < \delta\gamma/2$ such a relation is given by
\begin{equation}
    \vec x \sim \vec y 
    ~\Leftrightarrow~
    \abs{\vec x \oplus \vec y} \leq 2 \mu n
\end{equation}

We will define a slightly different relation, and call two states equivalent if they are connected by a sequence of spin flips that flip only a (small) finite fraction of spins at each step.
\begin{definition}\label{def:equiv_xi}
    Consider some subset $\Omega \subset \mathbb{F}_2^n$. Then, 
    $\vec x \sim \vec y$ iff there exists a sequence of states $\{\vec v_i\}_{i=0}^\ell$ such that $\vec x = \vec v_0$, $\vec v_\ell = \vec y$, $\vec v_i \in \Omega$ $\forall i$, and $\abs{\vec v_i \oplus \vec v_{i+1}} < \xi n$ $\forall i$ where $\xi > 0$ is a constant. 
\end{definition}

The above is a valid equivalence relation on any subset of states. 
Reflexivity is obvious ($\vec x\sim \vec x$, choose empty path), symmetry ($\vec x\sim \vec y \implies \vec y\sim \vec x  $) follows from path reversal, and transitivity ($\vec x\sim \vec y \wedge \vec y\sim \vec z  \implies \vec x\sim \vec z  $) follows from path concatenation.

By definition, any equivalence relation on a set $\Omega$ induces a decomposition into disjoint \emph{equivalence classes}\footnote{The classes are disjoint by transitivity.}. The particular equivalence relation defined above further implies that states belonging to different equivalence classes are separated at least by Hamming distance $\xi n$:

\begin{definition}[Cluster decomposition of low-energy states]\label{def:landscape_cluster_decomposition}
    Consider $H\in \mathbb{F}_2^{m\times n}$, an energy cutoff $\epsilon >0$ and an equivalence cutoff $\xi > 0$. 
    We can then decompose the set of low energy states of $H$ into clusters:
    \begin{equation}\label{eq:landscape_cluster_decomposition}
    \Omega(\epsilon) = 
    \biguplus_{j} C_{j}.
    \end{equation}
    where $\Omega(\epsilon)$ is the set of low energy states defined in \autoref{eq:omega_eps}, and we define the \emph{clusters} $C_j$ as equivalence classes under the relation in \cref{def:equiv_xi} with parameter $\xi$.

    By definition, the clusters are separated by extensive distance: $\dist(C_i, C_j) > \xi n$ for $i\neq j$. 
\end{definition}

While the above decomposition is technically valid for all $\xi > 0$, it becomes trivial when the parameter $\xi$ is too large (for example, if $\xi = 1$ there is only a single cluster containing all states). We will thus always have in mind the case of $\xi \ll 1$, but give precise upper bounds in all our statements below. 

For $\xi\ll 1$, the decomposition characterizes which states in $\Omega(\epsilon)$ are connected by moves that flip at most a small fraction of spins at a time. In particular, states in distinct clusters are \emph{not} connected by such moves while staying below the energy density cutoff $\epsilon$. In this sense, clusters can be viewed as local minima of the energy landscape: Going from one cluster to another via a path comprising small moves requires passing through a configuration with energy above $\epsilon n$.

We have already shown that different clusters are separated by an extensive distance; next, in order to show that $\Omega(\epsilon)$ has a complex decomposition with weight on many clusters, we will need to bound the diameter of individual clusters. This can be done, if the energy cutoff $\epsilon$ is sufficiently small, and $H$ is sufficiently expanding:

\begin{lemma}\label{lem:cluster_diameter}
    Let $H\in \mathbb{F}_2^{m \times n}$ and $(\delta, \gamma)$ expanding, and $\vec x, \vec y\in \Omega(\epsilon)$ with $\epsilon < \tfrac{1}{2}\delta\gamma$ and $\vec x \sim \vec y$ under the relation defined in \cref{def:equiv_xi} with $\xi < \frac{1}{2}(\delta - 2\epsilon/\gamma)$. Then $\abs{\vec x \oplus \vec y} \leq 2\epsilon n/\gamma$, i.e. ${\rm diam}(C_j) \leq 2\epsilon n/\gamma$ $\forall\; j$ in \cref{def:landscape_cluster_decomposition}.
\end{lemma}

\begin{proof}
    We prove by contradiction. 
    
    Define for convenience $\mu = 2\epsilon/\gamma$ and suppose $\abs{\vec x \oplus \vec y} > \mu n$ with $\vec x \sim \vec y$.
    Because $H$ by assumption is $(\delta, \gamma)$-expanding and $2\epsilon < \delta\gamma$, $H$ is also $(2\epsilon,\mu,\delta)$ clustering by \cref{lem:clustering} and since $\abs{H(\vec x \oplus \vec y)} \leq \abs{H\vec x} + \abs{H \vec y} < 2\epsilon n $ we have $\abs{\vec x \oplus \vec y} > \delta n$.
    
    Since $\vec x \sim \vec y$, by definition, there exists a sequence $\{\vec v_j\}_{j=0}^{\ell} \subset \Omega(\epsilon)$ such that $\abs{\vec v_i \oplus \vec v_{i+1}} < \xi n < \tfrac{1}{2}(\delta-\mu) n$, $\vec x = \vec v_0$, $\vec v_\ell = \vec y$.

    These two statements together however imply that there exists $0 \leq j \leq \ell$ such that 
    $\mu n < \abs{\vec v_0 \oplus \vec v_j} < \delta n$.
    By clustering of $H$ this means $2\epsilon < \abs{H(\vec v_0 \oplus \vec v_j)} \leq \abs{H\vec v_0} + \abs{H\vec v_j}$ which in turn implies $\vec v_j \notin \Omega(\epsilon)$. This however is in contradiction with $\{\vec v_j\}_{j=0}^{\ell} \subset \Omega(\epsilon)$ and hence in contradiction with $\vec x\sim\vec y$.
\end{proof}

We now turn to quantify the ``complexity'' of the energy landscape. To this end, we define the \emph{configurational entropy}, as well as the \emph{configurational min-entropy} of the set $\Omega(\epsilon)$.

\begin{definition}[Configurational entropy of the landscape]\label{def:landscape_sconf}
    Given $H\in \mathbb{F}_2^n$, and a decomposition of the set of low-energy states as defined in \cref{def:landscape_cluster_decomposition}, we define the \textbf{configurational entropy} of the landscape as
    \begin{subequations}
    \begin{align}
        \mathcal S_{\rm conf}(\epsilon) \equiv& n \,\mathfrak s_{\rm conf} = -\sum_j \mathfrak w_j \log_2 \mathfrak w_j \\
        \mathfrak w_j =& \frac{\abs{C_j}}{\abs{\Omega(\epsilon)}}.
    \end{align}
    We also define the \textbf{configurational min-entropy} of the landscape as
    \begin{equation}
        \mathcal S_{\rm conf}^{(\rm min)}(\epsilon) \equiv n \,\mathfrak s_{\rm conf}^{(\rm min)} = \log_2 \frac{1}{{\rm max}\{\mathfrak w_j\}}
    \end{equation}
    \end{subequations}
    Note that naturally, $\mathcal S_{\rm conf} \geq \mathcal S_{\rm conf}^{(\rm min)}$.
\end{definition}

Both entropies defined above depend on the matrix $H$, the cutoff $\epsilon$, and the parameter of the equivalence relation, $\xi$.
Thinking of the $\mathfrak w_j$ as a probability distribution over the clusters, then $\mathcal{S}_{\rm conf}$ and $\mathcal{S}_{\rm conf}^{(\rm min)}$ are exactly the Shannon and min-entropy of this distribution, respectively.
The former characterizes the number of distinct ``relevant'' components that make up $\Omega(\epsilon)$, while the latter characterizes the size of the largest component.
If the configurational entropy is extensive, i.e. $\mathfrak s_{\rm conf}(\epsilon) > c$ for some $c > 0$, then there are exponentially many (in system size) components of comparable size.
Thus, if the entropy density is larger at finite cutoff $\epsilon > 0$ as it is at $\epsilon = 0$, then almost all clusters contain no ground states. Physically, this means that the energy landscape is dominated by \emph{local} minima, and we will call this \emph{incongruence} below. Note that this terminology is inspired by, but our definition differs from, the idea of incongruence as introduced in Ref. \onlinecite{huse_fisher1987incongruent}.
If the configurational min-entropy is extensive, then no single component contains more than an exponentially small fraction of the total number of states in $\Omega(\epsilon)$. We will call this \emph{shattering} below.

Remarkably, by only assuming sufficiently strong expansion ($\gamma > \gamma^*$) and the fact that $H$ is full rank, we can lower bound the configurational entropy density and show that it is both finite, and an increasing function of the cutoff $\epsilon$. 
Intuitively, sufficiently strong expansion guarantees that the size of individual clusters is small enough (by \cref{lem:cluster_diameter}) such that capturing the total weight of $\Omega(\epsilon)$ requires many clusters.

\begin{theorem}[Complexity of the Energy Landscape]\label{thm:energy_landscape_shattering}
    Consider $H\in \mathbb{F}_2^{m\times n}$ of full rank and $(\delta, \gamma)$-expanding.
    Consider further the decomposition of the set of low-energy states into clusters as defined in \cref{def:landscape_cluster_decomposition}, with energy cutoff $\epsilon < \delta\gamma/2$ and parameter $\xi < \frac{1}{2}(\delta - 2\epsilon/\gamma)$.

    Then there exist positive constants $\gamma^*, \epsilon^*$, $c_1$, $c_2$, such that for $\gamma > \gamma^*$ and $0 < \epsilon < \epsilon^*$ and sufficiently large $n$
    \begin{align}
        \mathfrak s_{\rm conf}^{(\rm min)}(\epsilon) &\geq c_1 > 0\label{eq:app:landscape_shattering}\\
        \mathfrak s_{\rm conf}(\epsilon) &\geq c_2 + \mathfrak s_{\rm conf}(0)  > \mathfrak s_{\rm conf}(0) = r. \label{eq:app:landscape_incongr}
    \end{align}
    where $r = 1-m/n$ is the code rate.
    In this case, we say that the decomposition of the landscape displays \textbf{shattering} \eqref{eq:app:landscape_shattering} and \textbf{incongruence} \eqref{eq:app:landscape_incongr}.
\end{theorem}

\begin{proof}
    We use the fact that the configurational entropy is lower bounded by the min-entropy:
    \begin{equation}\label{eq:sconf_energy_proof1}
        \mathfrak s_{\rm conf} \geq \mathfrak s_{\rm conf}^{(\rm min)} = \frac{1}{n} \log_2\left(\frac{\abs{\Omega(\epsilon)}}{\abs{C_{\rm max}}}\right),
    \end{equation}
    where $C_{\rm max} = {\rm argmax}_{C_j} \abs{C_j}$. 
    \cref{eq:sconf_energy_proof1} can be lower-bounded by upper-bounding $\abs{\Omega(\epsilon)}$ and lower-bounding $\abs{C_j}$.
    By \cref{lem:cluster_diameter}, we have
    \begin{equation}\label{eq:sconf_energy_proof2}
        \abs{C_{\rm max}} \leq \mathcal B_n\left(\frac{2\epsilon}{\gamma}\right)
    \end{equation}
    where $\mathcal B_n(\rho) := \sum_{i=0}^{\rho n}{n\choose i}$ is the volume of a Hamming ball of (relative) size $\rho$.
    Further, since $H$ is full rank, we know that
    \begin{equation}\label{eq:sconf_energy_proof3}
        \abs{\Omega(\epsilon)} = 2^{rn} \mathcal B_m\left(\frac{\epsilon n}{m}\right)
        = 2^{rn} \mathcal B_m\left(\frac{\epsilon}{1-r}\right).
    \end{equation}
    This equality follows because the set of all states at energy density $\epsilon$ are given by the solutions to the linear equation $H\vec x = \vec s$ for all possible right hand sides with $\abs{\vec s} = \epsilon n$, and for each RHS $\vec s \in \mathbb{F}_2^m$ this equation has exactly $2^{rn}$ solutions. Here $r = 1-m/n$ is the code rate, so there are $2^{rn}$ ground states and symmetry sectors. 

   The above two expressions can be bounded using standard bounds for the volume of Hamming Balls \cite{worsch1994_binomial_bounds}. In particular, $\forall \zeta >0~\exists n_0$ such that $\forall n > n_0$
   \begin{align}\label{eq:sconf_energy_proof4}
       \frac{1}{1+\zeta}\frac{1}{\sqrt n} \frac{1}{\sqrt{\rho (1 - \rho)}} \Upsilon(\rho)^n \leq& \nonumber\\
       \sqrt{2\pi} \,\mathcal B_n(\rho)& \nonumber\\
       \leq (1+\zeta) \sqrt n &\sqrt{\frac{\rho}{1-\rho}} \Upsilon(\rho)^n 
   \end{align}
   with $\Upsilon(x) = x^{-x}(1-x)^{1-x}$.

   Plugging the upper bound in \autoref{eq:sconf_energy_proof4} into \autoref{eq:sconf_energy_proof2} and the lower bound in \autoref{eq:sconf_energy_proof4} into \autoref{eq:sconf_energy_proof3}, and using both in \autoref{eq:sconf_energy_proof1} then yields
   \begin{align}\label{eq:thm:landscape_complexity}
       \mathfrak s_{\rm conf} \geq \mathfrak s_{\rm conf}^{(\rm min)} \geq \; & r 
       + (1-r) \log_2 \Upsilon\left(\frac{\epsilon}{1-r}\right)
       - \log_2\Upsilon\left(\frac{2\epsilon}{\gamma}\right) \nonumber\\
       &- O\left(\frac{\log n}{n}\right)
   \end{align}
   which is the desired result by noting that for $\epsilon\to0$ we have $\mathfrak s_{\rm conf}\to r$. 
   It is also easy to see that for sufficiently large $\gamma > \gamma^*$ and sufficiently small $\epsilon < \epsilon^*$, the above lower bound is strictly larger than $r$. Thus, $\mathfrak s_{\rm conf}(\epsilon) > \mathfrak s_{\rm conf}(0) = r \geq 0.$ The value of $\gamma^*$ depends on the rate, and $\epsilon^*$ depends on the rate and $\gamma$.
\end{proof}

The above theorem can be instantiated using Gallager codes, which can be realized with arbitrarily large $\gamma$ (\cref{thm:gallagerexp}), and which are of full rank with high probability  (\cref{lem:gallager_no_redundancies}). Explicitly, we show the bound in \autoref{eq:thm:landscape_complexity} in \autoref{fig:complexity_energy} of the main text, for parameters chosen for a particular family of Gallager codes that realize sufficiently strong expansion.

\subsection{Spin Glass Order from Expansion}

In the previous section, we showed that certain non-redundant expander codes display a complex energy landscape, i.e., the set of configurations below a given energy density cutoff has a complex cluster decomposition with shattering and incongruence. 
In this subsection, we will be interested in properties of the \emph{Gibbs distribution} (also called the \emph{Gibbs state})
\begin{equation}\label{app:gibbs_dist}
    \pG(\vec x) = Z^{-1} \exp(-\beta \abs{H \vec x}),
\end{equation}
where $\beta$ is the inverse temperature and $Z = \sum_{\vec x} \exp(-\beta \abs{H\vec x})$ is the partition function.
We will show that similar to our landscape result, the Gibbs state at low temperature also displays \emph{shattering} and \emph{incongruence}, which together we will call \emph{spin glass order} (\cref{def:spin_glass_order}).  

To obtain these results, we will use the properties of the energy landscape derived above, together with the fact that the Gibbs state at a given temperature is supported almost entirely on states within an energy window of infinitesimal width. 

\subsubsection{Preliminaries}

We begin with a few preliminaries. We start by the following definition:

\begin{definition}[Microcanonical Energy Shell]\label{def:xi_beta}
Given a matrix $H\in\mathbb{F}_2^n$, we define the microcanonical energy shell of width $\omega n$
\begin{equation}\label{eq:xi_beta}
    \Xi_{\omega}(\beta) = \left\{\vec x\in \mathbb{F}_2^n; 
    \abs{\expval{E}_\beta - \abs{H\vec x}} \leq \omega n
    \right\}
\end{equation}
where $\expval{E}_{\beta} \equiv n\expval{\varepsilon}_{\beta} \equiv \sum_{\vec x}\pG(x) \abs{H\vec x}$ is the expectation value of the energy at inverse temperature $\beta$.
\end{definition}

In the absence of redundancies, it is straightforward to show that this set contains all but an exponentially small fraction of the total weight of the Gibbs state:
\begin{lemma}\label{lem:xi_beta_weight}
    Consider $H\in\mathbb{F}_2^{m\times n}$ of full rank. Then the Gibbs state is given by $\pG(\vec x) = Z^{-1} e^{-\beta \abs{H\vec x}}$ with $Z = 2^{n-m} (1 + e^{-\beta})^{m}$ and we have 
    \begin{equation}
        \pG[\Xi_{\omega}(\beta)] \geq 1 - 2 e^{-\zeta n}
    \end{equation}
    where $\Xi_{\omega}(\beta)$ is defined in \cref{def:xi_beta} and $\zeta = \omega^2 n/(2m)$.
\end{lemma}

\begin{proof}
    We compute $Z$ explicitly by noting that
    \begin{align}
        Z &= \sum_{\vec x \in \mathbb{F}_2^n} e^{-\beta \abs{H\vec x}} 
        = 2^k\sum_{\vec s\in \mathbb{F}_2^m} e^{-\beta \abs{\vec s}} \nonumber\\
        &= 2^{n-m} \prod_{j=1}^m (1 + e^{-\beta})
        = 2^{n-m} (1 + e^{-\beta})^{m}.
    \end{align}
    where in the first line we have used that $H$ has full rank.
    The inequality in the theorem follows directly from Hoeffding's inequality: since if $H$ is full rank, the energy $E$ is a sum of i.i.d.\ random variables.
\end{proof}

Now, given any decomposition of the configuration space, and an energy width $\omega n$, we can separate elements of the decomposition into \emph{typical} and \emph{atypical} elements.

\begin{definition}[Typical and atypical subsets]\label{def:atypical_clusters}
    Consider $\Hcheck\in \mathbb{F}_2^{m \times n}$, $\beta, \omega > 0$, and a decomposition of the configuration space into disjoint subsets of the form 
    \begin{equation}
        \mathbb{F}_2^n = \biguplus_{j} \Omega_j.
    \end{equation}
    Given an inverse temperature $\beta$ and energy width $\omega n$, we call an element of this decomposition, $\Omega_j$ \emph{typical} if
    \begin{equation}
        \frac{\pG[\Omega_j \cap \Xi_{\omega}(\beta)]}{\pG[\Omega_j]} \geq \frac{1}{2}.
    \end{equation}
    and we call $\Omega_j$ \emph{atypical} if
    \begin{equation}
        \frac{\pG[\Omega_j \cap \Xi_{\omega}(\beta)]}{\pG[\Omega_j]} < \frac{1}{2}.
    \end{equation}
\end{definition}

In other words, less than half the weight of an atypical subset is supported on configurations in the microcanonial energy shell $\Xi_{\omega}(\beta)$. Thus, atypical subsets $\Omega_j$ are those that mostly contain states at energies far above or far below the average at a given temperature. 

The name `atypical' is justified in the sense that the total weight of \emph{all} atypical clusters  is exponentially small in system size. 
Note, however, that the above definition of typicality depends on a specific choice of $\omega$. Below we will work with fixed $\omega$, but always have the limit $\omega \to 0$ in mind. Note that the order of limit is important: for the tail bounds above to be useful, the ``thermodynamic limit'' means that we take $n \to\infty$ first, and then $\omega\to0$.

\begin{lemma}[Atypical subsets carry almost no weight]\label{lem:atypical_clusters}
    Consider $H\in \mathbb{F}_2^n$ of full rank, and a decomposition of configuration space of the form $\mathbb{F}_2^n = \biguplus_{j} \Omega_j$
    Then, for all $\beta > 0$ and $\omega > 0$ atypical subsets as defined in \cref{def:atypical_clusters} carry only an exponentially small fraction of the Gibbs weight
    \begin{equation}
        \sum_{{\rm atypical}\, j} \pG[\Omega_j] \leq 4 e^{-\zeta n}
    \end{equation}
    where $\zeta = \omega^2 n/(2m)$.
\end{lemma}

\begin{proof}
    We have for any $j$ such that $\Omega_j$ is atypical
    \begin{align}
        1 &= \frac{\pG[\Omega_j \cap \Xi_{\omega}(\beta)]}{\pG[\Omega_j]}
            + \frac{\pG[\Omega_j \cap \Xi_{\omega}^c(\beta)]}{\pG[\Omega_j]} \\
        &< \frac{1}{2} + \frac{\pG[\Omega_j \cap \Xi_{\omega}^c(\beta)]}{\pG[\Omega_j]}
    \end{align}
    where $\Xi_{\omega}^c(\beta) = \mathbb{F}_2^n / \Xi_{\omega}(\beta) $ denotes the complement of $\Xi_{\omega}(\beta)$. We have used that $\Omega_j$ is atypical in the second line.
    This implies directly that
    \begin{align}
        \pG[\Omega_j] < 2\,\pG[\Omega_j \cap \Xi_{\omega}^c(\beta)].
    \end{align}
    Now note that
    \begin{align}
        \sum_{{\rm atypical}\, j} \pG[\Omega_j] &< 2 \sum_j \pG[\Omega_j \cap \Xi_{\omega}^c(\beta)] \\
        &\leq 2\,\pG[\Xi_{\omega}^c(\beta)] \\
        &\leq 4 e^{-\zeta n}
    \end{align}
    where in the second line we have used that the clusters are disjoint, and in the third line we used \cref{lem:xi_beta_weight}.
\end{proof}

\subsubsection{Gibbs State Decomposition\label{app:gibbs_decomposition_precise}}

In the spirit of the Gibbs state decomposition described in \autoref{sec:gibbs_decomposition} of the main text, we now define the decomposition of the Gibbs state into components more formally.

As a motivation, recall first the classical bottleneck theorem
\begin{theorem}[Classical Bottleneck theorem]
    Let $M:\chi\times\chi\to[0,1]$ be a Markov generator with configuration space $\chi$ with steady state $\pi:\chi\to[0,1]$.
    Let $\chi = A \uplus B_1 \uplus B_2 \uplus C$ be a partition such that $M(\vec x, \vec y) = 0$ if $\vec x \in B_2\uplus C$ and $\vec y\in A$ or if $\vec x \in A\uplus B_1$ and $\vec y \in C$. Then the restricted steady state $\pi^{(A)}(\vec x) = \pi(A)^{-1}\indicator_A(\vec x) \cdot \pi (\vec x)$ is an approximate steady state of $M$:
    \begin{equation}
        \sum_{\vec x \in \chi} \abs{(M\pi^{(A)})(\vec x) - \pi^{(A)}(\vec x)} \leq 2 \frac{\pi(B)}{\pi(A)}
    \end{equation}
    where $(M\pi^{(A)})(\vec x) = \sum_{\vec y\in \chi}M(\vec x, \vec y)\pi_A(\vec y)$ and $B = B_1\uplus B_2$.
\end{theorem}

Intuitively, the above states that for a given dynamics with steady state $\pi$, the restriction of $\pi$ to a subset of configuration that is ``surrounded'' by a region of low (relative) weight is an approximate steady state of the dynamics, if the dynamics cannot ``skip'' over the region of low weight.
The general statement is well known, but we use here a slightly unusual formulation to emphasize the fact that the restrictions of the steady state are approximate steady states of the dynamics if they are surrounded by a bottleneck. For a proof of the above version, see Theorem SI.1 of Ref.~\onlinecite{rakovszky2024bottleneck}, and we refer interested readers to Ref.~\onlinecite{levin2017markov} for a general introduction. 

The bottleneck theorem informs our definition of Gibbs state components below. In particular, consider a region of configuration space $\Omega$ and its $\eta$-boundary defined as  $\partial_\eta\Omega \equiv \{\vec x \in \Omega^c; \dist(\vec x,\Omega) \leq \eta n \}$ [where $\dist(\vec x,\Omega) = \min_{\vec y \in \Omega}(\abs{\vec x\oplus \vec y})$]. Then if
\begin{equation}
    \frac{\pG[\partial_{\eta}\Omega]}{\pG[\Omega]} = \Delta
\end{equation}
then the restricted Gibbs state $\pG^{(\Omega)}$ has a lifetime of $\Delta^{-1}$ under \emph{any} local dynamics (flipping less than $\eta n/2$ variables at a time) that has the Gibbs state as its unique steady state. 
To see this is, identify $A = \Omega$, $B_1 = \partial_{\eta/2}\Omega$ and $B_2 = \partial_{\eta}\Omega / B_1$ above.

For a family of Gibbs distributions $p_{{\rm G}, n}$, one can then define a Gibbs state decomposition into components via a sequence of decomposition of configuration space into regions surrounded by bottlenecks, and demanding that $\Delta(n)\to 0$ as $n\to 0$ above.
In the concrete definition of the Gibbs states decomposition below, we \emph{almost} follow the definition in Chapter 22.1 of Ref. \onlinecite{mezard2009information}, with the slight addition that we allow for a junk component $\Lambda$ carrying some vanishing fraction of the weight. 
This allows the components $\Omega_j$ to only contain ``typical'' states (see also \cref{lem:gibbs_decomposition_typical} below), and we can absorb atypical regions of configuration space, e.g. the bottlenecks (c.f. \autoref{fig:Gibbs_states}), into the junk component (which we will not demand to be surrounded by a bottleneck): 

\begin{definition}[Decomposition of the Gibbs state]\label{def:gibbs_decomposition}
    Consider a sequence of Gibbs states $p_{{\rm G}, n}:\chi^n\to[0,1]$ and associated configuration spaces $\chi^n$. For each $n$, consider a decomposition of the configuration space
    \begin{equation}
        \chi^n = \biguplus_{j} \Omega_{j, n} \uplus \Lambda_n
    \end{equation}
    such that
    
    \begin{enumerate}
    \item Each $\Omega_j$ is surrounded by a bottleneck: there exists $\eta > 0$, and a function $\Delta(n)$ with $\Delta(n)\xrightarrow[n\to\infty]{}0$ such that 
        \begin{equation}\label{eq:app:bottleneck}
        \frac{\pG[\partial_{\eta}\Omega_{j, n}]}{\pG[\Omega_{j, n}]} \leq \Delta(n)
        \end{equation}
        where the $\eta$-boundary of a subset $\Omega$ is defined as  $\partial_\eta\Omega \equiv \{\vec x \in \Omega^c; \dist(\vec x,\Omega) \leq \eta n \}$.
    \item The set $\Lambda$ contains only a vanishing fraction of the weight
        \begin{equation}
            p_{{\rm G}, n}(\Lambda_n) \xrightarrow[n\to\infty]{}0.
        \end{equation}
        We call $\Lambda$ the ``junk'' set.
    \end{enumerate}

    \noindent  Gibbs state components are then defined as the (normalized) restrictions of $\pG$ to  $\Omega_j$:
    \begin{equation}
       p_{{\rm G}, n}^{(\Omega_{j, n})} = \frac{\indicator_{\Omega_{j, n}} \cdot p_{{\rm G}, n}}{p_{{\rm G}, n}(\Omega_{j, n})}
    \end{equation}
    where $\indicator_\Omega$ is the indicator function of the set $\Omega$.
    The Gibbs distribution can now be written the convex sum of its components
    \begin{equation}\label{eq:def:gibbs_decomposition}
        p_{{\rm G}, n} = \sum_{j=1}^{M_n} w_{j, n}\,p_{{\rm G}, n}^{(\Omega_{j, n})} + w_{\Lambda_n}\,p_{{\rm G}, n}^{(\Lambda_n)}
    \end{equation}
    where $w_{j, n} \equiv p_{{\rm G}, n}(\Omega_{j, n})$ and $w_{\Lambda_n}\equiv p_{{\rm G}, n}(\Lambda_n)$.
    We then call \autoref{eq:def:gibbs_decomposition} a \textbf{Gibbs state decomposition}.
\end{definition}

Note that the above decomposition is not unique, since some subsets of $\chi$ that have small weight (e.g. part of the boundary) may be assigned to different components or the junk component without invalidating the bottleneck condition. 

One can define the components to be \emph{extremal} (sometimes also called \emph{pure}) by demanding that they cannot be further decomposed while still fulfilling the bottleneck condition [\autoref{eq:app:bottleneck}] above.

In the following, to avoid overly cluttered notation, we will sometimes drop the explicit subscripts indicating the dependence of quantities and subsets on $n$, but we will always have sequences of matrices and sequences of decompositions in mind.

Given the decomposition defined above, the weights $\{w_{j, n}\}$ define a distribution over components of the Gibbs states. We then define two notions of configurational entropy to characterize the Gibbs state. The first, defined as the Shannon entropy of the weights, characterizes the number of relevant components contributing to $\pG$. The second corresponds to an upper bound on the size of the largest weight.

\begin{definition}[Configurational Entropy of the Gibbs state]\label{def:sconfig_gibbs}
    Given a decomposition of the Gibbs state as defined in \cref{def:gibbs_decomposition}, we define the configurational entropy as
    \begin{equation}
        \Sconf \equiv n \sconf \equiv -\sum_j w_{j, n} \log_2 w_{j, n} - w_{\Lambda_n}\log_2 w_{\Lambda_n}.
    \end{equation}
    We also define the \textbf{configurational min-entropy} of the decomposition as
    \begin{equation}
        \Sconf^{(\rm min)}(\epsilon) \equiv n \,\sconf^{(\rm min)} = \log_2 \frac{1}{{\rm max}\{ w_{j, n}\}}
    \end{equation}
    Note that naturally, $\Sconf \geq \Sconf^{(\rm min)}$.
\end{definition}

Before showing spin glass order from expansion, we prove one more statement about the Gibbs state decomposition. Given any decomposition, we can define a new valid decomposition, such that all components are supported on typical sets (see \cref{def:atypical_clusters}), and the configurational entropy of the old decomposition is lower-bounded by that of the new one. 

\begin{lemma}\label{lem:gibbs_decomposition_typical}
    Consider a sequence of binary matrices $H_n\in \mathbb{F}_2^n$ of full rank, with a decomposition of the Gibbs state as in $\cref{def:gibbs_decomposition}$
    \begin{equation}
        \mathbb{F}_2^n = \biguplus_{j} \Omega_{j, n} \uplus \Lambda_n
    \end{equation}
    Then, $\forall \omega > 0$ there exists another valid decomposition
    \begin{equation}
        \mathbb{F}_2^n = \biguplus_{j} \Omega'_{j, n} \uplus \Lambda'_n
    \end{equation}
    such that all $\Omega'_{j}$ are typical subsets in the sense of \cref{def:atypical_clusters}, and 
    \begin{equation}
         \sconf \geq \sconf'
    \end{equation}
    for some $\zeta > 0$.
\end{lemma}

\begin{proof}
    The idea is to absorb all atypical components into the definition of $\Lambda'$, that is
    \begin{equation}
        \Lambda' = \Lambda \bigcup_{{\rm atypical}\,j} \Omega_j.
    \end{equation}
    Because of \cref{lem:atypical_clusters}, this does not change the fact that $\Lambda'$ has vanishing weight. Furthermore, since the remaining clusters are a subset of the original ones, $\{\Omega_i'\}_i \subset \{\Omega_j\}_j$, all $\Omega_i'$ are surrounded by a bottleneck. Hence the new decomposition is also valid. 

    What is left to show it that the configurational entropy can only decrease.
    To this end, for some choice of energy width $\omega > 0$, split the sum into typical and atypical clusters (\cref{def:atypical_clusters})
    \begin{align}
        \Sconf =& -\sum_{{\rm typical}\,j} w_j \log_2 w_j - \sum_{{\rm atypical}\,j} w_j\log_2 w_j \nonumber\\  &- w_{\Lambda}\log_2 w_{\Lambda} \\
        \geq& -\sum_{{\rm typical}\,j} w_j \log_2 w_j \nonumber\\
        &- \left(w_{\Lambda} + \sum_{{\rm atypical}\,j} w_j\right)
        \log_2\left(w_{\Lambda} + \sum_{{\rm atypical}\,j} w_j\right) \\
        =& -\sum_{{\rm typical}\,j} w_j \log_2 w_j + w_{\Lambda'}\log_2 w_{\Lambda'}. \\
        =& \Sconf'
    \end{align}
    where we have used that summarizing previously separated components into one can only lower the entropy of the distribution.
\end{proof}

\subsubsection{Spin Glass Order from expansion}

With these definitions in mind, we can now inspect the properties of the Gibbs state decomposition for LDPC codes (defined in terms of a parity check matrix $H\in \mathbb{F}_2^n$) with expansion. We will show that for codes without redundancies ($H$ is full rank) and sufficiently strong expansion ($\gamma > \gamma^*$ for some constant $\gamma^*$ that depends on the rate), there exists a decomposition of the Gibbs state as defined in \cref{def:gibbs_decomposition} such that the configurational entropy density is both finite at low $T$, and also strictly larger than the zero-temperature limit. 
We will take these features to be our definition of spin glass order.
\begin{definition}[Spin Glass Order]\label{def:spin_glass_order}
    Given a family of binary matrices, $H\in\mathbb{F}_2^n$, we say that it realizes \textbf{spin glass order} at inverse temperature $\beta$, if there exists positive constants $c_1, c_2, n^*$ and a Gibbs state decomposition [\cref{def:gibbs_decomposition}], such that for all $n > n^*$
    \begin{align}
        \sconf^{(\rm min)}(\beta) &\geq c_1 > 0, \label{eq:app:gibbs_shatter}\\
        \sconf(\beta) &\geq c_2 + \sconf(0) > \sconf(0) \geq 0.\label{eq:app:gibbs_incongr}
    \end{align}
    We then say that the Gibbs state displays \textbf{shattering} \eqref{eq:app:gibbs_shatter} and \textbf{incongruence} \eqref{eq:app:gibbs_incongr}, respectively.
\end{definition}

Note that with the order of limits $n\to\infty$ at finite $\beta$, then $n \sconf(0) = \Sconf(0)$ is just the logarithm of the number of ground states of the model. The constants $c_1$, $c_2$ guarantee that for all $n$ the configurational (min-)entropy densities stay bounded away from zero and the zero-temperature limit, respectively. (as opposed to something like $\sconf(\beta) = \sconf(0) + 1/n \geq \sconf(0)$).
Informally speaking, our definition of spin glass order demands that (i) no component carries more than an exponentially small fraction of the weight, and (ii) that the number of ``relevant'' Gibbs state components is (strictly) exponentially larger than the number of ground states of the model. This means that almost no component can contain ground states, and provides an intuitive picture for the hardness of reaching the ground state via annealing under local dynamics. 

Let us now show that the above scenario is realized in certain models with expansion. We begin by showing that for sufficiently small energy cutoffs $\epsilon < \delta\gamma/2$, and low temperatures such that $\expval{\varepsilon}_{\beta} < \epsilon$,  the cluster decomposition of $\Omega(\epsilon)$ (cf. \cref{def:landscape_cluster_decomposition}) can be used to obtain a valid Gibbs state decomposition in the sense of \cref{def:gibbs_decomposition}.

\begin{lemma}[Decomposition of the Gibbs State from Expansion]\label{lem:gibbs_decomposition_exp}
    Consider a sequence of binary matrices $H_n\in \mathbb{F}_2^{m\times n}$, of full rank and $(\delta, \gamma)$ expanding.  
    Then there exists constants $\omega^* > 0$, $\beta^* < \infty$, and a decomposition of configuration space
    \begin{align}
        \mathbb{F}_2^n = \biguplus_{j}\Omega_j \uplus \Lambda
    \end{align}
     such that for all $0 < \omega < \omega^*$ and $\beta > \beta^*$:
    \begin{enumerate}
        \item All $\Omega_j$ are typical subsets (\cref{def:atypical_clusters}),
        \item ${\rm diam}[\Omega_j\cap \Xi_{\omega}(\beta)] \leq  2 n (\expval{\varepsilon}_\beta + \omega)/\gamma$ with $\Xi_{\omega}(\beta)$ the microcanonical shell [\autoref{eq:xi_beta}].
        \item The $\{\Omega_j\}$ and $\Lambda \equiv \mathbb{F}_2^n / (\uplus_{j=1}^M\Omega_j)$ define a decomposition of the Gibbs state in the sense of \cref{def:gibbs_decomposition}.
        \item $\pG(\Lambda) < 4 e^{-\zeta n}$ for $\zeta = \omega^2 n /(2m)$ and sufficiently large $n$.
    \end{enumerate} 
\end{lemma}

\begin{figure}
    \centering
    \includegraphics{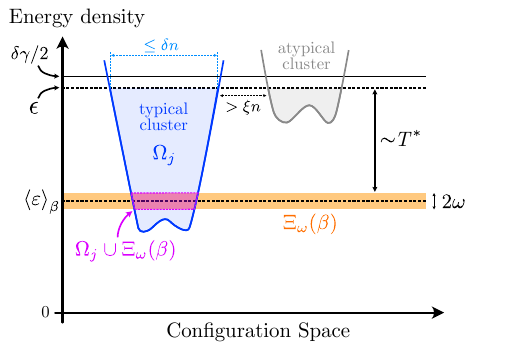}
    \caption{Sketch of the various energy scales and distances involved in the Gibbs state decomposition from expansion.}
    \label{fig:sgo_proof_bounds}
\end{figure}

\begin{proof}
    To aid understanding, we sketch the various quantities and choices of parameters involved in the proof in \autoref{fig:sgo_proof_bounds}.
    
    Consider the cluster decomposition of the set of low energy-density states $\Omega(\epsilon)$ with $\expval{\varepsilon}_{\beta} < \epsilon < \delta\gamma/2$ as defined in \cref{def:landscape_cluster_decomposition}, with $\xi < \tfrac{1}{2}(\delta - 2\epsilon / \gamma)$. That is, we have a decomposition of the total configuration space of the form
    \begin{equation}
        \mathbb{F}_2^n = \Omega(\epsilon)^c \uplus \biguplus_j C_j
    \end{equation}

    We now distinguish the clusters $C_j$ by whether they are typical or atypical using the definition of \cref{def:atypical_clusters}. Note that if $\epsilon > \expval{\varepsilon}_{\beta}$, then for all $\omega < \omega^* = \epsilon - \expval{\varepsilon}_{\beta}$, $\Omega(\epsilon)^c$ is an atypical set since $\Omega(\epsilon)^c\cap\Xi_{\omega} = \emptyset$. Then, we define the junk set $\Lambda$ as the union of all atypical clusters and $\Omega(\epsilon)^c$
    \begin{equation}
        \Lambda = \Omega(\epsilon)^c \cup \bigcup_{{\rm atypical}\,j} C_j 
    \end{equation}
    and by \cref{lem:atypical_clusters}, $\pG(\Lambda) < 4 e^{-\zeta n}$, which shows (4).

    Our decomposition of the configuration space that instantiates the above lemma is then given by 
    \begin{equation}
        \mathbb{F}_2^n = \Lambda \uplus \biguplus_{{\rm typical}\, j} C_j
    \end{equation}
    which shows (1). From here on, we will denote these typical $C_j$ by $\Omega_j$. 

    Next, we bound the diameter of the clusters. For this note that by our choice of $\xi < \tfrac{1}{2}(\delta - 2\epsilon / \gamma)$, by \cref{lem:cluster_diameter} the diameter of any $\Omega_j$ is bounded by $2\epsilon n / \gamma < \delta n$. Thus, for any two states $\vec x$, $\vec y \in \Omega_j$ we have $\abs{\vec x \oplus \vec y} < \delta n$. 
    Now consider any $\vec x, \vec y \in \Omega_j \cap \Xi_{\omega}(\beta)$. In this case we get from the triangular inequality and expansion
    \begin{align}
        \abs{H\vec y} + \abs{H\vec x} \geq \abs{H(\vec x \oplus \vec y)} \geq \gamma \abs{\vec x \oplus \vec y}.
    \end{align}
    Now suppose for contradiction that $\abs{\vec x \oplus \vec y}> 2 n (\expval{\varepsilon}_{\beta} + \omega)/ \gamma$. Then
    \begin{align}
        \abs{H\vec y} >& 2 n (\expval{\varepsilon}_{\beta} + \omega) - \abs{H \vec x} \\
        \geq& 2 n (\expval{\varepsilon}_{\beta} + \omega) - n(\expval{\varepsilon}_{\beta} + \omega) \\
        \geq& n (\expval{\varepsilon}_{\beta} + \omega)
    \end{align}
    which implies that $\vec y\notin \Xi_{\omega}(\beta)$, a contradiction. In conclusion, we have shown (2):
    \begin{equation}
        {\rm diam}[\Omega_j \cap \Xi_{\omega}(\beta)] \leq 2 n (\expval{\varepsilon}_{\beta} + \omega)/ \gamma.
    \end{equation}

    Finally, let us show the bottleneck condition (property (3)) is fulfilled for $\beta > \beta^*$. To this end, we note that the clusters are separated by extensive distance by definition. Choosing $\eta < \xi$, we can then write
    \begin{align}
        \frac{\pG[\partial_{\eta} \Omega_j]}{\pG[\Omega_j]} 
            &= \frac{
            \sum_{\vec x\in\partial_{\eta} \Omega_j} e^{-\beta\abs{\vec x}}}
                {\sum_{\vec y \in \Omega_j}e^{-\beta\abs{\vec y}}} \\
            &\leq\frac{ 2^n  e^{-\beta \epsilon n}}{e^{-\beta(\expval{\varepsilon}_{\beta} + \omega) n}} \\
            &= \exp(-n \left( \beta\left(\epsilon- \expval{\varepsilon}_{\beta} - \omega\right)- \ln 2\right) ).
    \end{align}
    In the second line, we used in the numerator that $\abs{\partial_{\eta} \Omega_j} < 2^n$ and also that for $\eta < \xi$,  $\partial_{\eta} \Omega_j \subset \Omega(\epsilon)^c$. In the denominator, we used that since $\Omega_j$ is typical, it contains at least one state in the microcanonical shell and the probability of this state lower bounds $\pG[\Omega_j]$.
    We then arrive at the last line. Note that it is independent of the cluster $j$ and goes to zero for $n\to\infty$ since by our choice of parameters above $0 < \omega < \omega^* = \epsilon - \expval{\varepsilon}_{\beta}$. Choosing $\beta^*$ such that $\beta^* = \ln 2 / (\omega^* - \omega)$ then concludes the proof.
     
\end{proof}

Note that the decomposition established in \cref{lem:gibbs_decomposition_exp} is not guaranteed to be a decomposition into \emph{extremal} states in the sense introduced before and in \autoref{sec:gibbs_decomposition} of the main text. However, as also mentioned before, the configurational entropy defined in \cref{def:sconfig_gibbs} can only grow if clusters are subdivided further. Because of this, the above decomposition will suffice to show the existence of spin glass order at low temperature.

We are now finally able to prove our lower bound of the configurational entropy of the Gibbs state from expansion.

\begin{theorem}[Spin Glass order from Expansion]
    Consider a family of binary matrices $H_n$ of full rank, $(\delta, \gamma)$-expanding, and the decomposition of the Gibbs state as defined in $\cref{lem:gibbs_decomposition_exp}$.
    
    Then, there exist positive constants $\gamma^*, T^*$ such that if $\gamma > \gamma^*$ and for $T < T^*$, the family realized spin glass order in the sense of \cref{def:spin_glass_order}.
\end{theorem}

\begin{proof}

Consider the decomposition of the Gibbs state in \cref{lem:gibbs_decomposition_exp}.
We can lower bound the configurational entropy by upper bounding the size of the largest component, which will also yield the necessary lower bound of the configurational min-entropy. Note that
\begin{align}
    \Sconf =& -\sum_{j} w_j \log_2 w_j - w_{\Lambda}\log_2 w_{\Lambda} \\
    \geq& -\log_2 w_{\rm max} + O(n e^{-\zeta n})
\end{align}
where $w_{\rm max} \equiv {\rm max}(\{w_j\})$ does not include the junk component, but we have already shown that the weight of that is exponentially small. 
Thus, if we can show that $-\log_2 w_{\rm max}$ is lower-bounded away from $r$, then both \autoref{eq:app:gibbs_shatter} and \autoref{eq:app:gibbs_incongr} follow.

We take $\Omega_{{\rm max}\,w}$ to be a component with weight $w_{\rm max}$.
For this we can use that, by definition, all $\Omega_j$ are typical subsets in the sense of \cref{def:atypical_clusters} and hence have large overlap with $\Xi_{\omega}(\beta)$:
\begin{align}
    w_{\rm max} &= \pG[\Omega_{{\rm max}\,w}] \\
    &< 2\,\pG[\Omega_{{\rm max}\,w} \cap \Xi_{\omega}(\beta)] \\
    &\leq 2 Z^{-1} \sum_{\vec x \in \Omega_{{\rm max}\,w} \cap \Xi_{\omega}(\beta)} e^{-\beta \abs{H \vec x}} \\
    &\le 2 Z^{-1} \, \abs{\Omega_{{\rm max}\, w} \cap \Xi_{\omega}(\beta)} e^{-\beta n (\expval{\varepsilon}_{\beta} - \omega)}
\end{align}
Now we use that also by \cref{lem:gibbs_decomposition_exp}, we have
\begin{equation}
    \abs{\Omega_{{\rm max}\, w} \cap \Xi_{\omega}(\beta)} \leq \mathcal B_n\left(\frac{2(\expval{\varepsilon}_{\beta} + \omega)}{\gamma}\right).
\end{equation}
and use the upper bound on the size of a Hamming ball in \autoref{eq:sconf_energy_proof4}, together with the exact expression of the partition function derived in \cref{lem:xi_beta_weight} to obtain our final upper bound for $w_{\rm max}$: $\forall \omega>0$
\begin{align}\label{eq:thm:pmax_upper_bound}
    \frac{1}{n}\log_2 \frac{1}{w_{\rm max}} \geq& 
    r + (1-r)\log_2(1 + e^{-\beta}) \nonumber\\
    &+ \beta(\expval{\varepsilon}_{\beta}-\omega)\log_2 e \nonumber\\
    &- \log_2\Upsilon\left(\frac{2(\expval{\varepsilon}_{\beta} + \omega)}{\gamma}\right) \nonumber\\
    &- O\left( \frac{\log n}{n} \right)
\end{align}

Notably,  for $\omega\to0$, this is \emph{exactly} the same bound as obtained for the configurational entropy of the landscape in \cref{thm:energy_landscape_shattering}, when substituting the energy density cutoff $\epsilon$ by the expectation value of the energy density
\begin{align}
    \expval{\epsilon} &= -\frac{1}{n} \pdv{\beta} \log_2 Z \\
    &= - \pdv{\beta} \left[r + (1-r) \log_2 (1+e^{-\beta})\right] \\
    &= \frac{1-r}{1 + e^{\beta}}.
\end{align}
Similar to there, evidently for $T\to0$ we have $\sconf \to r = 1-m/n$, which is the logarithm of the number of ground states (the dimension of the kernel of $H$), since $H$ is full rank. Further, it is easy to see that the bound is strictly larger than $r$ for sufficiently large $\gamma$.
\end{proof}

As before, the above theorem can be instantiated using Gallager codes, which can be realized with arbitrarily large $\gamma$ (\cref{thm:gallagerexp}), and which are of full rank with high probability  (\cref{lem:gallager_no_redundancies}). Explicitly, we show the bound in \autoref{eq:thm:pmax_upper_bound} in \autoref{fig:complexity_energy} of the main text, for parameters chosen for a particular family of Gallager codes that realize sufficiently strong expansion.

\section{Features of the energy landscape from deleting checks\label{app:landscape_proof}}

In the main text and in \appref{app:complexity} we have used expansion and the absence of redundancies to argue for a complex energy landscape. Here we present an alternative route to derive properties of the energy landscape. These allow only more limited statements, but they do not rely on the absence of redundancies and require weaker expansion.

\subsection{General results}

We will always assume that the number of checks is proportional to the number of bits $m = \Theta(n)$.

\begin{definition}\label{def:sparse}
    We call a binary matrix $\Hcheck\in\mathbb{F}_2^{m\times n}$ \emph{$(w,b)$-sparse} if each row has Hamming weight $\leq w$ and each column has Hamming weight $\leq b$.
\end{definition}

\begin{theorem}\label{thm:low_energy_separated}
    Consider an expander code of rate $r$ defined by a full-rank parity check matrix $\Hcheck \in \mathbb{F}_2^{m\times n}$ that is $(w,b)$-sparse and $(\delta,\gamma)$-expanding with $\gamma > 1$.
    
    It holds that\\
    \begin{enumerate}[label=(\roman*)]
        \item there exists an exponential number (within each symmetry sector) of states that have non-zero energy density and that are separated by Hamming distance $> \delta n$ from each other and from ground states,
        \item for any $0 < \epsilon < \min \lbrace \frac{\delta \gamma}{2}, \frac{1-r}{1+bw} \rbrace$ there exist exponentially many clusters of states of energy $\leq \epsilon n$ that do not contain ground states,
        \item these clusters are separated by a $\Theta(n)$ energy barrier from each other and from clusters containing ground states.
    \end{enumerate}
\end{theorem}
\begin{proof}
    In order to argue (i), we show that we can find a set~$C$ containing a linear number of checks that do not have overlapping support.
    Since all checks in $\Hcheck$ are linearly independent, each check $c$ gives rise to a (up to addition of a code word) unique state $x_c$ for which only $c$ is violated.
    In other words, $x_c$ is the solution of $\Hcheck x_c = e_c$ which is unique up to an element of $\ker \Hcheck$.
    Consequently, for any subset $S\subseteq C$ the state $x_{S} = \sum_{c\in S} x_c$ has energy~$|S|$, so that all linear sized subsets~$S$, of which there are exponentially many, give rise to states~$x_S$ with nonzero energy density.

    We will first find a suitable set of checks $C$.
    To this end, we define an auxiliary graph $G$ with vertex set given by the $m=(1-R)n$ checks in $\Hcheck$.
    Two vertices in $G$ are connected if and only if the corresponding checks have overlapping supports.
    Let $b$ be the maximum number of checks acting on a single bit and $w$ the maximum support of a check (cf.\ \Cref{def:sparse}).
    It holds that each vertex in $G$ has degree $\leq b w$.
    We obtain the set $C$ via a greedy procedure, where we pick any vertex of $G$, add it to $C$ and then remove that vertex as well as all of its neighbors from $G$.
    We iterate this procedure until $G$ is empty.
    At every step we remove at most $1+bw$ vertices, so that the final set has size $|C| \geq \frac{1-R}{1+b w} n$.

    We will now argue that the states $x_S$ defined above are separated by a linear Hamming distance.
    For this, we observe that deleting the checks in~$C$ gives a new code with check matrix $\Hcheck_s$ for which the states~$x_S$ are code words.
    Deleting the checks $C$ reduces the number of checks acting on each bit by at most one, so that we have $|\Hcheck_s x| \geq (\gamma - 1) |x|$ for any $x\in \mathbb{F}_2^n$ with $|x|\leq \delta n$.
    Hence, the code defined by the parity check matrix with checks in~$C$ deleted still has distance~$\geq \delta n$.
    This concludes the proof of statement (i).
    
    Statement (ii) follows by choosing the subsets $S\subseteq C$ to be of size $\lfloor \epsilon n \rfloor$ so that the states~$x_S$ belong to clusters that do not overlap with each other or with clusters containing ground states.
    
    To demonstrate (iii), take any state $x$ inside a cluster such that $e = \frac{|\Hcheck x|}{n} < \epsilon$.
    In order to reach a state in a different cluster, we need to flip at least~$\delta n$ bits in~$x$.
    Hence, we must pass through an intermediate state $x+\Delta x$ where $\frac{\epsilon}{\gamma}n<|\Delta x|<\delta n$.
    By \Cref{lem:clustering} we must have $|\Hcheck(x+\Delta x)| > \epsilon n$, so that the energy difference between the starting state and the intermediate state is $\Delta E = |\Hcheck (x+\Delta x)| - |\Hcheck x| \geq (\epsilon - e)n$.
    Hence, any pair of clusters are separated by a linear energy barrier.
    Note that this argument holds for clusters containing ground states and clusters only containing states of non-zero minimum energy density.
\end{proof}

We can in fact strengthen the statement of \Cref{thm:low_energy_separated} if we assume that $\gamma > 2$.
In the stronger version we do not only have to overcome a linear energy barrier, but also the energy cost grows linearly in the Hamming distance from~$x$ for all intermediate states (up to some linear cut-off $\delta n$).
\begin{definition}
    We say that a state $x\in \mathbb{F}_2^{n}$ is the \emph{bottom of a linear energy well} if there exists a constant $\delta >0$ such that for all bit-flip vectors $\Delta x \in \mathbb{F}_2^{n}$ with $|\Delta x| \leq \delta n$ we have that $\Delta E = |H(x+\Delta x)| - |Hx| \geq \Omega(|\Delta x|)$.
\end{definition}

\begin{theorem}[Stronger version of \Cref{thm:low_energy_separated}]\label{thm:energy_landscape_linear}
    Taking $\gamma > 2$ in the assumptions of \Cref{thm:low_energy_separated} it follows that there exists an exponential number (up to ground state symmetries) of states that have non-zero energy density and that are at the bottom of a linear energy well.
\end{theorem}
\begin{proof}
    Let the set of checks $C$ be chosen as in the proof of \Cref{thm:low_energy_separated}.
    Define $\Hcheck_s$ as the parity check matrix with all checks from $C$ removed from $\Hcheck$ and $\Hcheck_d$ the matrix consisting of all checks in $C$.
    First we observe that for any $x\in \ker \Hcheck_s$ and $\Delta x \in \mathbb{F}_2^n$ it holds that
    \begin{align}\label{eqn:linear_well}
    \begin{split}
        \Delta E =& |\Hcheck(x+\Delta x)| - |\Hcheck x| \\
        =& |\Hcheck_s(x+\Delta x)| + |\Hcheck_d(x+\Delta x)| - |\Hcheck_s x| - |\Hcheck_d x| \\
        =& |\Hcheck_s \Delta x| + |\Hcheck_d x| + |\Hcheck_d \Delta x|\\
        &-2\, |\Hcheck_d x \wedge \Hcheck_d \Delta x| - |\Hcheck_d x| \\
        \geq& |\Hcheck_s \Delta x| - |\Hcheck_d \Delta x| .
    \end{split}
    \end{align}
    where $\wedge$ refers to the bit-wise \emph{and}-operation and we used that $|\Hcheck_d x \wedge \Hcheck_d \Delta x|\leq |\Hcheck_d \Delta x|$.
    We have already shown in the proof of \Cref{thm:low_energy_separated} that $|\Hcheck_s \Delta x| \geq (\gamma - 1) |\Delta x|$ whenever $|\Delta x| \leq \delta n$.
    As we chose~$C$ as a set of checks that do not share any common bits we have that $|\Hcheck_d \Delta x| \leq |\Delta x|$.
    Hence, we have shown that $|\Hcheck(x+\Delta x)| - |\Hcheck x| \geq (\gamma - 2)|\Delta x|$ when $|\Delta x| \leq \delta n$.
\end{proof}

\subsection{Gallager codes}

We now show that random LDPC codes satisfy the assumptions of \Cref{thm:low_energy_separated} and \Cref{thm:energy_landscape_linear}.

By \Cref{thm:gallagerexp,lem:gallager_no_redundancies} choosing the left-degree $\ell$ to be odd and large enough we obtain a Gallager code with $\Hcheck_G$ being full-rank, $(r,\ell)$-sparse and $(\delta>0,\gamma>2)$-expanding and thus satisfying the assumptions of \Cref{thm:low_energy_separated} and \Cref{thm:energy_landscape_linear}.

\section{Redundancies in Tanner codes on expander graphs\label{app:redundancies}}

\begin{figure}
    \includegraphics{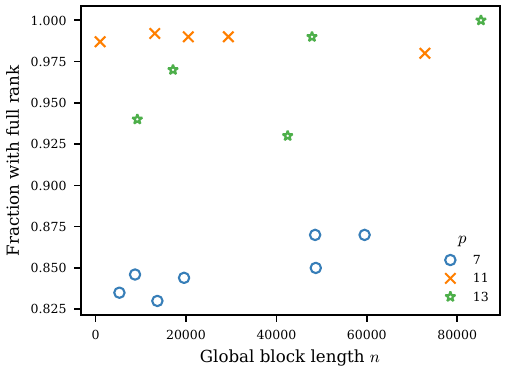}
    \caption{Number of redundant checks in Tanner codes defined on LPS expanders $X_{p, q}$ with local codes being random codes of block length $s=p+1$. We plot the fraction of instances with zero redundancies (i.e. the parity check matrix is full rank) against the block length of the global code $n$. The data was obtained from generating $10^3$ instances of codes with $n < 2 \times 10^4$ and $p \geq 11$ and $10^2$ instances otherwise.}
    \label{fig:app:redundancies_LPS}
\end{figure}

In this section, we study numerically the number of redundant checks in a code obtained from the Tanner construction. While we the codes used to define for the Tanner-Ising models studied numerically \autoref{sec:numerics} have \emph{zero} redundancies (that is their parity check matrix has full rank), there is no general guarantee for this to be the case. 
The fact that application of \cref{thm:landscape_expanders} and \cref{thm:sg_expanders} assumes the absence of redundancies warrants a more systematical study of their occurrence.

Before turning to the numerical results, we note that in Tanner codes, redundancies of the global code are equivalent to the codewords of a Tanner code defined on the same graph, but with the local code on each vertex replaced by its dual ($H_{\rm L}\to H_{\rm L}^\perp$). This is interesting because it means that under certain circumstances, one can rigorously guarantee the absence of \emph{local} redundancies. Concretely, if both the local and its dual have high distance $d_{\rm L}, d_{\rm L}^\perp \geq \lambda_2$, then both the original code and the code whose codewords correspond to the redundancies are expanding (and hence have high distance). This means that there cannot be any local redundancies.

We now turn to our numerical results. For the Tanner codes based on the hyperbolic  $\{3, 7\}$ lattice with the local code being the Hamming code, all of the finite instances we generated have zero redundancies. 
For the Tanner code based on degree-$7$ HGRRGs with the same local code we generated $10^3$ instances of the random graph for girth $3, 4$ and $5$ and $10^2$ instances of girth $6$ and also found that for all generated instances the parity check matrix $\Hcheck$ has full rank. 

A case of interest where redundancies do occur is that of Tanner codes defined on LPS expanders $X_{p, q}$ \cite{lps1988expanders,margulis1988expanders} of degree $s=p+1$. In particular, we consider the case where the local code is chosen randomly with $k_{\rm L} = \floor{s/2} + 1$. This case is of interest because in the limit of large degree $p \gg 1$ it has been rigorously established that these codes are expanding \cite{sipser_spielman1996}.
We generated $10^3$ instances of codes with $n < 2 \times 10^4$ and $p \leq 11$ and $10^2$ instances of all other codes. We show the fraction of instances with zero redundancies as a function of the size $n$ of the global code in \autoref{fig:app:redundancies_LPS}.
For large $n$ and $p > 7$, our data suggests that codes from this ensemble have zero redundancies with high probability. 
Note that in the case of $p+1=s=8$, many instances of the local code have distance $d_{\rm L}=2$.

\section{Sampling from the global Gibbs distribution\label{app:gibbs_samples}}

In this appendix, we explain how to sample states from the global Gibbs distribution $\pGibbs(\vec\sigma) \propto e^{-\beta E(\vec \sigma)}$ of the models studied in the main text. This is used to initialize the simulations of local dynamics in equilibrium at low temperatures.
Sampling from the Gibbs distribution is possible despite the fact local dynamics has exponentially long autocorrelation times at low temperature, because our models are based on \emph{linear} codes and either local redundancies (see below), but no [or at most $\order{1}$] global redundancies. 

\subsection{No redundancies}

Given a parity check matrix $\Hcheck\in\mathbb{F}_2^{m\times n}$ with full rank $m$, the energy of a state $\vec x\in\mathbb{F}_2^n$ is given by $E := \abs{\vec s}$, where $\Hcheck \vec x = \vec s \in \mathbb{F}_2^m$ is called the syndrome of $\vec x$.

Since $s_i \in \{0, 1\}$ and hence $\abs{\vec s} = \sum_i s_i$, we can write the Gibbs distribution as
\begin{equation}
    \pGibbs \propto e^{-\beta E(\vec \sigma)} 
            = e^{-\beta \abs{\vec s}}
            = \prod_i e^{-\beta s_i}.
    \label{eq:app:syndrom_iid}
\end{equation}
In the absence of redundancies, the components of $\vec s$ are i.i.d variables with distribution $p(s_i) = e^{-\beta s_i}$. This means we can sample from the Gibbs distribution simply by drawing $m$ i.i.d variables $s_i$ to construct a syndrome $\vec s$ and then get a corresponding state $\vec x$ by solving the linear equation $\Hcheck \vec x = \vec s$. Note that the solution to this linear equation, $\vec x$, is not unique. To sample the Gibbs distribution faithfully, we hence add to $\vec x$ a random vector from the kernel $\ker \Hcheck$.

Redundancies, i.e. dependencies between the rows of $\Hcheck$, make the distribution of syndromes dependent and hence the above procedure impossible in general.
There are however cases in which the above procedure can be adjusted and we discuss two of them in the following. 
First, we discuss the case of \emph{local} redundancies, which appear, for example, in the ``symmetrized'' models considered in the main text. 
Second, if the number of dependencies is small, we can adjust our general procedure to take this into account, however with an overhead that scales \emph{exponentially} with the number of redundancies.

\subsection{Local redundancies}

We now discuss the case of how to sample from the Gibbs distribution of the ``symmetrized'' models considered in the main text. The symmetrized model is based on the Tanner code on a graph $G$, $\tanner{G}{\Hcheck_{\rm L}^{\rm (sym)}}$, where given a local code $\Hcheck_{\rm L} \in \mathbb{F}_2^{m_0}$ without redundancies,
$\Hcheck_{\rm L}^{\rm (sym)} \in \mathbb{F}_2^{m}$ contains not only the rows of $\Hcheck_{\rm L}$ but \emph{all} $m = 2^{m_0}-1$ non-trivial linear combination of rows. The symmetrized local code $\Hcheck_{\rm L}^{\rm (sym)}$ has  rank $m_0$ and the global code $\mathcal T$ has an extensive number of redundancies. 

We can still sample from the global Gibbs distribution without the exponential overhead needed in the general case with a linearly growing number of redundancies, by exploiting their local structure in the symmetrized model. 

To this end, we sample the syndrome \emph{per vertex} of the graph $G$, that is per local code $\Hcheck_{\rm L}^{\rm (sym)}$. First, note that
\begin{align}\label{eq:symm-check}
    \abs{\Hcheck_{\rm L}^{\rm (sym)} \vec x_{\rm L}} =
    \begin{cases}
        0 & \text{if}~ \vec x_{\rm L} \in \ker{\Hcheck_{\rm L}} \\
        J_{m_0} := 2^{m_0-1} & \text{else}
    \end{cases},
\end{align}
that is the symmetrized model only distinguishes between codewords and non-codewords. Given the fact that the checks of local codes on different vertices are still mutually independent we can then write 
\begin{equation}
    \pGibbs \propto \prod_{v \in G} \exp[-\beta \abs{\vec s_{v}}]
\end{equation}
where $\vec s_v$ is the syndrome restricted to the vertex $v$.
Second, note that there are $2^{m_0}$ distinct syndrome vectors, which correspond to all combinations of violating the $m_0$ independent checks of the non-symmetrized local code. 

We hence draw a random \emph{non-redundant} syndrome $\vec s_v^{(0)}\in\mathbb{F}_2^{m_0}$ at each vertex using 
\begin{align}
    p(\vec s_v^{(0)}) \propto
    \begin{cases}
        1 & \text{if}~\vec s_v^{(0)} = 0 \\
        \exp(-\beta\,J_{m_0}) & \text{else}
    \end{cases}
    \label{eq:app:p_syndrome_dependent}
\end{align}
with a suitable normalization.
Given one such $\vec s_v^{(0)}$ per vertex $v$ we can construct a global non-redundant syndrome $\vec s^{(0)}\in\mathbb{F}_2^{m_{\rm G}}$, where $m_{\rm G} = n_v\,m_0$ is the total number of independent checks of the global code $\mathcal T$. This is in the image of the non-symmetrized global parity check matrix, $\Hcheck^{(0)}$, and hence we can compute a state from it by inverting the linear equation $\Hcheck^{(0)} \vec x = \vec s^{(0)}$.
Since $\vec s^{(0)}$ was drawn from the Gibbs distribution given by the energy $E = \Hcheck^{(\rm sym)}$, the state $\vec x$ will also be from that distribution.

We can use \autoref{eq:app:p_syndrome_dependent} to derive an exact expression for the energy of a symmetrized code in global equilibrium. In particular since all vertices are independent
\begin{align}\label{eq:E-G}
    \expval{E}_{\rm G} &= \frac{2n}{s} \expval{E_v} \\
    &= \frac{2n}{s} 2^{m_0-1}  \frac{ (2^{m_0}-1) e^{-\beta J_{m_0}}}{1 +  (2^{m_0}-1)e^{-\beta J_{m_0}}}
\end{align}
where $\expval{E_v}$ is the expectation value of the energy at a single vertex. Note that this expression does only depend on the graph degree $s$ and number of independent checks $m_0$ in the local code.

\subsection{Constant number of redundancies}

Consider now the general case of a parity check matrix $\Hcheck\in\mathbb{F}_2^{m\times n}$ with rank $m_0$. As before, we want to sample from the Gibbs distribution $\pGibbs \propto e^{-\beta \abs{\vec s}}$.
In the case of no redundancies, we can simply choose a random $\vec s$ by sampling the syndrome components from the i.i.d Gibbs distribution. 
Note that \autoref{eq:app:syndrom_iid} is still true even in the presence of redundancies, but now only some syndromes $\vec s$ are compatible with the constraints corresponding to the redundancies, in other words the equation $\Hcheck \vec x = \vec s$ has solutions only for a subset of all possible $\vec s$.
The allowed syndrome vectors $\vec s$ are given exactly by the image of $\Hcheck$ and a projector onto that can be efficiently constructed. 
We can hence proceed as before by drawing $\vec s$ as $m$ i.i.d. variables, but postselect for those syndromes in the image of $\Hcheck$. The chance of successful postselection is then $2^{m_0-m}$ and hence the overhead is exponentially large in the number of redundancies $m - m_0$. 

A particularly simple but relevant case is that of a global $\mathbb{Z}_2$ bit-flip symmetry, which arises if all checks are of even weight. This is the case for example when placing the extended $[8, 4, 4]$-Hamming code on a degree-eight graph. 
In this case, there is a single global redundancy, corresponding to the product of all checks being trivial. A particular simple basis for the image of $\Hcheck$ in this case is
\begin{equation}
    \mathcal B = \mqty(
    1 & 0 & 0 & \dots & 0 & 0 & 1 \\
    0 & 1 & 0 & \dots & 0 & 0 & 1 \\
    0 & 0 & 1 & \dots & 0 & 0 & 1 \\
      &   &   & \ddots\\
      0 & 0 & 0 & \dots & 1 & 0 & 1 \\
    0 & 0 & 0 & \dots & 0 & 1 & 1 \\
    )
\end{equation}
which means that given an i.i.d. samples syndrome $\vec s$, it suffices to check that $\abs{\vec s}$ is even. Given an even-weight $\vec s$, we can then truncate the full syndrome by ignoring the last entry and solve $\Hcheck^{(0)} \vec x = \vec s^{(0)}$, where $\Hcheck^{(0)}$ is the parity check matrix with the last row deleted and $s^{(0)}$ the truncated syndrome. 
Note that the choice of the last check as ``redundant'' is arbitrary.

Finally, we note that the techniques described above to treat local and (few) global redundancies may be easily combined, e.g. to sample from the Gibbs distribution of a symmetrized Tanner code with a global symmetry.

\section{Recursive calculations on tree graphs}\label{app:TensorNetwork}
In this appendix, we elaborate on the recursive methods used to identify transitions in Tanner codes defined on locally tree-like graphs. We formulate these techniques in the language of tensor networks. Throughout the discussion, we highlight similarities and differences from the Bethe lattice Ising model discussed in~\appref{app:TreeIsing}.
\subsection{Tensor network formulation}\label{app:tn}
The unnormalized Gibbs state, $e^{-\beta E(\vec{\sigma})}$, of a Tanner code factorizes into $\prod_v e^{-\beta E_v(\vec{\sigma}_v)}$ where $\vec{\sigma}_v$ denotes the spin configuration on the edges adjacent to $v$. We can represent this conveniently as a tensor network (see also Refs. \cite{alkabetz2021tensor,pancotti2023one} for similar treatments). Consider assigning a tensor $T_v(\vec{\sigma}_v) = e^{-\beta E_v(\vec{\sigma}_v)}$ to every vertex; for example, if the vertex has degree \gdeg, this is a \gdeg-legged tensor, with each leg $i=1,\ldots,\gdeg$ assigned a variable $\sigma_i = \pm 1$.  Graphically, this tensor is the object shown in~\autoref{fig:tree-sketch}(b)ii. We also assign a 3-legged tensor to every edge as follows. Let $e=(v,v')$ be an edge connecting vertices $v,v'$. We define a tensor
\begin{equation}
    \delta_e(\sigma_v,\sigma_{v'},\sigma_e) = \begin{cases}
        1 & \sigma_v = \sigma_{v'} = \sigma_e \\
        0 & \mathrm{otherwise}
    \end{cases}
\end{equation}
where $\sigma_v, \sigma_v'$ are ``virtual'' legs, $\sigma_e$ is a physical index denoting the value of the bit on that edge, and all three variables take values $\pm 1$. We can now form a network out of these tensors by contracting the virtual legs of $\delta_e$ with the appropriate legs of $T_v$ and $T_{v'}$. The output of the contraction is a tensor with one free leg for each edge, which evaluates to $e^{-\beta E(\vec{\sigma})}$ when the free legs are assigned the configuration $\vec{\sigma}$.

From this tensor network representation, various quantities can be obtained. First, note that if we sum over the physical index $\sigma_e$, we get $\sum_{\sigma_e = \pm 1} \delta_e(\sigma_v,\sigma_{v'},\sigma_e) = \delta(\sigma_v,\sigma_{v'})$. Thus, for any $\sigma_e$ that has been summed over, the tensors $T_v$ and $T_{v'}$ are contracted along the leg corresponding to the edge $e=(v,v')$. 
In particular, the partition function $\mathcal{Z} = \sum_{\vec{\sigma}} e^{-\beta E(\vec{\sigma})}$ is equal to contracting a tensor network made out of the tensors $T_v$. 
Similarly, the (unnormalized) marginal distribution of the variables $\sigma_e$ within some region $A \subset E$ is obtained by by contracting the same tensor network in the complement $B = E \setminus A$ and keeping the original network, with the physical legs~$\sigma_e$ free, within $A$. We can also compute the partition function conditioned on the spins within a certain region, i.e. $\mathcal{Z}(\vec{\sigma}_A)$, simply by contracting fixed spins $\vec{\sigma}$ at those locations. 

The tensor network formulation is valid for any underlying graph, but in general, contracting these tensor networks is made difficult by the presence of loops. This difficulty is avoided if we consider a model on a tree graph, where no loops are present. Expectation values at the root of a tree can be obtained by contracting the tensor network inward from the leaves to the root, and each layer of contraction translates into a recursion relation for the quantity of interest, e.g., the conditional magnetization $m_r$ (\autoref{eq:m-recursion-ferro}) or the distribution $Q^{(r)}$ (\autoref{eq:Q-r}). Root observables are then related to bulk observables by a final recursion step joining two rooted trees, as detailed below. In the remainder of this section, we use these recursion relations to derive the ferromagnetic and spin glass temperatures, as well as an analytical upper bound on \Tglass~using a two-copy tensor network.

\subsection{Recursion relation for general boundary conditions}
Consider forming a rooted tree by joining together $\gdeg - 1$ branches, where the $i$th branch has boundary configuration $\vec{\sigma}_\partial^{(i)}$, partition function $\mathcal{Z}_\partial(\vec{\sigma}_\partial^{(i)})$, and conditional magnetization $m_i$ on the root spin. Let $\bm{m} = (m_2,...,m_\gdeg)$, and let $(\sigma_2,...,\sigma_\gdeg)$ denote a particular spin configuration on the roots of these branches. In terms of the tensors $T_v$ and $p_i(\sigma_i) = (1 + \sigma_i m_i)/2$, the tree produced by joining these branches has a root magnetization of (cf.~\autoref{eq:m-recursion})
\begin{align}
F(\bm{m}) &= \left(\sum_{\sigma_1,\dots,\sigma_\gdeg} T_v(\sigma_1,\dots,\sigma_\gdeg) \sigma_1 \prod_{i=2}^s p_i(\sigma_i) \right)/z(\bm{m})
\end{align}
where (cf. \autoref{eq:z})
\begin{align}
z(\bm{m}) &= \sum_{\sigma_1,\dots,\sigma_\gdeg} T_v(\sigma_1,\dots,\sigma_\gdeg) \prod_{i=2}^s p_i(\sigma_i).
\end{align}

In the main text, we introduced a family of boundary conditions parameterized by $\alpha$ (\autoref{eq:alpha}), which weight the boundary conditions by a factor $\mathcal{Z}_\partial(\vec{\sigma}_\partial)^\alpha$. This global reweighting factor can be handled \textit{locally} at the level of a single recursion step, by noting that the boundary configuration of the output is $\vec{\sigma}_\partial = (\vec{\sigma}_\partial^{(2)},\vec{\sigma}_\partial^{(3)},...,\vec{\sigma}_\partial^{(\gdeg)})$, and its partition function is
\begin{align}\label{eq:Z-partial}
\mathcal{Z}_\partial(\vec{\sigma}_\partial) &= \sum_{\sigma_1,\dots,\sigma_\gdeg} T_v(\sigma_1, \dots,\sigma_\gdeg) \prod_{i=2}^\gdeg \left(p_i(\sigma_i) \mathcal{Z}_\partial(\vec{\sigma}_\partial^{(i)})\right) \notag \\
&= z(\bm{m}) \prod_{i=1}^b \mathcal{Z}_\partial(\vec{\sigma}_\partial^{(i)}).
\end{align}

Now let $Q_{\alpha}^{(r)}$ denote the distribution of the root magnetization of a depth $r$ tree, for a given $\alpha$. The recursion relation for $\alpha = 1$ presented in the main text (\autoref{eq:eq18} straightforwardly generalizes to
\begin{equation}\label{eq:alpha-recursion}
Q_\alpha^{(r+1)}(m) = \frac{\int \delta(m - F(\bm{m})) z(\bm{m})^\alpha \prod_{i=2}^\gdeg dQ_\alpha^{(r)}(m_i)}{\int z(\bm{m})^\alpha \prod_{i=2}^\gdeg dQ_\alpha^{(r)}(m_i)}.
\end{equation}
For $\alpha > 0$, the factor $z(\bm{m})^\alpha$ penalizes the insertion of defects that arise from joining branches with ``inconsistent'' boundary conditions.

For the symmetrized Hamming [7,4,3] code, $T_v(\sigma_1,...,\sigma_\gdeg)=+1$ if $(\sigma_1,...,\sigma_\gdeg)$ is a codeword, and $\exp(-4\beta) \equiv y$ otherwise (cf. \autoref{eq:symm-check}). The $\mathbb{Z}_2$ symmetry of the code implies $Q_\alpha^{(r)}(m) = Q_\alpha^{(r)}(-m)$. The initial condition is a pair of delta functions,
\begin{equation}\label{eq:init}
Q_\alpha^{(0)}(m) = \frac{1}{2}\left[\delta(m-1) + \delta(m+1)\right].
\end{equation}

In general,~\autoref{eq:alpha-recursion} can only be simulated numerically, as the number of delta functions in the distribution quickly proliferates. Even in such cases, we can gain quantitative insight into the distribution and its moments, as we now show. 

\subsection{Memory transition} 

An exceptional case where the fixed point distribution \textit{is} analytically solvable is $\alpha=\infty$, where the only allowed boundary configurations are those consistent with a global codeword. As a result, the distribution at depth $r$ concentrates onto a pair of delta functions:
\begin{equation}
Q_\infty^{(r)}(m) = \frac{1}{2} \left[\delta(m-m_r) + \delta(m + m_r)\right].
\end{equation}
Thus,~\autoref{eq:alpha-recursion} reduces to an equation for a single variable, $m_r$. This is the same flow variable defined in the main text (\autoref{eq:m-recursion-ferro}), where we took the slightly different perspective favoring a \textit{particular} codeword, such as the all zero codeword, which breaks the $\mathbb{Z}_2$ symmetry and selects $m = + m_r$. 

When the local code is the symmetrized Hamming [7,4,3] code,~\autoref{eq:m-recursion-ferro} evaluates to
\begin{equation}\label{eq:mr}
f(m_r) \equiv m_{r+1} =  \frac{4 m_r^3 (1 - y)}{1 + 7 y + 3 m_r^4 (1-y)}.
\end{equation}

\begin{figure}[t]
\includegraphics[width=\linewidth]{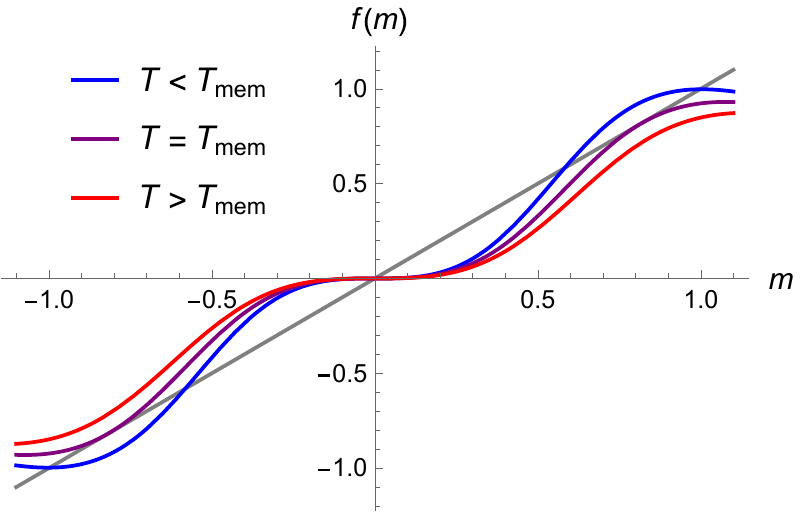}
\caption{\autoref{eq:mr} at $T=0.6<\Tdyn$ (blue), $T=\Tdyn$ (purple), $T=1.6>\Tdyn$ (red), along with the line $Y=m$ (gray).\label{fig:fm}}
\end{figure}

\autoref{fig:fm} shows $f(m)$ for temperatures below, at, and above \Tdyn. The fixed points of~\autoref{eq:mr} are intersections between the curves $Y=f(m)$ and the line $Y=m$~\cite{eggarter1974cayley,baxter2007exactly}.
For $y \leq 1/25$, corresponding to $\Tdyn = 4/\ln(25)$, there are five intersections:
\begin{equation}
m_* = 0, \pm \sqrt{\frac{2\pm \sqrt{(1-25y)/(1-y)}}{3}}.
\end{equation}
As shown in~\autoref{fig:ferro} in the main text, these fixed points arrange into a stable paramagnetic fixed point $m_*=0$; a symmetric pair of stable, ``ferromagnetic'' fixed points near $\pm 1$; and a symmetric pair of unstable fixed points in between. The stability can be inferred from the slope of the curve ($\frac{\mathrm{d} f}{\mathrm{d} m}|_{m_*}$) relative to the slope of the gray line (unity) at the intersections in ~\autoref{fig:fm}: a small perturbation away from a fixed point with $\frac{\mathrm{d} f}{\mathrm{d} m}|_{m_*} < 1$ flows back to the fixed point, while a perturbation away from a fixed point with $\frac{\mathrm{d} f}{\mathrm{d} m}|_{m_*} < 1$ will flow away.

At $T=\Tdyn$, the pairs of stable and unstable nontrivial fixed points ``annihilate'' and become marginal, with $\frac{\mathrm{d} f}{\mathrm{d} m}|_{m_*} = 1$. For $T>\Tdyn$, only the trivial paramagnetic fixed point remains.

\subsubsection{Bulk field}

Rather than fixing boundary conditions, another way to probe the memory transition is by adding \emph{bulk} magnetic fields. Choosing a uniform field which favors either the all-zero or all-ones codeword, the total energy is modified $E \to E + h_\text{bulk} \sum_e \sigma_e$. Through a redefinition of the local terms $E_v$ that appear in the tensors $T_v$, we obtain the same overall tensor network structure, which allows us to evaluate bulk observables. 

When the graph is a tree, we can again obtain the bulk magnetization recursively. The initial condition $m_0$ of the flow is set by the boundary field; if the boundaries are left free, $m_0=0$. The nonzero bulk field modifies the recursion relation~\autoref{eq:mr} 
and its fixed points, which are shown as a function of $h_{\rm bulk}$ at three representative temperatures in~\autoref{fig:bulk-field}a. 

At all nonzero temperatures, the root magnetization is a smooth function of $h_{\rm bulk}$ at small $h_{\rm bulk}$, in contrast to the Ising model on a tree, where the root susceptibility diverges below \Tdyn. This is because in the vicinity of $h_{\rm bulk}=0$, there exists a \textit{stable} fixed point $m_{\mathrm{para}}(h,T)$ which is smoothly connected to the zero-field paramagnet: it passes through the origin and depends only weakly on $h_{\rm bulk}$. This is the fixed point to which the initial condition $m_0=0$ flows at weak field. However, at low temperatures, there exists a critical field $\pm h_\text{bulk}^\text{crit}(T)$ at which $m_{\mathrm{para}}(h,T)$ disappears, so the magnetization jumps discontinuously to the nearest remaining stable fixed point, near $\pm 1$. This first-order transition occurs if there is an interval of $h_{\rm bulk}$ on which multiple fixed points coexist, i.e., if the allowed root magnetization $m(h)$ is multi-valued on some interval.

As \autoref{fig:bulk-field}a illustrates, this bulk-field-induced transition persists up to $T_h \approx 1.893 > \Tdyn$. $h_{\rm bulk}^{\rm crit}(T)$ smoothly increases with temperature up to $T_h$, and while it is finite for any nonzero temperature, $\lim_{T\rightarrow 0} h_{\text{bulk}}^\text{crit}(T) = 0$ (\autoref{fig:bulk-field}b). 

While $h_{\rm{bulk}}^{\rm crit}(T)$ shows no sharp change at \Tdyn, the signatures of the memory transition appear when one considers the effect of bulk and boundary fields together. Recall that for $T\leq\Tdyn$, $m(h=0)$ is multi-valued: depending on the boundary field, the system flows to a stable fixed point with $m_*>0, m_*=0,$ or $m_* < 0$. Turning on a bulk field, we let $h_{\rm bulk}^{\rm opp}(T)$ denote the magnitude of the largest bulk field at which negatively magnetized and positively magnetized stable fixed points coexist. For $T<\Tdyn$, $h_{\rm bulk}^{\rm opp}(T) > 0$: up to this point, a strong enough negative boundary field can overwhelm the positive boundary field. As shown in~\autoref{fig:bulk-field}b, $h_{\rm bulk}^{opp}(T)$ decreases with temperature, vanishing at $T=\Tdyn$. 

Then, in the intermediate regime $\Tdyn < T < T_{h}$, there is still a first-order transition at finite \textit{bulk} field, but in the absence of a bulk field, not even an infinitely strong \textit{boundary} field (i.e., polarized BCs) is sufficient to order the bulk, since $m(h)$ is single-valued at $h=0$. This again contrasts with the Ising model, where the same \Tdyn~can be diagnosed by susceptibility to bulk fields \textit{or} boundary fields. The minimal bulk field necessary to stabilize a nontrivial fixed point, denoted $h_{\rm bulk}^{\rm min}(T)$ in~\autoref{fig:bulk-field}b, is identically zero for $T<\Tdyn$, then increases monotonically up to $T_h$, with $h_{\rm bulk}^{\rm min}(T_h) = h_{\rm bulk}^{\rm crit}(T_h)$.

\begin{figure}[t]
\includegraphics[width=\linewidth]{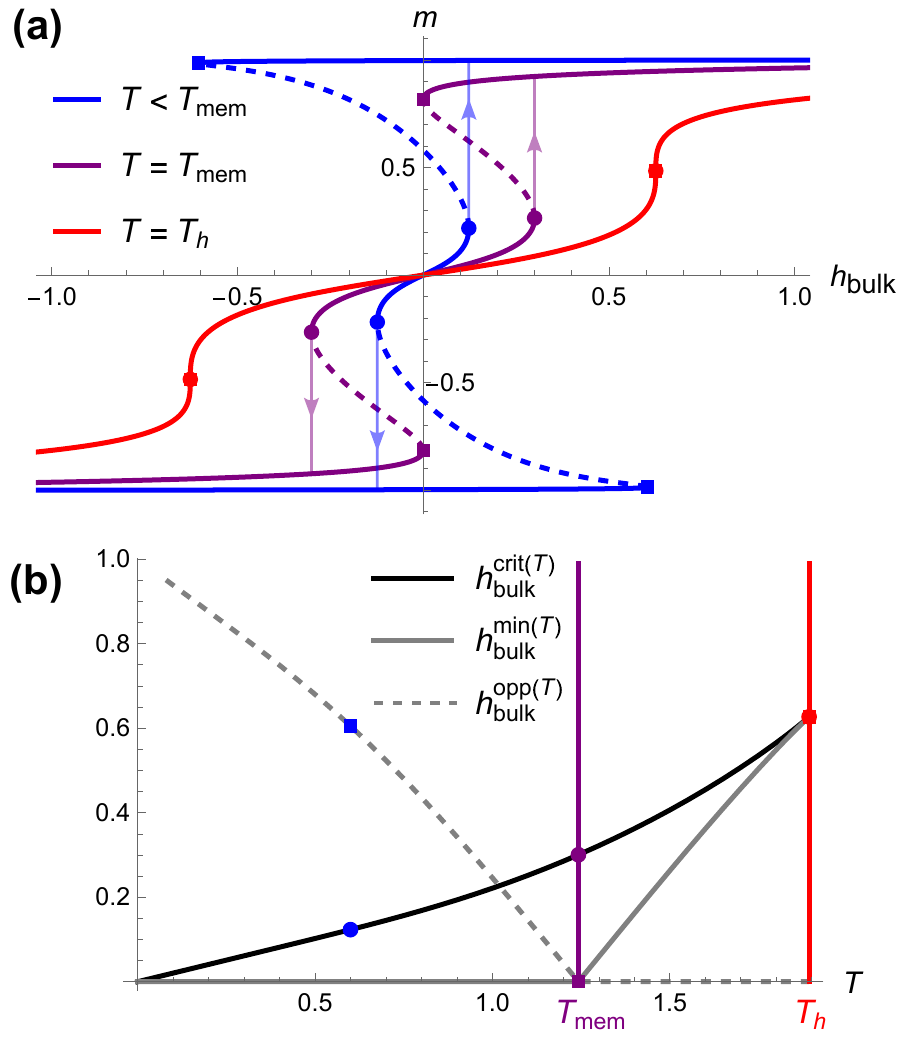}
\caption{(a) Fixed points under uniform bulk field $h_{\mathrm{bulk}}$, at temperatures $T=0.6 < \Tdyn$ (blue), $T=\Tdyn$ (purple), and $T=T_h$ (red). Solid (dashed) curves indicate stable (unstable) fixed points. Light lines marked with arrows indicate the discontinuous jump in the root magnetization at $\pm h_{\mathrm{bulk}}^{\mathrm{crit}}(T)$ for the lower two temperatures. Solid circles mark the endpoints of $m_\mathrm{para}(h,T)$. Solid squares mark $\pm h_{\rm bulk}^{\rm opp}(T)$ (for $T<\Tdyn$) and $\pm h_{\rm bulk}^{\rm min}(T)$ (for $T>\Tdyn$). (b) $h_{\rm bulk}^{\rm{crit}}$ (solid black), $h_{\rm bulk}^{\rm{min}}$ (solid gray), and $h_{\rm bulk}^{\rm opp}$ (dashed gray) as a function of $T$ at low temperatures. Purple, blue, and red points correspond to the points marked in (a).\label{fig:bulk-field}}
\end{figure}

\subsection{Spin glass transitions}

As we have seen, the tensor network formulation naturally allows us to calculate bulk observables conditioned on a particular boundary configuration~$\vec{\sigma}_\partial$. 
In particular, letting $\vec{\sigma}$ denote the configuration of spins on the interior, we can consider the magnetization of some bulk spin $e$ conditioned on the boundary spins
\begin{equation}
    \mathbb{E}[\sigma_e | \vec{\sigma}_\partial] \equiv \sum_{\vec{\sigma}} p(\vec{\sigma}|\vec{\sigma}_\partial) \sigma_e = \frac{\sum_{\vec{\sigma}} e^{-\beta E(\vec{\sigma},\vec{\sigma}_\partial)}\sigma_e}{\sum_{\vec{\sigma}} e^{-\beta E(\vec{\sigma},\vec{\sigma}_\partial)}}.
\end{equation}
In the tensor network language, this corresponds to summing over all bulk legs while keeping the boundary legs fixed in a particular configuration $\vec{\sigma}_\partial$; the numerator involves the insertion of an additional Pauli $Z$ matrix on the physical leg corresponding to $\sigma_e$, before summation. The spin glass transition is probed by the variance of the conditional expectation value as we vary over an ensemble of boundary conditions $\vec{\sigma}_\partial$:
\begin{equation}\label{eq:BC_variance}
    \sum_{\vec{\sigma}_\partial} p_\partial(\vec{\sigma}_\partial) \left[\mathbb{E}[\sigma_e | \vec{\sigma}_\partial]\right]^2 - \left[ \sum_{\vec{\sigma}_\partial} p_\partial(\vec{\sigma}_\partial) \mathbb{E}[\sigma_e | \vec{\sigma}_\partial] \right]^2.
\end{equation}

Taking $p_\partial \propto \mathcal{Z}_\partial^{\alpha}$, and noting that the second term vanishes by symmetry,~\autoref{eq:BC_variance} simplifies to:
\begin{equation}\label{eq:BC_variance_gamma}
     \overline{\mathbb{E}[\sigma_e|\vec{\sigma}_\partial]^2}_\alpha = \frac{\sum_{\vec{\sigma}_\partial} Z_\partial(\vec{\sigma}_\partial)^{\alpha - 2}\left(\sum_{\vec{\sigma}} e^{-\beta E(\vec{\sigma},\vec{\sigma}_\partial)} \sigma_e \right)^2}{ \sum_{\vec{\sigma}_\partial} \mathcal{Z}_\partial(\vec{\sigma}_\partial)^{\alpha}}.
\end{equation}
For $\alpha=1$,~\autoref{eq:BC_variance_gamma} becomes nonzero at the transition temperature \Tglass~numerically determined in the main text.

\subsubsection{Two-copy calculation: $\alpha=2$}

The computation becomes analytically tractable if we instead take $\alpha=2$:
\begin{equation}\label{eq:BC_var_annealed}
     \overline{\mathbb{E}[\sigma_e|\vec{\sigma}_\partial]^2}_2 = \frac{\sum_{\vec{\sigma}_\partial}\left(\sum_{\vec{\sigma}} e^{-\beta E(\vec{\sigma},\vec{\sigma}_\partial)} \sigma_e \right)^2}{ \sum_{\vec{\sigma}_\partial}\left(\sum_{\vec{\sigma}} e^{-\beta E(\vec{\sigma},\vec{\sigma}_\partial)}\right)^2}.
\end{equation}
This latter quantity can be interpreted as an \textit{annealed} average, over infinite boundary fields of random sign, of the correlation function $\langle \sigma_e^{(1)} \sigma_e^{(2)} \rangle$ on two replicas, shown in~\autoref{fig:double-tree}. It has the advantage that both the numerator and the denominator can be evaluated efficiently in the tensor network formalism as we now detail.

Consider first the denominator, $\sum_{\vec{\sigma}_\partial}\mathcal{Z}_\partial(\vec{\sigma}_\partial)^2$. The partition function $\mathcal{Z}_\partial(\vec{\sigma}_\partial)$ is itself a tensor network, where all bulk legs have been summed over but the boundary legs have been left free. Now, evaluating the average of its square simply amounts to taking two copies of the same tensor network and contracting the boundary legs between the two copies. The calculation of the numerator is completely analogous, with an additional Pauli-$Z$ matrix acting on $\sigma_e$ on both copies ($\sigma^{(1)}, \sigma^{(2)}$ in~\autoref{fig:double-tree}a). Overall, this amounts to evaluating the same kind of tensor network as in~\appref{app:tn}, but with the local tensors now given by $T_v \otimes T_v$ for each vertex $v$ (\autoref{fig:double-tree}b). The legs from the two copies are combined into a single leg with bond dimension $4$, corresponding to physical spins $((1,1),(1,-1),(-1,1),(-1,-1))$. Contracting the boundary legs amounts to imposing the boundary condition $p^{(0)} = (1/2,0,0,1/2)$ on this $4$-dimensional Hilbert space.

Again, when the underlying graph is a tree, we can turn this tensor network construction into a recursion relation for the quantity of interest, whose fixed points characterize the bulk observables.
Let us focus again on the case of the local code being the symmetrized $[7,4,3]$ Hamming code. By symmetry, the distribution $p^{(r)}$ can be parameterized by a single parameter $\mu_r$ as
\begin{equation}
    p^{(r)} = \frac{1}{4}\begin{pmatrix}
    1 + \mu_r \\
    1 - \mu_r \\
    1 - \mu_r \\
    1 + \mu_r
    \end{pmatrix}.
\end{equation}
The correlation between the spins $\sigma^{(1)}, \sigma^{(2)}$ at the roots of two depth-$r\rightarrow\infty$ copies is then\footnote{The correlation between bulk spins in two copies of a Cayley tree, where each non-leaf node has degree $b+1$ and there is no distinguished root, would then be $2\mu/(1+\mu^2)$.}
\begin{equation}
    \langle \sigma^{(1)} \sigma^{(2)} \rangle = \mu_{r\rightarrow\infty}.
\end{equation}

\begin{figure}[t]
\includegraphics[width=\linewidth]{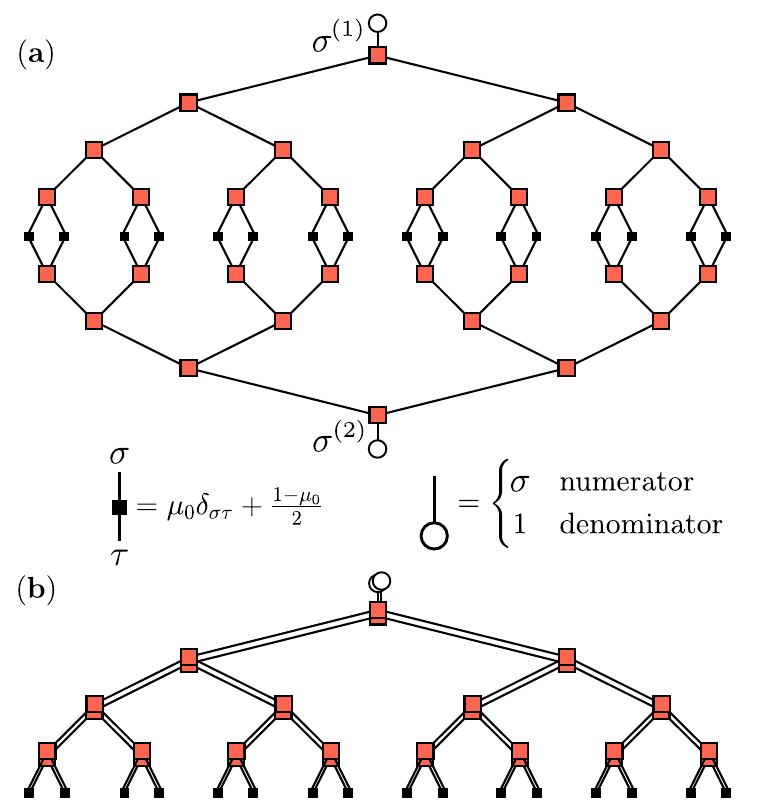}
\caption{Two-copy tensor network used to calculate the variance with $\alpha=2$ BCs. (a) Two trees are contracted along their boundaries (black square tensors), with $\mu_0=0$ in~\autoref{eq:BC_var_annealed}). The root spins of the two copies, $\sigma^{(1)}$ and $\sigma^{(2)}$, are summed over in the denominator of~\autoref{eq:BC_var_annealed}, and contracted with a Pauli Z in the numerator. (b) Folding over the two copies yields a tree with the same structure as in the one-copy calculation, but with bond dimension 4 on each doubled leg. \label{fig:double-tree}}
\end{figure}

To evaluate~\autoref{eq:BC_var_annealed}, we take $\mu_0 = 1$, i.e. the boundary legs of the two copies are perfectly correlated. More generally, initializing to $\mu_0 \in (0,1]$ corresponds to a setup with ferromagnetically coupled boundaries, while $\mu_0 \in [-1,0)$ corresponds to boundaries with antiferromagnetic coupling. The question of whether the (annealed) bulk variance is nonzero thus generalizes to asking whether a given nonzero $\mu_0$ flows to the trivial fixed point $\mu=0$, where the copies are decoupled, or to a nontrivial fixed point where the coupling at the boundary persists into the bulk. The analysis is therefore identical to that of uniformly polarized boundary conditions, but with the initial condition $\mu_0$ now interpreted as the strength of coupling between two trees, whereas $m_0$ parameterized the polarization on the boundary of a single tree.

One layer of the doubled tensor network contraction induces a recursion relation on the variable $\mu_r$:
\begin{equation}
\mu_{r+1} = \frac{4\mu_r^3 (1-y)^2}{(1+7y)^2 + 3\mu_r^4(1-y)^2} \equiv f_2(\mu_r).
\end{equation}

\begin{figure}
\includegraphics[width=\linewidth]{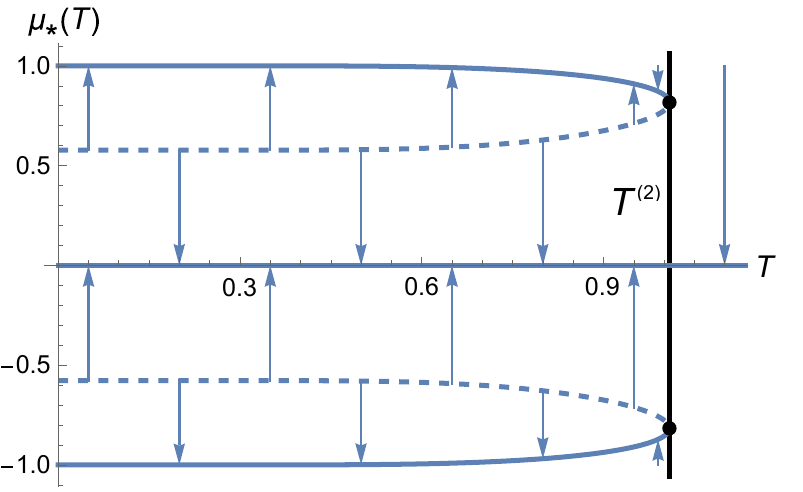}
\caption{Fixed points of the recursion relation for the two-copy ``annealed'' calculation (\autoref{eq:mstar}). Dashed (solid) curves indicate unstable (stable) fixed points. Two pairs of nontrivial solutions merge at $T=T^{(2)}$, where both solutions become marginal (black circles). The trivial solution at $\mu_*=0$ is stable at all $T$. \label{fig:annealed}}
\end{figure}

The fixed points $\mu_* = f_2(\mu_*)$, shown in~\autoref{fig:annealed}, are at
\begin{equation}\label{eq:mstar}
    \mu_* = 0, \pm \sqrt{\frac{2(1-y)\pm \sqrt{1-y(50+143y)}}{3(1-y)}}.
\end{equation}
At high temperatures, only the trivial solution is real. The nontrivial fixed points, which come in pairs symmetric about $\mu=0$, appear at
\begin{equation}
    T^{(2)} = \frac{4}{\ln(143/(16 \sqrt{3} -25))} \approx 1.009.
\end{equation}
For $T\leq T^{(2)}$, the initial condition $\mu^{(0)}=1$ flows to the top branch in~\autoref{fig:annealed}, so as the temperature is lowered, the annealed variance (\autoref{eq:BC_var_annealed}) jumps from zero to a finite value at $T=T^{(2)}$, signaling a first-order transition. As with the transition at \Tdyn, the finite jump results from the merging and annihilation of pairs stable and unstable fixed points.

\subsubsection{Population dynamics: $\alpha = 1$}
The distribution under properly correlated boundary conditions ($\alpha=1$), which shows a transition at $\Tglass = 0.704 \pm 0.005$, is iterated numerically as follows. 

We would like to eliminate the factor of $z(\vec{m})$ in~\autoref{eq:eq9}, so as to sample from the recursion relation without carrying around reweighting factors. To do so, define the conditional distribution
\begin{equation}\label{eq:tildeQ}
\tilde{Q}^{(r)}_{\sigma_0}(m) = (1 + m \sigma_0) Q^{(r)}(m) \Leftrightarrow Q^{(r)} = \frac{1}{2} \left(\tilde{Q}_{1} + \tilde{Q}_{-1}\right),
\end{equation}
where we have suppressed the $\alpha=1$ subscript. $\tilde{Q}_{\sigma_0}$ is the distribution of conditional root magnetizations across the following ensemble: freeze the root to $\sigma_0$ and sample configurations on the leaves, $\vec{\sigma}_\partial$, with probability proportional to $\mathcal{Z}_{\sigma_0,\vec{\sigma}_\partial}$, where $\mathcal{Z}_{\sigma_0,\vec{\sigma}_\partial}$ is the partition function on the tree with the root and leaves frozen to $\sigma_0, \vec{\sigma}_{\partial}$ respectively.

Substituting~\autoref{eq:tildeQ} into~\autoref{eq:eq18} yields 
\begin{align}\label{eq:eq9}
\tilde{Q}^{(r+1)}_{\sigma_0}(m) &\propto \sum_{\sigma_2,\dots,\sigma_{\gdeg}} T_v(\sigma_0,\sigma_2,\dots,\sigma_\gdeg) \notag \\
&\times \int \delta(m - F(\vec{m})) \prod_{i=2}^{\gdeg} d \tilde{Q}_{\sigma_i}^{(r)}(m_i).
\end{align}

Due to the fortuitous cancellation of $z(\vec{m})$, we can now simulate~\autoref{eq:eq9} via the population dynamics method of Ref.~\cite{mezard2006reconstruction}. First, initialize two populations of size $M$, $\tilde{Q}_1$ and $\tilde{Q}_{-1}$. The initial condition for $Q^{(0)}$ (\autoref{eq:init}) translates to $\tilde{Q}^{(0)}_{\sigma}(m) = \delta(m-\sigma)$, i.e. every element in $\tilde{Q}_{\sigma}$ takes the value $\sigma$. Then, to simulate~\autoref{eq:eq9} for a given $\sigma$, for $j=1,...,M$:
\begin{enumerate}
\item Sample the spin configuration $\sigma_2,...,\sigma_\gdeg$ with Boltzmann weight $T_v(\sigma,\sigma_2,...,\sigma_\gdeg)$.
\item For each $\sigma_i$, independently sample $m_i$ from population $\tilde{Q}_{\sigma_i}^{(r)}$.
\item Set the $j$th element of the population $\tilde{Q}_\sigma^{(r+1)}$ equal to $F(m_2,...,m_\gdeg)$.
\end{enumerate}

This method was used in Ref.~\cite{mezard2006reconstruction} to analyze the $q$-state Potts model with spin degrees of freedom at the vertices of a tree. For those models, the $\sigma_i$ in step (1) can be sampled independently, a simplification that no longer applies when the local code has checks of weight $\geq 3$.

Given a population of size $M$ at iteration $r$, only $M/(\gdeg-1)$ independent samples can be generated at iteration $r+1$. Therefore, the above method introduces some correlations between samples, so we must take $M$ large enough that the final result is insensitive to the precise choice of $M$. To check that $M$ is large enough, we perform 10 independent simulations and verify that each run yields consistent results, as evidenced by the small error bars in~\autoref{fig:Ir-plateau}a. $M$ ranges from $10^5$ well away from \Tglass, to $1.5 \times 10^6$ at $T=0.708, 0.71$. As $T \nearrow \Tglass$, $I(r)$ converges more slowly towards its plateau value. At $T=0.708$, $I(r)$ appears to be falling off a plateau, with $I(r)$ dropping sharply to 0 between $r=60$ and $r=80$. However, as different runs begin to diverge at $r\approx 50$, the average mutual information beyond this point is uncertain, and we cannot definitively conclude whether $T=0.708$ is just above or just below the spin glass transition.\footnote{For this reason, we truncate~\autoref{fig:Ir-plateau}a at $r=55$ and exclude $T=0.708$ from the right panel.}

\begin{figure}
\includegraphics[width=
\linewidth]{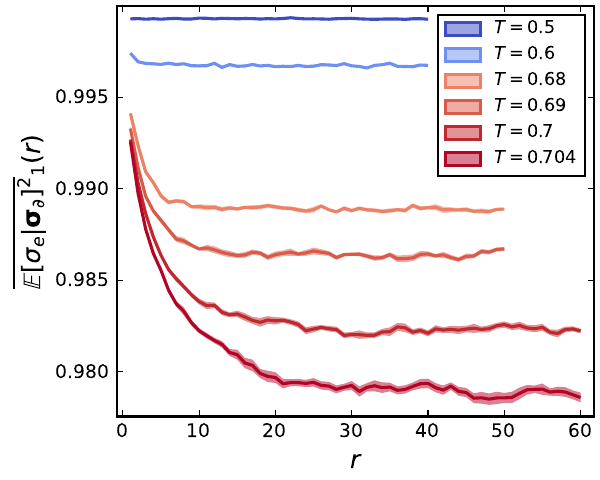}
\caption{Variance of the bulk magnetization with free boundary conditions ($\alpha=1$), determined from~\autoref{eq:var-sample}, as a function of the tree depth $r$. For each $T$, a population size of $M=10^5$ was used, and $M$ bulk samples were drawn according to~\autoref{eq:eta-bulk} for each population. Ribbons show the standard deviation across 10 independent runs.}
\end{figure}

From the sampled $\tilde{Q}_\sigma(m)$, we can also evaluate the variance of a spin deep in the bulk (\autoref{eq:BC_variance}). To transform the magnetization of a root spin (on a dangling edge connected to only one vertex/local code) to the magnetization of a bulk spin (connected to two vertices), a final recursion step is required. For a given $\sigma, r$ independently sample $m_1$ and $m_2$ from $\tilde{Q}^{(r)}_\sigma(m)$, and let
\begin{equation}\label{eq:eta-bulk}
    m_{bulk} = \frac{m_1 + m_2}{1 + m_1 m_2}~.
\end{equation}
Repeating this $M$ times for $\tilde{Q}_1$ and $\tilde{Q}_{-1}$ and aggregating the two populations
 yields the distribution of the central spin in a Cayley tree with $R=r$ generations. As $r\rightarrow \infty$, this distribution converges towards that of a spin deep in the bulk of the Bethe lattice, with its variance estimated numerically as
\begin{equation}\label{eq:var-sample}
\overline{\mathbb{E}[\sigma_e|\vec{\sigma}_\partial]^2}_1 = \frac{1}{2M} \sum_{i=1}^{2M} m_{i, bulk}^2.
\end{equation}
At $T=0.5$ and $T=0.6$, the plateau value is in good agreement with the long-time, large-system autocorrelation from Monte Carlo dynamics on HGRRGs, shown in in~\autoref{fig:glass}b. 

\subsubsection{Quenched boundary condition: $\alpha=0$}

Above we have seen that while both $\alpha=1$ and $\alpha=2$ yield a first-order ``spin glass'' transition, the latter transition occurs at a higher temperature, as the larger $\alpha$ more strongly suppresses the frustrated boundary conditions which tend to disorder the bulk. Let us also briefly comment upon the case $\alpha = 0$, for which~\autoref{eq:BC_variance_gamma} evaluates to:
\begin{equation}\label{eq:BC_variance_quenched}
    \overline{\mathbb{E}[\sigma_e|\vec{\sigma}_\partial]^2}_0 = \frac{1}{2^{|\partial|}}\sum_{\vec{\sigma}_\partial} \left(\frac{\sum_{\vec{\sigma}} e^{-\beta E(\vec{\sigma},\vec{\sigma}_\partial)}\sigma_e}{\sum_{\vec{\sigma}} e^{-\beta E(\vec{\sigma},\vec{\sigma}_\partial)}}\right)^2.
\end{equation}
\autoref{eq:BC_variance_quenched} can be interpreted as a quenched random average, as is performed for the Ising spin glass in Ref.~\cite{chayes1986mean}. In the Ising model, this i.i.d., uncorrelated boundary condition emerges naturally from a quenched $\pm J$ random bond Ising model with polarized boundaries, by applying a gauge transformation to push the disorder onto the leaves, leaving a ferromagnetic model in the bulk~\cite{chayes1986mean}. As noted in~\appref{app:TreeIsing}, the spin glass transition of the Ising model is \textit{continuous}, and \Tglass~is robust to the particular ensemble of boundary conditions (so long as $\alpha$ is finite). 

The story is much different for the first-order transition of our Tanner-Ising model. We analyze $\alpha = 0$ by numerically sampling~\autoref{eq:alpha-recursion}. We numerically find that the only fixed point of this recursion relation, including at zero temperature, is the paramagnetic fixed point $Q(m) = \delta(m)$, corresponding to zero variance in the bulk. Thus, the $\alpha=0$ BC fails to detect any spin glass transition. This behavior is a cautionary example of the argument made in Ref.~\cite{mezard2006reconstruction}: the i.i.d. BC is not a ``good'' model of a spin glass 
because it lacks the correlations that would emerge from a proper point-to-set construction on a locally tree-like closed graph. Indeed, $\alpha=0$ corresponds to sampling a Gibbs state on the full tree at \textit{infinite} temperature, and fails to capture the replica symmetry breaking expected of models with loops.

\subsection{Configurational entropy}\label{sec:complexity-tree}
As discussed in the main text, the configurational entropy is the difference of two terms, $s_{\beta}$ and $s_{\alpha=1}$. We measure the entropy density in bits, so that $s = \beta(\varepsilon - f)/\ln(2)$ where $\varepsilon, f$ are the energy density and free-energy density, respectively, at a given temperature.

Consider a Tanner-Ising model with a symmetrized local code $H_L^{(\mathrm{sym})}$ at each vertex, and no redundancies aside from those generated locally by the symmetrization. The partition function on graph $G$ of $n_v$ vertices is
\begin{equation}
    \mathcal{Z}(G, H_L^{(\mathrm{sym})}) = 2^k (1 + (2^{m_0}-1) y)^{n_v}
\end{equation}
where $m_0$ is the number of linearly independent checks in $H_L$, $k=n_v(s/2-m_0)$ is the number of logical bits, and $y=\exp(-2^{m_0-1} \beta)$. Therefore, the bulk free energy per spin is
\begin{align}\label{eq:f-para}
   \langle f \rangle_\beta &= -\frac{1}{\beta n} \ln\left[\mathcal{Z}(G, H_L^{(\mathrm{sym})})\right] \notag \\
   &= r\ln 2+ \frac{2}{s} \ln[1 + (2^{m_0}-1)y]
\end{align}
where $r$ is the global code rate. 

Meanwhile, the bulk energy per spin is 
\begin{equation}
    \expval{\varepsilon}_\beta  =\frac{\partial(\beta f)}{\partial \beta} = \frac{2^{m_0}}{s} \frac{(2^{m_0}-1)y}{1 + (2^{m_0}-1)y},
\end{equation}
in agreement with~\autoref{eq:E-G}. 

The same bulk energy density is obtained when $\alpha=1$ BCs are imposed. This is because $\alpha=1$ corresponds to ``free'' BCs in the sense that $p(\vec{\sigma},\vec{\sigma}_\partial) \propto e^{-\beta E(\vec{\sigma},\vec{\sigma}_\partial)}$, so linear averages of local quantities (such as the energy around a vertex) match the paramagnet. On the other hand, below \Tglass, imposing $\alpha=1$ BCs modifies the bulk free-energy density, as the bulk of the tree is confined to a single extremal component.\footnote{Note that the thermodynamic relation $\frac{\partial (\beta f)}{\partial \beta} = e$ is no longer satisfied with these boundary conditions, because the derivative does not commute with the average over BCs.} 

As demonstrated in forthcoming work~\cite{Sommers2025}, the free-energy density can be recovered from the fixed point distribution $Q^*$ as
\begin{align}\label{eq:f1}
    \beta &f_{\alpha=1}(T) = -\frac{2}{z(\vec{0})\gdeg} \int z(\vec{m}) \ln(z(\vec{m})) \prod_{i=2}^\gdeg dQ^*(m_i) \notag \\
    &+ \frac{2(\gdeg-2)}{\gdeg} \int \frac{1 + m_1 m_2}{2} \ln \left(\frac{1+m_1m_2}{2}\right) \prod_{i=1}^2 dQ^*(m_i).
\end{align}

When the local code is the symmetrized [7,4,3] Hamming code, substituting $\gdeg=7, m_0=3,r=1/7,y=\exp(-4\beta)$ into~\autoref{eq:f-para} yields
\begin{align}
\langle f \rangle_\beta = \frac{1}{7} \ln(2) + \frac{2}{y} \ln(1 + 7y).
\end{align}
We can also obtain a closed-form expression for the fixed point distribution $Q^*$, to leading order in $y$. Substituting this expression into~\autoref{eq:f1} yields
\begin{equation}
\beta f_{\alpha=1}(y) = 2y \ln(2y) + O(y^2).
\end{equation}
We therefore arrive at (cf.~\autoref{eq:s-conf-series})
\begin{align}
s_{\mathrm{conf}}(y) &= [\beta \langle f\rangle_\beta - \beta f_{\alpha=1}(y)]/\ln(2) \notag \\
&= \frac{1}{7} + 2 y (1 - \ln(2y))/\ln(2).
\end{align}
Together with $\expval{\varepsilon}_\beta = 7y/(1+7y)$, this yields the gray curve in~\autoref{fig:complexity-tree}. The numerical curve was obtained by running the population dynamics until it approximately converged, aggregating the two populations, and resampling to estimate the two contributions to~\autoref{eq:f1}.

\end{document}

%% file: preamble.tex
\usepackage[utf8]{inputenc}
\usepackage{amsmath}
\usepackage{amsthm}
\usepackage{enumitem}
\usepackage{amssymb}
\usepackage{wasysym}
\usepackage{oplotsymbl}
\usepackage{bm} % bold italic vectors
\usepackage{dsfont} % identity operator
\usepackage{mathrsfs} % mathscr font
\usepackage{physics}
\usepackage{mathtools}

\usepackage{xcolor}
\usepackage{framed}
\definecolor{darkblue}{rgb}{0,0,.65}
\definecolor{darkgreen}{rgb}{0.28,0.41,0.19}
\definecolor{nicegreen}{rgb}{0.28,0.85,0.19}

\usepackage[colorlinks=true,linkcolor=blue,allcolors=blue]{hyperref}
\usepackage[capitalize,nameinlink]{cleveref}

\def\equationautorefname~#1\null{Eq. (#1)\null}

\newcommand{\appref}[1]{\hyperref[#1]{App.~\ref*{#1}}}

\renewcommand\vec{\bm}

\DeclareMathOperator\linspan{span}
\DeclareMathOperator\dist{dist}

\DeclarePairedDelimiter\floor{\lfloor}{\rfloor}

\newcommand{\id}{\ensuremath{\mathds{1}}}

\newcommand{\pG}{\ensuremath{p_{\rm G}}}

\newcommand{\Sconf}{\ensuremath{S_{\rm conf}}}
\newcommand{\sconf}{\ensuremath{s_{\rm conf}}}
\newcommand{\Hcheck}{\ensuremath{H}}
\newcommand{\hamil}{\ensuremath{\mathcal H}}

\newcommand{\indicator}{\ensuremath{\id}}

\newcommand{\wbit}{\ensuremath{w_{\rm bit}}}
\newcommand{\wcheck}{\ensuremath{w_{\rm check}}}

\newtheorem{prototheorem}{Theorem}[section]
\newtheorem{protodef}[prototheorem]{Definition}

\newtheorem{protolemma}[prototheorem]{Lemma}

\newtheorem{protoconjecture}[prototheorem]{Conjecture}

\newenvironment{theorem}{\colorlet{shadecolor}{gray!15}\begin{shaded}\begin{prototheorem}}
{\end{prototheorem}\end{shaded}}

\newenvironment{lemma}{\colorlet{shadecolor}{gray!15}\begin{shaded}\begin{protolemma}}
{\end{protolemma}\end{shaded}}

\newenvironment{definition}{\colorlet{shadecolor}{gray!15}\begin{shaded}\begin{protodef}}
{\end{protodef}\end{shaded}}